%% file: main.tex
\pdfoutput=1
\newif\iffull 
\newif\ifarxiv 

\makeatletter \@input{texdirectives.tex} \makeatother

\documentclass[acmsmall,review,screen,nonacm]{acmart}
\usepackage{xr}

\ifarxiv
\makeatletter
\def\@authorsaddresses{}
\fancypagestyle{firstpagestyle}{%
  \fancyhf{}%
  \renewcommand{\headrulewidth}{\z@}%
  \renewcommand{\footrulewidth}{\z@}%
}
\fancypagestyle{standardpagestyle}{%
  \fancyhf{}%
  \renewcommand{\headrulewidth}{\z@}%
  \renewcommand{\footrulewidth}{\z@}%
    \fancyhead[RO]{\@headfootfont\thepage}%
    \fancyhead[RE]{\@headfootfont\@shortauthors}%
    \fancyfoot[RO,LE]{}%
}
\def\@mkbibcitation{}
\makeatother
\renewcommand\footnotetextcopyrightpermission[1]{}
\fi

\input{agt-definitions.sty}
\input{definitions}

\newcommand{\myparagrapht}[1]{\noindent{\bf #1}}
\newcommand{\myparagraph}[1]{\vspace{0.5em}\myparagrapht{#1.}}
\DeclareDocumentCommand{\iffullv}{m O{}}{\iffull{#1}\else{#2}\fi}

\lstdefinelanguage{scala}{
    morekeywords={let,in,if,then,else,fun, unpack, pack, as, error},
    otherkeywords={=>,!,:=,+},
    sensitive=true,
    morecomment=[l]{//},
    morecomment=[n]{/*}{*/},
    morestring=[b]",
    morestring=[b]',
    morestring=[b]""",
    escapeinside={(*}{*)},
    moredelim=**[is][{\btHL}]{`}{`}
  }

\makeatletter
\def\smallunderbrace#1{\mathop{\vtop{\m@th\ialign{##\crcr
   $\hfil\displaystyle{#1}\hfil$\crcr
   \noalign{\kern3\p@\nointerlineskip}%
   \tiny\upbracefill\crcr\noalign{\kern3\p@}}}}\limits}
\makeatother

\renewcommand\footnotemark{}

\begin{document} 
\iffull
\title{Gradual System F (with Appendix)}
\else
\title{Gradual System F}
\fi
\titlenote{This work is partially funded by CONICYT/FONDECYT Regular/1190058, CONICYT/Doctorado Nacional/2015-21150510 \& 21151566, and by the ERC Starting Grant SECOMP (715753).
}

\author{Elizabeth Labrada}
\affiliation{%
  \institution{University of Chile}
  \department{Computer Science Department (DCC)}
  \streetaddress{Beauchef 851}
  \city{Santiago}
  \country{Chile}
}
\author{Mat{\'i}as Toro}
\affiliation{%
  \institution{University of Chile}
  \department{Computer Science Department (DCC)}
  \streetaddress{Beauchef 851}
  \city{Santiago}
  \country{Chile}
}
\author{\'Eric Tanter}
\affiliation{%
  \institution{University of Chile}
  \department{Computer Science Department (DCC)}
  \streetaddress{Beauchef 851}
  \city{Santiago}
  \country{Chile}
  }

\addtocontents{toc}{\protect\setcounter{tocdepth}{0}}

\begin{abstract}
Bringing the benefits of gradual typing to a language with parametric polymorphism like System F, while preserving relational parametricity, has proven extremely challenging: first attempts were formulated a decade ago, and several designs have been recently proposed, with varying syntax, behavior, and properties.
Starting from a detailed review of the challenges and tensions that affect the design of gradual parametric languages, this work presents an extensive account of the semantics and metatheory of \gsf, a gradual counterpart of \sysF . In doing so, we also report on the extent to which the Abstracting Gradual Typing methodology can help us derive such a language. Among gradual parametric languages that follow the syntax of \sysF, \gsf achieves a unique combination of properties. We clearly establish the benefits and limitations of the language, and discuss several extensions of \gsf towards a practical programming language.
\end{abstract}  

\ifarxiv\else
\fi

\maketitle

\section{Introduction}

There are many approaches to integrate static and dynamic type checking~\cite{,cartwrightFagan:pldi1991,abadiAl:toplas1991,
matthewsFindler:popl2007,tobinFelleisen:dls2006,biermanAl:ecoop2010}. In particular, gradual typing supports the smooth integration of static and dynamic type checking by introducing the notion of {\em imprecision} at the level of types, which induces a notion of {\em consistency} between plausibly equal types~\cite{siekTaha:sfp2006}. A gradual type checker does a best effort statically, treating imprecision optimistically. The runtime semantics of the gradual language detects at runtime any invalidation of optimistic static assumptions. Such detection is usually achieved by compilation to an internal language with explicit casts, called a cast calculus.
In addition to being type safe, a gradually-typed language is expected to satisfy a number of properties, in particular that it conservatively extends a corresponding statically-typed language, that it can faithfully embed dynamically-typed terms, and that the static-to-dynamic transition is smooth, a property formally captured as the (static and dynamic) gradual guarantees~\cite{siekAl:snapl2015}.

Since its early formulation in a simple functional language~\citep{siekTaha:sfp2006}, gradual typing has been explored in a number of increasingly challenging settings such as subtyping~\citep{siekTaha:ecoop2007,garciaAl:popl2016},
references~\citep{hermanAl:hosc10,siekAl:esop2015},
 effects~\citep{banadosAl:icfp2014,banadosAl:jfp2016}, ownership~\citep{sergeyClarke:esop2012}, typestates~\citep{wolffAl:ecoop2011,garciaAl:toplas2014}, information-flow typing~\citep{disneyFlanagan:stop2011,fennellThiemann:csf2013,toroAl:toplas2018}, session types~\citep{igarashiAl:icfp2017b}, refinements~\citep{lehmannTanter:popl2017}, set-theoretic types~\citep{castagnaLanvin:icfp2017}, Hoare logic~\citep{baderAl:vmcai2018} and parametric polymorphism~\citep{ahmedAl:popl2011,ahmedAl:icfp2017,ina:oopsla2011,igarashiAl:icfp2017,xieAl:esop2018}.

In the case of parametric polymorphism, a long-standing challenge has been to prove that the gradual language preserves a rich semantic property known as {\em relational parametricity}~\citep{reynolds:83}, which dictates that a polymorphic value must behave uniformly for all possible
instantiations. The first gradual language to come with a proof of parametricity is the cast calculus \lamB~\citep{ahmedAl:icfp2017}, which extends the archetypal polymorphic lambda calculus also known as \sysF with casts and a dynamic (or unknown) type. \lamB is
also used as a target language by~\citet{xieAl:esop2018}, who explore the treatment of {\em implicit} polymorphism. Another recent effort is \sysFg, an actual gradual source language that is compiled to a cast calculus akin to \lamB, called \sysFc~\cite{igarashiAl:icfp2017}. These efforts highlighted the many challenges involved in design a gradual parametric languages that satisfies both parametricity and the expected properties of gradually-typed languages~\cite{siekAl:snapl2015}. Inspired by these challenges, \cite{newAl:popl2020} propose a fairly different design, which does not start from \sysF, but instead consider a language with explicit terms for sealing and unsealing values. This change of perspective allows their language \polyG to satisfy all the expected theorems, but at the cost of a fairly different programming model, with its own limitations.

\myparagraph{Contributions} This work starts from the identification of several design issues in existing gradual languages, and studies a gradual parametric language in the style of \sysF by applying a general methodology to gradualize programming languages~\cite{garciaAl:popl2016}. The resulting language, called \gsf (for Gradual \sysF), 
embodies a number of important design choices. The first is in its name: it is an extension of \sysF, and therefore sticks to the traditional syntax of the polymorphic lambda calculus, where terms need not bother with sealing explicitly as in \polyG. \gsf also pays attention to respect type instantiations on imprecise types in order to adequately extend \sysF towards \sysFw, in the same way that the gradually-typed lambda calculus~\cite{siekTaha:sfp2006} extends the simply-typed lambda calculus towards \sysF.
This characteristic alone differentiates \gsf from other gradual parametric languages. 

We introduce gradual parametricity informally and discuss its challenges by reviewing closely related work (\S\ref{sec:need}), and present a quick tour of \gsf, including its design principles and main properties (\S\ref{sec:gsf-inform}).
We then explain how we derive \gsf from a variant of \sysF called \SPFL (\S\ref{sec:spfl-lang}), by following the Abstracting Gradual Typing methodology (AGT)~\cite{garciaAl:popl2016}. While the statics of \gsf follow naturally from \SPFL using AGT (\S\ref{sec:gsf-statics}), the dynamic semantics are more challenging (\S\ref{sec:gsf-dynamics}/\S\ref{sec:evidence}). In particular, satisfying parametricity forces us to sacrifice one of the desirable properties of gradual languages, known as the dynamic gradual guarantee~\cite{siekAl:snapl2015}. We study this tension in details and expose a weaker form of this gradual guarantee that \gsf does satisfy (\S\ref{sec:gsf-dgg}). While weaker, this guarantee is stronger that what other gradual languages based on \sysF achieve. 
We then review the notions of gradual parametricity from the literature and present the gradual parametricity that \gsf satisfies, along with gradual free theorems (\S\ref{sec:gsf-param}). We then study the dynamic end of the gradual typing spectrum supported by \gsf, and show that beyond a standard dynamically-typed language, \gsf can faithfully embed a language with dynamic sealing primitives~\cite{sumiiPierce:popl2004} (\S\ref{sec:embedding}). Finally, we study two extensions of \gsf. The first explores a novel technique to circumvent the limitations of the explicit treatment of polymorphism in \gsf, in order to recover the ability to fully deal with typed-untyped interoperability (\S\ref{sec:dynamic-implicit-polymorphism}). The second studies the addition of existential types to \gsf, which are at the core of data abstraction (\S\ref{sec:existentials}).

\myparagraph{Prior publication}
This article substantially revises and extends a prior conference publication~\cite{toroAl:popl2019}. First of all, the presentation of related approaches (\S\ref{sec:need}) has been largely revised, with novel illustrations and analysis, and now also includes the recent work on \polyG~\cite{newAl:popl2020}. This revised version presents several novel technical contributions: the detailed analysis of the dynamic gradual guarantee violation is new, and so is the development of the weaker guarantee that \gsf satisfies (\S\ref{sec:gsf-dgg}). The presentation of gradual parametricity (\S\ref{sec:gsf-param})
includes a new comparison between the approaches of \citet{ahmedAl:icfp2017} and \citet{newAl:popl2020}, shedding light on the current design space. The results that follow from both parametricity and the weaker guarantee subsume those presented in the earlier publication. Finally, the last three sections---embedding dynamic sealing, dynamic implicit polymorphism, and existential types---are completely novel developments.

\myparagraph{Supplementary material}
Auxiliary definitions and proofs of the main results can be found in Appendices.
Additionally, an interactive prototype of \gsf, which exhibits both typing derivations and reduction traces, is available online: \url{https://pleiad.cl/gsf}.
All the examples mentioned in this paper, as well as others, are readily available in the online demo.

\section{Gradual Parametricity: Motivation and Challenges}
\label{sec:need}

We start with a quick introduction to parametric polymorphism and parametricity, before motivating gradual parametricity through examples and finally exposing different challenges in the design gradually parametric languages.

\subsection{Background: Parametric Polymorphism}
\label{sec:back:param}

Parametric polymorphism allows the definition of terms that can operate over any type, with the introduction of type variables and universally-quantified types. For instance, a function of type $\forall X. X -> X$ can be used at any type, and returns a value of the same type as its actual argument.
For the sake of this work, it is important to recall two crucial distinctions that apply to languages with parametric polymorphism, one syntactic---whether polymorphism is explicit or implicit---and one semantic---whether polymorphic types impose strong behavioral guarantees or not.

\myparagraph{Explicit vs Implicit}
In a language with {\em explicit} polymorphism, such as the Girard-Reynolds polymorphic lambda calculus (\emph{a.k.a.}~\sysF)~\citep{girard,reynolds:ps1974}, the term language includes explicit type abstraction $\Lambda X.e$ and explicit type application $e\;[T]$, as illustrated next:
\begin{lstlisting}[numbers=none]
let f : (*$\forall X.X -> X$*) = (*$\Lambda$*)X.(*$\lambda$*)x:X.x in f [Int] 10
\end{lstlisting}
The function \lstinline{f} has the polymorphic (or universal) type $\forall X.X -> X$. By applying \lstinline{f} to type $\Int$ (we also say that \lstinline{f} is {\em instantiated} to $\Int$), the resulting function has type $\Int -> \Int$; it is then passed the number \lstinline{10}. Hence the program evaluates to \lstinline{10}.

In contrast to this intrinsic, Church-style formulation, the Curry-style presentation of {\em polymorphic type assignment}~\cite{curry1972} does not require type abstraction and type application to be reflected in terms. This approach, known as {\em implicit} polymorphism, has inspired many languages such as ML and Haskell.
Technically, implicit polymorphism induces a notion of subtyping that relates polymorphic types to their instantiations~\citep{mitchell:ic1988,oderskyLaufer:popl1996}; \eg~$\forall X. X -> X <: \Int -> \Int$.
Implicitly-polymorphic languages generally use an explicitly-polymorphic language underneath (\eg~GHC Core), providing the convenience of implicitness through an inference phase that produces an explicitly-annotated program. In essence, the use of the subtyping judgment $\forall X. X -> X <: \Int -> \Int$ is materialized in terms by introducing an explicit instantiation $[\Int]$, and vice-versa, the use of the judgment
$\Int -> \Int <: \forall X. \Int -> \Int$ is materialized by inserting a type abstraction constructor $\Lambda X$. 

\myparagraph{Genericity vs. Parametricity}
Some languages with universal type quantification also support intensional type analysis or reflection, which allows a function to behave differently depending on the type to which it is instantiated. 
For instance, in Java, a generic method of type $\forall X. X -> X$
can use \lstinline{instanceof} to discriminate the actual type of the argument, and behave differently for \lstinline{String}, say, than for \lstinline{Integer}. 
Therefore these languages only support {\em genericity}, \ie~the fact that a value of a universal type can be safely instantiated at any type.\footnote{We call this property {\em genericity}, by analogy to the name {\em generics} in use in object-oriented languages like Java and C\#.}

Parametricity is a much stronger interpretation of universal types, which dictates that a polymorphic value {\em must behave uniformly} for all possible
instantiations~\citep{reynolds:83}. This implies that one can derive interesting theorems about the behavior of a program by just looking at its type, hence the name ``free theorems'' coined by~\citet{wadler:fpca89}.
For instance, one can prove using parametricity that any polymorphic list permutation function commutes with the polymorphic map function.
Technically, parametricity is expressed in terms of a (type-indexed) {\em logical relation} that denotes when two terms behave similarly when viewed at a given type. All well-typed terms of \sysF are related to themselves in this logical relation, meaning in particular that all polymorphic terms behave uniformly at all instantiations~\citep{reynolds:83}.

Simply put, if a value \lstinline{f} has type
$\forall X. X -> X$, genericity only tells us that
\lstinline{f [Int] 10} reduces to {\em some} integer, while
parametricity tells the much stronger result that \lstinline{f [Int] 10} necessarily evaluates to \lstinline{10} (\ie~\lstinline{f} has to be the identity function). 
In the context of gradual typing, \citet{ina:oopsla2011} have explored genericity with a gradual variant of Java. All other work has focused on the challenge of enforcing parametricity~\citep{ahmedAl:popl2011,ahmedAl:icfp2017,igarashiAl:icfp2017,xieAl:esop2018,newAl:popl2020}.

\subsection{Gradual Parametricity in a Nutshell}
\label{sec:gp-nut}

Gradual parametricity supports imprecise typing information, yet ensures that assumptions about parametricity are enforced at runtime whenever they are not provable statically. 
In the following program, function \lstinline{f} is given the polymorphic type $\forall X. X -> X$, and is therefore expected to behave parametrically. It is then instantiated at type $\Int$, and applied to the value $10$.

\begin{lstlisting}[numbers=none]
let g: ? = (*$\lceil$*)(*$\lambda$*)a.(*$\lambda$*)b.if b then a else a + 1(*$\rceil$*) in
let f: (*$\forall$*)X.X(*$->$*)X = (*$\Lambda$*)X.(*$\lambda$*)x:X.g x (*$\framebox{v}$*) in
f [Int] 10
\end{lstlisting}
Note that \lstinline{f} is implemented using a function \lstinline{g} of unknown type, which is the result of embedding (\lstinline[mathescape]{$\lceil\cdot\rceil$}) untyped code into the gradual language.
Because of the type consistency introduced by the imprecision of declared types, this program is well-typed; however the compliance of \lstinline{f} with respect to its declared parametric behavior is unknown statically. 

In particular, by parametricity,  we can deduce that \lstinline{f} should behave like the identity function (\S\ref{sec:back:param}), i.e. the program should reduce to $10$. But \lstinline{g} itself behaves as an identity function only if its second argument is \lstinline{true}. Therefore, \lstinline{f} behaving as the identity function actually  depends on whether the second argument passed to \lstinline{g}, \lstinline[mathescape]{$\framebox{v}$} above, is \lstinline{true} or \lstinline{false}.
If \lstinline[mathescape]{$\framebox{v}$} is \lstinline{true}, \lstinline{f} behaves as the identity function, and the program above successfully reduces to $10$.
Conversely, if \lstinline[mathescape]{$\framebox{v}$} is \lstinline{false} then the program could safely reduce to $11$; however, this would be a violation of parametricity because \lstinline{f} would not be behaving as the identity function. Therefore a runtime error is raised.

Therefore, as a consequence of gradual parametricity, by only looking at the type of \lstinline{f} and its use, we can prove that if the program above terminates, it should either produce 10, or fail with a runtime error, possibly denoting that \lstinline{f} did not behave like an identity function.

This simple example highlights two crucial characteristics of gradual parametricity. First, to enforce parametricity gradually requires more than tracking type safety. If we let the program reduce to $11$ when \lstinline[mathescape]{$\framebox{v}$} is \lstinline{false}, then type {\em safety} is not endangered; only type {\em soundness} (\ie~parametricity) is. Second, in presence of gradual types, the conclusions of free theorems from parametricity should be relaxed to admit both runtime type errors, and non-termination (which follows from gradual simple types~\citep{siekTaha:sfp2006}).

\subsection{Challenges of Gradual Parametricity}
\label{sec:state-of-the-art}
While the basics of gradual parametricity illustrated above are well understood and uncontroversial, the devil is in the details: we can find very different approaches to gradual parametricity in the literature, which reflect the many challenges and tensions involved in the design space of gradual parametric languages. 
These tensions arise from the different desirable metatheoretical properties, both from a parametricity point of view and from a gradual typing point of view, the tensions and interactions between them, as well as from some design decisions that emerge from these interactions. 
As we will see, even the notion of parametricity in a gradual context raises questions. 

Then, beyond parametricity and type safety, there are several important properties that are relevant for gradual languages~\cite{siekAl:snapl2015}, most notably the conservative extension of the static discipline, the gradual guarantees, and the embedding of a dynamically-typed counterpart. The former states that, on fully static programs, a gradual language should behave exactly like its static counterpart. The static and dynamic gradual guarantees state that making types less precise does not introduce new static or dynamic type errors, respectively. The latter characterizes the flexibility and potential expressiveness gain of gradual types over the static discipline. 

Here we focus on these metatheoretical results and their treatment in the literature, and discuss other design considerations in 
\S\ref{sec:design}. The aim of these sections is to informally expose the subtleties and tensions involved in the design of a gradual parametric language, by reporting on how existing languages differ in their respective choices.

\myparagraph{Parametricity}
Establishing that a gradual parametric language enforces parametricity has been a long-standing open issue: early work on the polymorphic blame calculus did not prove parametricity~\citep{ahmedAl:stop2009,ahmedAl:popl2011}, and the first parametricity result was established several years later for a variant of that calculus, \lamB~\citep{ahmedAl:icfp2017}. In fact, \lamB is a cast calculus, not a gradual source language, meaning that the program written above would not be valid; explicit casts should be sprinkled in different places to achieve the same result. \citet{igarashiAl:icfp2017} developed a gradual source language, \sysFg, which does support the intended lightweight, cast-free syntax of gradual languages. Following the early tradition of gradual typing~\citep{siekTaha:sfp2006}, the semantics of \sysFg are given by translation to a cast calculus, \sysFc, which is a close cousin of \lamB. Igarashi \etal in fact do not prove parametricity, but conjecture that due to the similarity between \sysFc and \lamB, parametricity should hold. \citet{xieAl:esop2018} develop a language (here referred to as \csa) with implicit polymorphism, which compiles to \lamB and therefore inherits its parametricity result. More recently, 
\citet{newAl:popl2020} explore a completely different point in the design space of gradual parametricity by considering a source language \polyG that departs from the syntax of \sysF in requiring explicit sealing and unsealing terms. The advantage of that approach is that certain issues observed in prior work can be sidestepped; in particular the obtained {\em notion} of gradual parametricity is stronger and more faithful to the original presentation of Reynolds than that of prior work. We come back to the discussion of different notions of gradual parametricity, which gets fairly technical, in \S\ref{sec:gsf-param}.

\myparagraph{Conservation Extension of \sysF} The conservative extension of a static discipline means that a gradual language should coincide with such static discipline on fully-precise terms, \ie~terms that do not use the unknown type at all. Of course, for this property to be meaningful, one must explicit what the ``static language'' is. Most work (\lamB, \sysFg, \csa, as well as the present work) consider \sysF as the starting point, meaning that \sysF programs should be valid programs in these gradual languages, and behave as they would in \sysF. 

On the one hand, \sysFg is a conservative extension of \sysF, and \csa of an implicit variant of \sysF. 
On the other hand, and despite its similar syntax for fully-typed terms, \lamB is not a conservative extension of \sysF.
Consider the type of a polymorphic identity function, $\forall X. X -> X$. 
In \lamB this type is not only compatible with $\Int -> \Int$ (which is a defining feature of {\em implicit} polymorphism, despite \lamB being explicitly polymorphic), but it is also compatible with $\Int -> \Bool$. Likewise, $\forall X. X -> X$ is compatible with $\forall X. \Int -> \Bool$. Therefore the type system of \lamB disagrees with that of \sysF on some fully static terms. 

The case of \polyG is different: as alluded to above, in order to better accommodate some desired metatheoretical results, \polyG departs from the syntax of \sysF. 
For instance, the following \sysF program defines a function \lstinline|f|, which is the identity function, and instantiates it at type $\Int$, applies it to $1$, and then adds $1$ to the result, yielding $2$.

\begin{lstlisting}[numbers=none]
let f : (*$\forall X.X -> X$*) = (*$\Lambda$*)X.(*$\lambda$*)x:X.x in (f [Int] 1) + 1
\end{lstlisting}
This program is rejected statically in \polyG, because the sealing and unsealing that implicitly underlies polymorphic behavior in \sysF must
happen {\em explicitly} in the syntax of terms:
\begin{lstlisting}[numbers=none]
let f : (*$\forall X.X -> X$*) = (*$\Lambda$*)X.(*$\lambda$*)x:X.x in unseal(*$_\mathtt{X}$*)(f [X=Int] (seal(*$_\mathtt{X}$*) 1)) + 1
\end{lstlisting}
Note that the syntax of \polyG forces an outward scoping of type variables, \ie~\lstinline{[X=Int]} above puts \lstinline{X} in scope for subsequent use by \lstinline[mathescape]{seal$_\mathtt{X}$} and 
\lstinline[mathescape]{unseal$_\mathtt{X}$}. \polyG is certainly a conservative extension of that special static source language with explicit sealing.

\myparagraph{Gradual Guarantees}
The gradual guarantees (both static and dynamic) were introduced by \citet{siekAl:snapl2015} in order to formally capture the expectations of programmers using gradual languages: namely that a loss of precision should be harmless. The major tension faced by gradual parametric languages in the past decade has been to attempt to reconcile these guarantees with parametricity.

While this article will dive into this question repeatedly and at a quite technical level, let us present here what happens on the surface, for a programmer.
Consider this program, which is the same as above, except that the return type of the function \lstinline{f} is now unknown:
\begin{lstlisting}[numbers=none]
let f : (*$\forall X.X -> \?$*) = (*$\Lambda$*)X.(*$\lambda$*)x:X.x in (f [Int] 1) + 1
\end{lstlisting}
Following the motto that imprecision is harmless, a programmer might expect this program to both typecheck and run without errors, yielding $2$.
However, in \lamB and \sysFc, the above program\footnote{Rather, its elaboration with explicit casts, as mandated by \lamB and \sysFc, which are cast calculus and not source languages.} fails with a runtime error, because the result of \texttt{f [Int] 1} is {\em sealed}, and therefore unusable directly.
\citet{ahmedAl:popl2011} justify this behavior (already present in early work~\citep{ahmedAl:stop2009}), or the alternative of always failing before returning, based on a claim about gradual free theorems, framed as a consequence of parametricity. This can be quite surprising because the underlying value is the \sysF identity function, which {\em does} behave parametrically.

This failure behavior is a violation of the dynamic gradual guarantee (DGG). But in fact, technically, it only is a violation if we consider the program above to actually be a more imprecise variant of the \sysF program presented previously (where the return type of \lstinline{f} is \lstinline{X} instead of $\?$). While this change of what precision {\em is} might not make intuitive sense to programmers (removing static type information yields a program that is {\em by definition} less precise), it can make sense when working on the metatheory. That is the approach followed by \sysFg. Specifically, \sysFg does not allow losses of precision in {\em parametric} positions of a polymorphic type. For instance, $\forall X. X -> \Int$ is considered more precise than $\forall X. X -> \?$, but 
$\forall X. X -> X$ is not. Because precision induces consistency, it means that $\forall X. X -> X$ and $\forall X. X -> \?$ are inconsistent with each other.

This means that the program above does not typecheck in \sysFg. So the restriction on precision imposed breaks the intuition of programmers that, starting program from a well-typed program, removing static type information yields a program that is {\em by definition} less precise---and should also be well-typed.  \citet{igarashiAl:icfp2017} prove the static gradual guarantee for \sysFg based on this 
restricted notion of precision, and leave the dynamic guarantee as a conjecture, so it is unclear whether the restrictions imposed are indeed sufficient.

\polyG was designed to address the tension between parametricity and the DGG that manifested in prior work, and does so by using a syntax with explicit sealing and unsealing, as introduced previously.
If we start with the following fully-static program:
\begin{lstlisting}[numbers=none]
let f : (*$\forall X.X -> X$*) = (*$\Lambda$*)X.(*$\lambda$*)x:X.x in unseal(*$_\mathtt{X}$*)(f [X=Int] (seal(*$_\mathtt{X}$*) 1)) + 1
\end{lstlisting}
Then making the return type of \lstinline{f} unknown yields a program that still typechecks and runs successfully. This is clear because the sealing behavior that was causing problem is now explicit in the terms, and therefore not affected by a loss of precision in types. However, this choice of syntax is not innocuous. Consider the following imprecise program, where the body of \lstinline{f} has been elided:
\begin{lstlisting}[numbers=none]
let f : (*$\forall X.X -> \?$*) = (*$\framebox{body}$*) in unseal(*$_\mathtt{X}$*)(f [Int] (seal(*$_\mathtt{X}$*) 1)) + 1
\end{lstlisting}
As we said, if \lstinline[mathescape]{$\framebox{body}$ is $\Lambda$X.$\lambda$x:X.x}, then the program evaluates to $2$ as expected.
However, if \lstinline{f} is a constant function, 
\eg~\lstinline[mathescape]{$\framebox{body}$} is 
\lstinline[mathescape]{$\Lambda$X.$\lambda$x:X.1}, then this \polyG program fails because the call-site unsealing of the value returned by \lstinline{f} is now invalid. If one removes \lstinline[mathescape]{unseal$_\mathtt{X}$} around the application of 
\lstinline{f}, the program behaves properly. But now, the case where 
\lstinline[mathescape]{$\framebox{body}$} 
is the polymorphic identity function fails, when trying to add $1$ to a sealed value.

This means that in \polyG, the decision to use unsealing or not at a call site cannot be done {\em modularly}: one needs to know the implementation of \lstinline{f} to decide. This limitation of \polyG is precisely what makes it enjoy a very complete metatheory, reconciling parametricity and the DGG. Said differently, the objective of sticking to \sysF syntax where sealing and unsealing does not manifest at the term level is precisely what makes other gradual parametric languages suffer in some aspects of the metatheoretical wish list.

\myparagraph{Embedding of a Dynamic Language}
A gradual language is expected to cover a spectrum between two typing disciplines, such as simple static typing and dynamic typing. 
The static end of the spectrum is characterized by the conservative extension result mentioned previously. The dynamic end of the spectrum is usually captured by an embedding from a given dynamic language to the gradual language~\cite{siekAl:snapl2015}. For instance, in the case of GTLC~\cite{siekTaha:sfp2006}, the dynamic language is an untyped lambda calculus with primitives.

This aspect of the design space has not received any attention in the literature on gradual parametricity so far. Of course, because a polymorphic language includes a simply-typed one at its core, one naturally expects an untyped lambda calculus to be embeddable in a gradual polymorphic language. As we will see in \S\ref{sec:embedding}, the dynamic end of the spectrum for a gradual parametric language can be even more interesting.

\subsection{Additional Challenges and Design Considerations}
\label{sec:design}

We now look at other design considerations that are relevant for gradual parametric languages.

\myparagraph{Polymorphism and Interoperability}
\lamB, \sysFg, and \polyG are languages with {\em explicit} polymorphism, \ie~with explicit type abstraction and type application terms. Despite this, \lamB and \sysFg accommodate some form of implicit polymorphism, with different flavors. The underlying motivation is to address the issue of {\em interoperability} between typed and untyped code. The archetypal example being the \sysF polymorphic identity function, that one would like to be able to use in untyped code as an $\? -> \?$ function, or vice versa, using the untyped identity function as a polymorphic one.

In \lamB~\cite{ahmedAl:icfp2017}, two type compatibility rules to support this kind of implicit polymorphism:
\begin{mathpar}
  \inference[(Comp-AllR)]{\Sigma; \Delta, X |- T_1 <: T_2 & X \not\in T_1}
  {\Sigma; \Delta |- T_1 <: \forall X. T_2}
  \and
  \inference[(Comp-AllL)]{\Sigma; \Delta |- T_1[\? / X] <: T_2}
  {\Sigma; \Delta |- \forall X. T_1 <: T_2}
\end{mathpar}
Clearly, these rules permit $\forall X. X -> X$ to be compatible with $\? -> \?$. But, as first identified by \citet{xieAl:esop2018} and recalled above, the problem with these rules is that they also break the conservative extension of \sysF: the type $\forall X. X -> X$ is compatible with both $\forall X. \Int -> \Bool$ and $\Int -> \Bool$.

As an explicitly polymorphic language, \sysFg does not relate $\forall X. X -> X$ with any of its static instantiations. However, it {\em does} relate that type with $\? -> \?$, considered to be {\em quasi-polymorphic}, on the basis that using the unknown type should bring some of the flexibility of implicit polymorphism.

\citet{xieAl:esop2018} argue that it is preferable to clearly separate the subtyping relation induced by implicit polymorphism from the consistency relation induced by gradual types. Their notion of consistent subtyping extends the notion of~\citet{siekTaha:ecoop2007}. As a result, \csa features intuitive and straightforward definitions of precision and consistency, while accommodating the flexibility of implicit polymorphism in full.
Note that the implicit polymorphism of \citet{xieAl:esop2018} faces other challenges, most notably the lack of coherence of the runtime semantics.

\myparagraph{Respecting Type Instantiations}
What should type instantiations on terms of unknown type mean? Below is a simple program in which the polymorphic identity function ends up instantiated to $\Int$ and passed a $\Bool$ value:
\begin{lstlisting}[numbers=none]
let g : ? = (*$\Lambda$*)X.(*$\lambda$*)x:X.x in g [Int] true
\end{lstlisting}
This program in \sysFg, and its adaptation to \lamB, both return \lstinline{true}, despite the explicit instantiation to $\Int$. Internally, this happens because \lstinline{g} is first consistently considered to be of type $\forall X.\?$ in order to accommodate the type instantiation, but then the instantiation yields a substitution of $\Int$ for $X$ in $\?$, which in both languages is just $\?$. There is no tracking of the decision to instantiate the underlying value to $\Int$. 

In contrast, the same program in \polyG:
\begin{lstlisting}[numbers=none]
let g : ? = (*$\Lambda$*)X.(*$\lambda$*)x:X.x in g [X=Int] (seal(*$_\mathtt{X}$*) true)
\end{lstlisting}
does not even typecheck, because sealing \lstinline{true} with \lstinline{X} requires the types to coincide. If we ascribe \lstinline{true} to the unknown type before sealing it, then the program typechecks but fails at runtime, thereby respecting the type instantiation to \lstinline{Int}.

\myparagraph{Expressiveness of Imprecision}
One of the major interest of gradual types is that they {\em soundly} augment the expressiveness of the original static type system. Let us illustrate first in a simply-typed setting (STLC refers to the simply-typed lambda calculus with base types), and how imprecision allows bridging the gap towards \sysF:

\begin{enumerate}
\item Consider the STLC term $t = \lambda x: \_. x$, which behaves as the identity function. $t$ is incomplete because the type annotation on $x$ is missing so far. 
\item $t$ is operationally valid at different types, but it cannot be given a general type in STLC. Its type has to be fixed at either $\Int -> \Int$, $\Bool -> \Bool$, etc. 
\item Intuitively, a proper characterization of $t$ requires going from simple types to parametric polymorphism, such as \sysF. In \sysF, we could use the type $\forall X. X -> X$ to precisely specify that $t$ can be applied with any argument type and return the same type.
\item With a gradual variant of STLC, we can give term $t$ the imprecise type $\? -> \?$ to statically capture the fact that $t$ is definitely a function, without committing to specific domain and codomain types. 
\item This lack of precision is soundly backed by runtime enforcement, such that the term $(t\; 3)\; 1$ evaluates to a runtime type error.
\end{enumerate}

\noindent We can unfold the exact same line of reasoning, starting from \sysF, and bridging the gap towards \sysFw:

\begin{enumerate}
\item Consider the \sysF term $t = \lambda x: \_. (x\;[\Int])$, which behaves as an instantiation function to $\Int$. $t$ is incomplete because the type annotation on $x$ is missing so far. 
\item $t$ is operationally valid at different types, but cannot be given a general type in \sysF. Its type has to be fixed at either $(\forall X. X\!->\!X) \!->\! (\Int\!->\!\Int)$, $(\forall XY. X\!->\!Y\!->\!X) \!->\! (\forall Y.\Int\!->\!Y\!->\!\Int)$, etc. 
\item Intuitively, a proper characterization of $t$ requires going from \sysF to higher-order polymorphism, such as \sysFw. In \sysFw, we could use the type $\forall P. (\forall X. P\; X) -> (P\; \Int)$ to precisely specify that $t$ instantiates any polymorphic argument to $\Int$.
\item With a gradual variant of \sysF, we ought to be able to give term $t$ the imprecise type $(\forall X.\?) -> \?$ to statically capture the fact that $t$ is definitely a function that operates on a polymorphic argument, without committing to a specific domain scheme and codomain type. 
\item This lack of precision ought to be soundly backed by runtime enforcement, such that, given $id: \forall X. X->X$, the term $(t\; id)\; \mathtt{true}$ should evaluate to a runtime type error.
\end{enumerate}

The fact that \lamB and \sysFc do not respect type instantiations on imprecise types mean that in these systems, the $(t\; id)\; \mathtt{true}$ does not raise any error.\footnote{In \sysFc, $(t\; id)\; \mathtt{true}$ fails because $\forall X.\?$ is not deemed consistent with $\forall X.X->X$. Consequently, $t$ must be declared to take an argument of type $\?$ instead of $\forall X.\?$. The result is the same as in \lamB however: no runtime error is raised.}
Therefore, while these higher-order polymorphic patterns can be expressed, they are unsound.

Respecting type instantiations is however not sufficient to be able to accommodate the augmented expressiveness described above. Consider the example in \polyG:
\begin{lstlisting}[numbers=none]
let t : (*$(\forall X.\?) -> \?$*) = (*$\lambda$*)x:(*$(\forall X.?)$*). x [X=Int] in (t id) (seal(*$_\mathtt{X}$*) true)
\end{lstlisting}

This program does not typecheck in \polyG because the type variable \lstinline{X} used explicitly in the body of the function is {\em no longer in scope} at the use site to seal the value \lstinline{true}. Recall that \lstinline{[X=Int]} puts \lstinline{X} in scope for the rest of the {\em lexical} scope of the instantiation, but it does not cross function boundaries. So, in addition to the modularity issues presented in the previous section, the explicit (un)sealing mechanism of \polyG cannot accommodate higher-order patterns like the above, which requires abstracting over type applications.

\section{\gsf, Informally} 
\label{sec:gsf-inform}

This paper presents the design, semantics and metatheory of \gsf, a gradual counterpart of \sysF. Here, we briefly introduce the methodology and principles we follow to design 
\gsf, and briefly review its properties and examples of use.

\subsection{Design Methodology}
\label{sec:design-methodology}
In order to assist language designers in crafting new gradual languages, \citet{garciaAl:popl2016} proposed the Abstracting Gradual Typing methodology (AGT, for short). The promise of AGT is that, starting from a specification of the {\em meaning} of gradual types in terms of the set of possible static types they represent, one can systematically derive all relevant notions, including precision, consistent predicates (\eg~consistency and consistent subtyping), consistent functions (\eg~consistent meet and join), as well as a direct runtime semantics for gradual programs, obtained by reduction of gradual typing derivations augmented with evidence for consistent judgments.

The AGT methodology has so far proven effective to assist in the gradualization of a number of disciplines, including effects~\citep{banadosAl:icfp2014,banadosAl:jfp2016}, record subtyping~\citep{garciaAl:popl2016}, set-theoretic types~\citep{castagnaLanvin:icfp2017}, union types~\citep{toroTanter:sas2017}, refinement types~\citep{lehmannTanter:popl2017} and security types~\citep{toroAl:toplas2018}. The applicability of AGT to gradual parametricity is an open question repeatedly raised in the literature---see for instance the discussions of AGT by \citet{igarashiAl:icfp2017} and \citet{xieAl:esop2018}. 
Considering the variety of successful applications of AGT, and the complexity of designing a gradual parametric language, in this work we decide to adopt this methodology, and report on its effectiveness.

\subsection{Design Principles}
\label{sec:design-principles}
Considering the many concerns involved in developing a gradual language with parametric polymorphism, we should be very clear about the principles, goals and non-goals of a specific design. In designing \gsf, we respect the following design principles:

\begin{description}[leftmargin=0cm]

\item[\sysF syntax:] \gsf is meant to be a gradual version of \sysF, and as such, adopts its syntax of both terms and types. Types are only augmented with the unknown type $\?$ to introduce the imprecision that is at the core of gradual typing. In particular, this precludes the use of unconventional syntactic constructs like the explicit (un)sealing terms of \polyG. 

\item[Explicit polymorphism:] 
\gsf is a gradual counterpart to \sysF, and as such, is a fully {\em explicitly} polymorphic language: type abstraction and type application are part of the term language, reflected in types. \gsf gradualizes type information, not term structure. The impact on interoperability is addressed in a second phase (see below).

\item[Simple statics:] \gsf embodies the complexity of dynamically enforcing parametricity solely in its dynamic semantics; its static semantics is as straightforward as possible.

\item[Natural precision:] Precision is intended to capture the level of static typing information of a gradual type, with $\?$ as the most imprecise, and static types as the most precise~\cite{siekAl:snapl2015}. \gsf preserves this simple intuition.
\end{description}

\subsection{Challenges and Properties}
\label{sec:properties}

Regarding the challenges and properties discussed previously, here is where \gsf stands:

\begin{description}[leftmargin=0cm]
\item[Type safety:] \gsf is type safe, meaning all programs either evaluate to a value, halt with a runtime error, or diverge. Well-typed \gsf terms do not get stuck.

\item[Conservative extension:] \gsf is a conservative extension of \sysF: both languages  coincide in their static and dynamic semantics for  fully static programs.

\item[Parametricity:] \gsf enforces a notion of gradual parametricity (\S\ref{sec:gsf-param}), directly inspired by \lamB~\citep{ahmedAl:icfp2017}.

\item[Static gradual guarantee:] By virtue of the simple statics principle stated above, \gsf satisfies the static gradual guarantee, \ie~typeability is monotonic with respect to the natural notion of precision. 

\item[Dynamic gradual guarantee:] In order to satisfy all the principles above, 
\gsf does not satisfy the dynamic gradual guarantee (DGG) for the natural notion of precision. In addition to studying this tension, we show that \gsf satisfies a weaker DGG, which in particular implies that imprecise ascriptions are harmless (\S\ref{sec:gsf-dgg}). 

\item[Embedding of a Dynamic Language:] We show that a standard dynamically-typed language can be embedded in \gsf. While novel, this result is not particularly surprising. More interestingly, we prove that \gsf can embed a dynamically-typed language with runtime sealing primitives~\cite{sumiiPierce:popl2004} (\S\ref{sec:embedding}).

\item[Interoperability:] While the core of \gsf is explicitly polymorphic, we describe an extension of \gsf that introduces a form of implicit polymorphism {\em dynamically}, which allows \gsf to properly support typed-untyped interoperability scenarios without compromising on other aspects (\S\ref{sec:dynamic-implicit-polymorphism}). 

\item[Faithful instantiations:] \gsf enforces type instantiations of imprecise types.

\item[Expressive imprecision:] \gsf soundly supports imprecise higher-order polymorphic patterns, bridging the gap towards \sysFw.
\end{description}

\subsection{\gsf in Action}
\label{sec:gsf-action}

We now briefly illustrate \gsf in action with a number of examples that correspond the main properties of the language. Other illustrative examples are available with the online interactive prototype. The different sections of the rest of this paper also come back to such representative examples as needed.

First, \sysF programs are \gsf programs, and behave as expected:
\begin{lstlisting}[numbers=none]
let f : (*$\forall X.X -> X$*) = (*$\Lambda$*)X.(*$\lambda$*)x:X.x in (f [Int] 1) + 1  ----> 2
\end{lstlisting}

\gsf enforces gradual parametricity. Recall the example from \S\ref{sec:gp-nut}:
\begin{lstlisting}[numbers=none]
let g: ? = (*$\lambda$*)a:?.(*$\lambda$*)b:?.if b then a else a + 1 in
let f: (*$\forall$*)X.X(*$->$*)X = (*$\Lambda$*)X.(*$\lambda$*)x:X.g x (*$\framebox{v}$*) in
f [Int] 10
\end{lstlisting}
As expected, if \lstinline[mathescape]{$\framebox{v}$} is \lstinline{true}, the program reduces to \lstinline{10}, and if \lstinline[mathescape]{$\framebox{v}$} is \lstinline{false}, the program fails with a runtime error when the body of the function \lstinline{g}attempts to perform an addition, since this type-specific operation is a violation of parametricity.

In \gsf, the natural notion of precision is used for typing, meaning that the following program is a less precise version than the \sysF program given at the beginning of this section. Also, imprecise ascriptions on values are harmless:
\begin{lstlisting}[numbers=none]
let f : (*$\forall X.X -> \?$*) = (*$\Lambda$*)X.(*$\lambda$*)x:X.x in (f [Int] 1) + 1     ----> 2
\end{lstlisting}

\gsf enforces type instantiations even when applied to an imprecisely-typed value: 
\begin{lstlisting}[numbers=none]
let g : ? = (*$\Lambda$*)X.(*$\lambda$*)x:X.x in g [Int] true               ----> error
\end{lstlisting}

\gsf soundly augments the expressiveness of \sysF to higher-order polymorphic code:
\begin{lstlisting}[numbers=none]
let t : (*$(\forall X.\?) -> \?$*) = (*$\lambda$*)x:(*$(\forall X.?)$*).x [Int] in (t id) 1     ----> 1
let t : (*$(\forall X.\?) -> \?$*) = (*$\lambda$*)x:(*$(\forall X.?)$*).x [Int] in (t id) true  ----> error
\end{lstlisting}

In \gsf, we can exploit the underlying runtime sealing mechanism used to enforce gradual parametricity in order to emulate the runtime sealing primitives of the cryptographic lambda calculus~\cite{sumiiPierce:popl2004} (\S\ref{sec:embedding}). Indeed, we can define a pair of functions of type \mbox{$\forall X. X -> \? \times \? -> X$}. The first component of the pair is a function of type $X -> \?$, which intuitively justifies sealing the argument (of type $X$) at runtime, but not unsealing the returned value (of type $\?$). Dually, the type of the second component is $\? -> X$, which only justifies unsealing the returned value. The underlying mechanism ensures that the unsealing function only succeeds if its argument was sealed by the first function:

\begin{lstlisting}[numbers=none]
let p : (*$\forall X. X -> \? \times \? -> X$*) = (*$\Lambda$*)X.(*$\langle$*)(*$\lambda$*)x:X.x::?, (*$\lambda$*)x:?.x::X(*$\rangle$*) in
let su = p [?] in let seal = fst su in let unseal = snd su in
(unseal (seal 1)) + 1    ----> 2     
\end{lstlisting}
On the second line, \lstinline{p [?]} creates a fresh pair of functions with an underlying type name that acts as the runtime sealing key: therefore, \lstinline[mathescape]{seal} seals the value \lstinline{1}, and \lstinline[mathescape]{unseal} unseals it. 
The whole program reduces to $2$. Unsealing the sealed value with any other generated unsealing function, or attempting to add directly to the sealed value, yields a runtime error.

By following \sysF, \gsf only supports explicit polymorphism. This means that certain desirable interoperability scenarios are not supported. For instance, the following program fails at runtime:
\begin{lstlisting}[numbers=none]
  let g : ? = (*$\lambda$*)x:(*$(\forall. X. X -> X)$*).x [Int] 1
  let h : ? = (*$\lambda$*)x:?.x
  g h
\end{lstlisting}
The runtime error is raised when \lstinline|g| is applied to \lstinline|h|, because 
$\? -> \?$ (the ``underlying type'' of \lstinline|h|) is not consistent with the polymorphic function type $\forall X. X -> X$. We study (and have implemented) an extension of \gsf with a {\em dynamic} adaptation mechanism that accounts for implicit polymorphism (\S\ref{sec:dynamic-implicit-polymorphism}). In essence, in the scenario above, the runtime system automatically wraps a type abstraction around \lstinline|h| instead of failing at the application. A dual adaptation occurs for missing type applications. This extended \gsf smoothly supports interaction with untyped code.

Another extension of \gsf studied in this article is that of existential types (\S\ref{sec:existentials}), which are the foundation of data abstraction and information hiding~\cite{mitchellPlotkin:toplas1888}. 
As an example, 
consider a semaphore abstract datatype (ADT) with operations: \lstinline{bit} to create a semaphore, \lstinline{flip} to produce a semaphore in the inverted state, and \lstinline{read} to consult the state of the semaphore. This interface can be expressed with the existential type \lstinline{Sem} $\triangleq \exists X.\{X, X -> X, X -> X\}$. 
Consider the embedding of an untyped implementation of a semaphore, \lstinline{u}, which is then declared to have the static existential type \lstinline{Sem}, using the unknown type \? as the representation type:
\begin{lstlisting}[numbers=none]
let u : ? = (*$\lceil$*){bit = true,flip = (*$\lambda$*)x.not x,read = (*$\lambda$*)x.x}(*$\rceil$*) in
let t : Sem = pack(*$\langle$*)?, t(*$\rangle$*) as Sem in
unpack(*$\langle$*)X,x(*$\rangle$*) = t in x.read (x.flip (*$\framebox{v}$*)) 
\end{lstlisting}
If \lstinline[mathescape]{$\framebox{arg}$} is \lstinline{x.bit}, then the program runs properly without error, yields \lstinline{false}. But, if the programmer tries to violate type abstraction by passing a boolean value such as \lstinline{true} for \lstinline[mathescape]{$\framebox{arg}$}, a runtime error is raised. Gradual existential types 
accommodate a variety of scenarios, including both imprecise ADT signatures and implementations. As we will see, gradual parametricity also allows proving representation independence results between gradual ADTs.

\begin{figure}[!hp]
  \begin{small}
  \begin{displaymath}
    \begin{array}{rcll}
    \multicolumn{4}{c}{  
      x \in \Var, X \in \VarType, \alpha \in \TypeName \quad
      \sstore  \in \TypeName \finto \Type,
      \Delta \subset \VarType,
      \Gamma \in \Var \finto \Type
      }\\
      T & ::= & \basetype | T -> T | \forall X. T|\pairtype{T}{T}| X | \alpha   & \text{(types)}\\
      t & ::= & \const | \lambda x:T.t | \Lambda X. t | \pair{t}{t}| x |t :: T | \op{\vectorOp{t}} | t\;t |t\;[T] | \proj{i}[t] & \text{(terms)}\\
      v & ::= & \const | \lambda x:T.t | \Lambda X. t | \pair{v}{v} & \text{(values)}
    \end{array}   
  \end{displaymath}
 \end{small}
 \begin{small}
 \begin{flushleft}
  \framebox{$\EnvSS t : T$}~\textbf{Well-typed terms}
  \end{flushleft}
  \begin{mathpar}
    \inference[(T$\const$)]{  \ftype(\const) = \basetype & \sswellGamma}{\EnvSS \const : \basetype}\and 
    \inference[(T$\lambda$)]{\EnvSS[\sstore][\Delta][\Gamma, x : T] t : T'
    }{\EnvSS \lambda x:T.t : T -> T'}\and
    \inference[(T$\Lambda$)]{\EnvSS[\sstore][\Delta, X][\Gamma]t : T & \sswellGamma  }{\EnvSS \Lambda X.t : \forall X.T}\and
     \inference[(Tpair)]{\EnvSS t_1 : T_1 & \EnvSS t_2 : T_2}{\EnvSS \pair{t_1}{t_2} : \pairtype{T_1}{T_2}}\and
      \inference[(Tx)]{x:T\in\Gamma & \sswellGamma}{\EnvSS x : T} \and
    \inference[(Tasc)]{\EnvSS t : T && \eqrules{T}{T'}}{\EnvSS t :: T' : T'} \and
 \inference[(Top)]{\EnvSS \vectorOp{t} : \vectorOp{T_1} & \ftype(op) = \vectorOp{T_2} -> T \\ \eqrules{\vectorOp{T_1}}{\vectorOp{T_2}}}{\EnvSS \op{\vectorOp{t}} : T}
     \and
    \inference[(Tapp)]{\EnvSS t_1 : T_1 & \EnvSS t_2 : T_2 \\ \eqrules{\dom(T_1)}{T_2}
    }{\EnvSS t_1\;t_2 :\cod(T_1 )}
     \and
    \inference[(TappT)]{ \EnvSS t : T& \Sigma; \Delta |- T' 
    }{\EnvSS t\; [T'] : \insta(T,T')}
    \and
    \inference[(Tpair$i$)]{ \EnvSS t : T }{\EnvSS \proj{i}[t] : \projt{i}{T}}
  \end{mathpar}
  \newline
  \begin{displaymath}
    \begin{block}
    \dom : \Type \rightharpoonup \Type \\
    \dom(T_1 -> T_2) = T_1\\
    \dom(T)\undefinedow
    \end{block}  
    \quad
    \begin{block}
    \cod : \Type \rightharpoonup \Type \\
    \cod(T_1 -> T_2) = T_2\\
    \cod(T)\undefinedow
    \end{block}    
    \quad
    \begin{block}
    {\insta : \Type^2 \rightharpoonup \Type} \\
    {\insta(\forall X. {T}, T') = T[T'/X]}\\
    {\insta(T, T')\undefinedow}
    \end{block} 
    \quad
    \begin{block}
    {\projts_i : \Type \rightharpoonup \Type} \\
    {\projt{i}{\pairtype{T_1}{T_2}} = T_i}\\
    {\projt{i}{T}\undefinedow}
    \end{block} 
  \end{displaymath}
\end{small}

\begin{small}
\begin{flushleft}
\framebox{$\eqrules{T}{T}$}~\textbf{Type equality}
\end{flushleft}
  \begin{mathpar} 
           \inference{|- \sstore}{\eqrules{B}{B}} \and
           \inference{|- \sstore & X \in \Delta}{\eqrules{X}{X}} \and
           \inference{\eqrules{T_1}{T_1'} & \eqrules{T_2}{T_2'}}{\eqrules{T_1 -> T_2}{T_1' -> T_2'}}\and
           \inference{\eqrules{T_1}{T_2}[\sstore][\Delta,X]}{\eqrules{\forall X. T_1}{\forall X. T_2}}\and
           \inference{\eqrules{T_1}{T_1'} & \eqrules{T_2}{T_2'}}{\eqrules{\pairtype{T_1}{T_2}}{\pairtype{T_1'} {T_2'}}}\\
           \inference{|- \sstore & \alpha \in \dom(\sstore)}{\eqrules{\alpha}{\alpha}} \and
           \inference{ \eqrules{\sstore(\alpha)}{T}}{\eqrules{\alpha}{T}}\and
           \inference{   \eqrules{T}{\sstore(\alpha)}}{\eqrules{T}{\alpha}}
    \end{mathpar}
\end{small}

\begin{small}
\begin{flushleft}
\framebox{$ \storeeval{t} \longrightarrow \storeeval t$}
~\textbf{Notion of reduction}
\end{flushleft}
\begin{mathpar}
\storeeval v :: T \longrightarrow \storeeval v
  \and
\storeeval \op{\vectorOp{v}} \longrightarrow \storeeval \redop{\vectorOp{v}}
  \and
  \storeeval (\lambda x : T. t) \;v \longrightarrow \storeeval t[v/x]
 \and 
\storeeval (\Lambda X. t)\; [T] \longrightarrow  \storeeval[\sstore, \alpha := T] t[\alpha/X] \quad\text{where }\alpha \not\in \dom(\sstore)
  \and 
\storeeval \proj{i}[\pair{v_1}{v_2}] \longrightarrow \storeeval v_i\\
\end{mathpar}
\begin{flushleft}
\framebox{$ \storeeval t \longmapsto \storeeval t$}
~\textbf{Evaluation frames and reduction}
\end{flushleft}
\begin{displaymath}
\begin{array}{rcll}
f & ::= & { [] }:: T | \op{\vectorOp{v}, [] , \vectorOp{t}} | [] \ t | v \ [] | [] \ [T] | \pair{[]}{t}| \pair{v}{[]} | \proj{i}[[]] & \text{(term frames)} 
\end{array}
\end{displaymath}
\begin{mathpar}
\inference[]{\storeeval t \longrightarrow \storeeval[\sstore'] t'
    }{ \storeeval t \longmapsto \storeeval[\sstore'] t'} \and
\inference[]{\storeeval t \longmapsto \storeeval[\sstore'] t'
    }{ \storeeval f[t] \longmapsto \storeeval[\sstore'] f[t']}
\end{mathpar}
\end{small}
 \caption{\SPFL: Simple Static Polymorphic Language with Runtime Type Generation}  
  \label{fig:spfl}
  \label{fig:spfl-syntax-statics}  
  \label{fig:spfl-dyn}
\end{figure}

\section{Preliminary: The static language \SPFL}
\label{sec:spfl-lang}

We systematically derive \gsf by applying AGT to a largely standard polymorphic language similar to \sysF, called \SPFL (Figure~\ref{fig:spfl}). In addition to the standard \sysF types and terms, \SPFL includes base types $B$ inhabited by constants $b$, typed using the auxiliary function $\ftype$, and primitive n-ary operations $\aop$ that operate on base types and are given meaning by the function $\redopn$. \SPFL also includes pairs $\pair{t_1}{t_2}$, and the associated projection operations $\proj{i}[t]$,\footnote{We omit the constraint $i \in \set{1,2}$ when operating on pairs throughout this paper.} as well as type ascriptions $t :: T$. 

The statics are standard. The typing judgment is defined over three contexts: a type name store $\sstore$ (explained below), a type variable set $\Delta$ that keeps track of type variables in scope, and a standard type environment $\Gamma$ that associates term variables to types. We adopt the convention of using partial type functions to denote computed types in the rules: $\dom$ and $\cod$ for domain and codomain types, $\insta$ for the resulting type of an instantiation, and $\projts_i$ for projected types. These partial functions are undefined if the argument is not of the appropriate shape. We also make the use of type equality explicit as a premise whenever necessary. These conventions are helpful for lifting the static semantics to the gradual setting~\citep{garciaAl:popl2016}. For closed terms, we write $\emptyenv; \emptyenv; \emptyenv |- t : T$, or simply $|- t : T$.

The dynamics are standard call-by-value semantics, specified using reduction frames. The only peculiarity is that they rely on {\em runtime type generation}: upon type application, a fresh type name $\alpha$ is generated and bound to the instantiation type $T$ in a global type name store $\sstore$. The notion of reduction and reduction rules all carry along the type name store. While type names only occur at runtime, and not in source programs, reasoning about \SPFL terms as they reduce requires accounting for programs with type names in them. This is why the typing rules are defined relative to a type name store as well. Similarly, type equality is relative to a type name store: a type name $\alpha$ is considered equal to its associated type in the store. The recursive definition of equality modulo type names is necessary to derive equalities~\citep{igarashiAl:icfp2017}. 
For instance, in the reduction of the well-typed program 
$(\mathit{id}\;[\Int -> \Int])\;(\mathit{id}\;[\Int])$, where $\mathit{id}$ is the polymorphic identity function, the equality 
$\eqrules{\alpha}{\beta -> \beta}[\alpha := \Int -> \Int, \beta := \Int]$ should be derivable.

Rules in Figure~\ref{fig:spfl} appeal to auxiliary well-formedness judgments, omitted for brevity\iffullv{ (\S\ref{asec:SFLanguage}).}[.] A type $T$ is well-formed ($\sstore; \Delta |- T$) if it only contains type variables in the type variable environment $\Delta$, and type names bound in a well-formed type name store. A type name store is well-formed ($|- \sstore$) if all type names are distinct, and associated to well-formed types. A type environment $\Gamma$ binds term variables to types, and is well-formed ($\Sigma; \Delta |- \Gamma$) if all types are well-formed. 

The decision of using type names instead of the traditional substitution semantics is 
in anticipation of gradualization, and based on prior work 
that has shown that runtime type generation is key in order to be able to distinguish between different type variables instantiated with the same type~\citep{matthewsAhmed:esop2008,ahmedAl:popl2011,ahmedAl:icfp2017}. We follow the approach already in \SPFL because we want the dynamics and type soundness argument of the static language to help us with \gsf, as afforded by AGT~\cite{garciaAl:popl2016}.

Unsurprisingly, \SPFL is type safe, and all well-typed terms are parametric. These results also follow from the properties of \gsf, and the strong relation between both languages.

\section{\gsf: Statics}
\label{sec:gsf-statics}

The first step of the Abstracting Gradual Typing methodology (AGT) is to define the syntax of gradual types and give them meaning through a concretization function to the set of static types they denote. Then, by finding the corresponding abstraction function to establish a Galois connection, the static semantics of the static language can be lifted to the gradual setting.

\subsection{Syntax and Syntactic Meaning of Gradual Types}
\label{sec:syntax-syntactic-meaning-gradual-types}
We introduce the syntactic category of gradual types $\cT \in \GType$, by admitting the unknown type in any position, namely:
$$ \cT  ::= \basetype | \cT -> \cT| \forall X. \cT| \pairtype{\cT}{\cT}| X  |\alpha| \?$$
Observe that static types $T$ are syntactically included in gradual types $G$.

\begin{figure}
\begin{footnotesize}
\begin{minipage}{0.4\textwidth}
    \begin{align*}
    \cs : \GType &-> \Pow^{*}(\Type)\\
    \conc{\basetype} &=  \set{\basetype}\\
    \conc{\cT_1 -> \cT_2} &=   \{T_1 -> T_2| T_1 \in \conc{\cT_1}, T_2 \in \conc{\cT_2} \} \\
    \conc{\pairtype{\cT_1}{\cT_2}} &= \{T_1 \pairsy T_2| T_1 \in \conc{\cT_1}, T_2 \in \conc{\cT_2}\}\\
    \conc{X} &= \set{X} \\ 
    \conc{\alpha} &= \set{\alpha} \\ 
    \conc{\forall X. \cT} &= \{\forall X. T| T \in \conc{\cT}[\store;\Delta,X]\}\\
    \conc{\?} &= \Type
  \end{align*}
\end{minipage}\quad\quad
\begin{minipage}{0.4\textwidth}
\begin{align*}
   \as : \Pow^{*}(\Type)&-> \GType \\
    \abst{\set{\basetype}} &=  \basetype\\
    \abst{\set{\overline{T_{i1} -> T_{i2}}}} &=  \abst{\set{\overline{T_{i1}}}} -> \abst{\set{\overline{T_{i2}}}}\\
    \abst{\set{\overline{T_{i1} \pairsy T_{i2}}}} &=  \abst{\set{\overline{T_{i1}}}} \pairsy \abst{\set{\overline{T_{i2}}}}\\
    \abst{\set{X}} &= X \\
    \abst{\set{\alpha}} &= \alpha \\
    \abst{\set{\overline{\forall X.T_i}}} &= \forall X. \abst{{\set{\overline{T_i}}}}[\store;\Delta, X]\\
    \abst{\set{\overline{T_i}}} &= \? \ \ \mathit{otherwise} 
  \end{align*}
\end{minipage}

\end{footnotesize}
 \caption{Type concretization ($\cs$) and abstraction ($\as$)}
  \label{fig:c-a}
\end{figure}

The syntactic meaning of gradual types is straightforward: the unknown type represents any type, and a precise type (constructor) represents the equivalent static type (constructor). In other words, $\Int -> \?$ denotes the set of all function types from $\Int$ to any static type. Perhaps surprisingly, we can simply extend this syntactic approach to deal with universal types, type variables, and type names; the concretization function $\cs$ is defined in Figure~\ref{fig:c-a}. Note that the definition is purely syntactic and does not even consider well-formedness ($\?$ stands for {\em any} static type); notions built above concretization, such as consistency, will naturally embed the necessary restrictions (\S\ref{sec:lifting-statics}).

Following the abstract interpretation framework, the notion of precision is not subject to tailoring: precision coincides with set inclusion of the denoted static types~\citep{garciaAl:popl2016}.

\begin{restatable}[Type Precision]{definition}{TypePrecision}
\label{def:precision}
$\tprules{\cT_1}{\cT_2}$ if and only if $\conc{\cT_1} \subseteq \conc{\cT_2}$.
 \end{restatable}

\begin{restatable}[Precision, inductively]{proposition}{Precisioninductively}
\label{Precisioninductively} The inductive definition of type precision given in Figure~\ref{fig:gsf-statics} is equivalent to Definition~\ref{def:precision}.
\end{restatable}

Observe that both $\forall X.X -> \?$ and $\forall X.\? -> X$ are more precise than $\forall X.? -> \?$, and less precise than $\forall X.X -> X$, thereby reflecting the original intuition about precision~\citep{siekAl:snapl2015}. Also $\forall X.? -> \?$ and $\? -> \?$ are unrelated by precision, since they correspond to different constructors (and \gsf is a language with {\em explicit} polymorphism); they are both more precise than $\?$, of course.

The notion of precision induces a notion of precision \emph{meet} between gradual types, which coincides with the abstraction of the intersection of both concretizations \cite{garciaAl:popl2016}.
\begin{restatable}[Precision Meet]{definition}{Meet}
\label{def:meet}
$\cT_1 \meet \cT_2 \triangleq \abst{\conc{\cT_1} \cap \conc{\cT_2}}$.
\end{restatable}

Dual to concretization is abstraction, which produces a gradual type from a non-empty set of static types. The abstraction function $\as$  is direct (Figure~\ref{fig:c-a}): it preserves type constructors and falls back on the unknown type whenever an heterogeneous set is abstracted. 
$\as$ is both sound and optimal: it produces the {\em most precise} gradual type that over-approximates a given set of static types.

\begin{restatable}[Galois connection]{proposition}{Galoisconnection}
\label{Galoisconnection}
$\pair{\cs}{\as}$ is a Galois connection, \ie:\\
$a)$ (Soundness) for any non-empty set of static types $S = \set{\overline{T}}$, we have $S \subseteq \conc{\abst{S}}$\\
$b)$ (Optimality) for any gradual type $G$, we have $\tprules{\abst{\conc{G}}}{G}$.
\end{restatable}

\begin{figure}[hp]
\begin{small}
  \begin{displaymath}
    \begin{array}{rcll}
      \multicolumn{4}{c}{  
      x \in \Var, X \in \VarType, \alpha \in \TypeName \quad
      \store  \in \TypeName \finto \GType,
      \Delta \subset \VarType,
      \Gamma \in \Var \finto \GType
      }\\
      \cT & ::= & \basetype | \cT -> \cT| \forall X. \cT| \pairtype{\cT}{\cT} | X  | \alpha | \? & \text{(gradual types)}\\

      t & ::= & \const | \lambda x:\cT.t | \Lambda X. t | \pair{t}{t}| x |t :: \cT | \op{\vectorOp{t}} | t\;t | t\;[\cT] | \proj{i}[t]   & \text{(gradual terms)}\\
    \end{array}   
  \end{displaymath}
 \end{small}
  \begin{small}
   \begin{flushleft}
  \framebox{$\EnvSG t : \cT$}~\textbf{Well-typed terms}
  \end{flushleft}
  \begin{mathpar}
    \inference[(G$\const$)]{  \ftype(\const) = \basetype & \swellGamma 
    }{\EnvSG \const : \basetype}
    \and
    \inference[(G$\lambda$)]{\EnvSG[\store][\Delta][\Gamma, x : \cT] t : \cT'
    }{\EnvSG \lambda x:\cT.t : \cT -> \cT'}
     \and
    \inference[(G$\Lambda$)]{\EnvSG[\store][\Delta, X][\Gamma]t : \cT & \swellGamma  }{\EnvSG \Lambda X.t : \forall X.\cT}
     \and
     \inference[(Gpair)]{\EnvSG t_1 : \cT_1 & \EnvSG t_2 : \cT_2  }{\EnvSG \pair{t_1}{t_2} : \pairtype{\cT_1}{\cT_2}}
     \and
      \inference[(Gx)]{x:\cT\in\Gamma & \swellGamma
    }{\EnvSG x : \cT}
    \and
    \inference[(Gasc)]{\EnvSG t : \cT && \ceqrules{\cT}{\cT'}
    }{\EnvSG t :: \cT' : \cT'} \and

   \inference[(Gop)]{\EnvSG \vectorOp{t} : \vectorOp{\cT_1} & \ftype(op) = \vectorOp{\cT_2} -> \cT \\ \ceqrules{\vectorOp{\cT_1}}{\vectorOp{\cT_2}}}{\EnvSG \op{\vectorOp{t}} : \cT}
     \and
    \inference[(Gapp)]{\EnvSG t_1 : \cT_1 & \EnvSG t_2 : \cT_2 \\ \ceqrules{\consistent{\dom}(\cT_1)}{\cT_2}
    }{\EnvSG t_1\;t_2 :\consistent{\cod}(\cT_1 )}
     \and
    \inference[(GappG)]{ \EnvSG t : \cT& \gtwf{\cT'} 
    }{\EnvSG t\; [\cT'] : \cinsta(\cT, \cT')}
    \and
    \inference[(Gpair$i$)]{ \EnvSG t : \cT }{\EnvSG 
    \proj{i}[t] : \cprojt{i}{\cT}}\\
  \end{mathpar}
  \begin{displaymath}
    \begin{block}
    \cdom : \GType \rightharpoonup \GType \\
    \cdom(\cT_1 -> \cT_2) = \cT_1\\
    \cdom(\?) = \? \\
    \cdom(\cT)\undefinedow
    \end{block}  
    ~ \quad
    \begin{block}
    \ccod : \GType \rightharpoonup \GType \\
    \ccod(\cT_1 -> \cT_2) = \cT_2\\
    \ccod(\?) = \? \\
    \ccod(\cT)\undefinedow
    \end{block} 
    ~ \quad
    \begin{block}
    {\cinsta : \GType^2 \rightharpoonup \GType} \\
    {\cinsta(\forall X. {\cT}, \cT') = \cT[\cT'/X]}\\
    {\cinsta(\?, \cT') = \?} \\
    {\cinsta(\cT, \cT')\undefinedow}
    \end{block} 
    ~
    \quad
    \begin{block}
    {\cprojts_i : \GType \rightharpoonup \GType} \\
    {\cprojt{i}{\pairtype{\cT_1}{\cT_2}} = \cT_i}\\
    {\cprojt{i}{\?} = \?}\\
    {\cprojt{i}{\cT}\undefinedow}
    \end{block} 
    \end{displaymath}
\end{small}
\begin{small}
\begin{flushleft}
\framebox{$\ceqrules{\cT}{\cT}$}~\textbf{Type consistency}
\end{flushleft}
    \begin{mathpar} 
           \inference{|- \store}{\ceqrules{B}{B}} \and
           \inference{|- \store & X \in \Delta}{\ceqrules{X}{X}} \and
           \inference{\ceqrules{\cT_1}{\cT_1'} & \ceqrules{\cT_2}{\cT_2'}}{\ceqrules{\cT_1 -> \cT_2}{\cT_1' -> \cT_2'}}\and
           \inference{\ceqrules{\cT_1}{\cT_2}[{\store; \Delta, X}]}{\ceqrules{\forall X. \cT_1}{\forall X. \cT_2}}\and
           \inference{\ceqrules{\cT_1}{\cT_1'} & \ceqrules{\cT_2}{\cT_2'}}{\ceqrules{\cT_1 \pairsy \cT_2}{\cT_1' \pairsy \cT_2'}}\and
           \inference{|- \store & \alpha \in \dom(\store)}{\ceqrules{\alpha}{\alpha}} \and
          \inference{\ceqrules{\store(\alpha)}{\cT}}{\ceqrules{\alpha}{\cT}}\and
           \inference{\ceqrules{\cT}{\store(\alpha)}}{\ceqrules{\cT}{\alpha}}\and
           \inference{\gtwf{\cT}}{\ceqrules{\cT}{\?}}\and
           \inference{\gtwf{\cT}}{\ceqrules{\?}{\cT}}
    \end{mathpar}
  \end{small}
  \begin{small}
\begin{flushleft}
\framebox{$\tprules{\cT}{\cT}$}~\textbf{Type precision}
\end{flushleft}
\begin{mathpar}
\inference{}{\tprules{\basetype}{\basetype}}
 \and
 \inference{}{\tprules{X}{X}}\and
           \inference{\tprules{\cT_1}{\cT_1'} & \tprules{\cT_2}{\cT_2'}}{\tprules{\cT_1 -> \cT_2}{\cT_1' -> \cT_2'}}\and
           \inference{\tprules{\cT_1}{\cT_2}[\store;\Delta, X]}{\tprules{\forall X. \cT_1}{\forall X. \cT_2}}\and
           \inference{\tprules{\cT_1}{\cT_1'} & \tprules{\cT_2}{\cT_2'}}{\tprules{\cT_1 \pairsy \cT_2}{\cT_1' \pairsy \cT_2'}}\and
           \inference{}{\tprules{\alpha}{\alpha}}\and
           \inference{}{\tprules{\cT}{\?}}  
\end{mathpar}
\end{small}

 \caption{\gsf: Syntax and Static Semantics}
  \label{fig:gsf-statics}
\end{figure}

\subsection{Lifting the Static Semantics}
\label{sec:lifting-statics}

The key point of AGT is that once the meaning of gradual types is agreed upon, there is no space for ad hoc design in the static semantics of the language. The abstract interpretation framework provides us with the {\em definitions} of type predicates and functions over gradual types, for which we can then find equivalent inductive or algorithmic {\em characterizations}.

In particular, a predicate on static types induces a counterpart on gradual types through {\em existential} lifting. Our only predicate in \SPFL is type equality, whose existential lifting is type consistency:

\begin{restatable}[Consistency]{definition}{Consistency}
\label{def:consistency}
  $\ceqrules{\cT_1}{\cT_2}$ if and only if $\eqrules{T_1}{T_2}[\sstore]$  for some $\sstore \in
  \conc{\store}$, $T_i \in \conc{\cT_i}$.
\end{restatable}
For closed types we write
$\cT_1 \sim \cT_2$.
This definition uses a {\em gradual} type name store $\store$, which binds type names to gradual types. Its concretization is the pointwise concretization:
\begin{align*}
    \conc{\cdot} &=  \emptyset & 
    \conc{\store, \alpha := \cT} &= \set{\sstore, \alpha := T | \sstore \in \conc{\store}, T \in \conc{\cT}}
  \end{align*}
Note that because consistency is the consistent lifting of static type equality, which does impose well-formedness, consistency is only defined on well-formed types (\ie~$\ceqrules{X}{X}[\emptyenv;\emptyenv]$ does {\em not} hold).

\begin{restatable}[Consistency, inductively]{proposition}{Consistencyinductively} The inductive definition of type consistency given in Figure~\ref{fig:gsf-statics} is equivalent to Definition~\ref{def:consistency}.
\end{restatable}
Again, observe that the resulting definition of consistency relates any two types that only differ in unknown type components, without any restriction.
Also, because of explicit polymorphism, top-level constructors must match, so $\? -> \?$ is not consistent with $\forall X.\? -> \?$. However, in line with gradual typing, both are consistent with $\?$, as expected.
Therefore \gsf does not treat $\? -> \?$ as a special ``quasi-polymorphic'' type, unlike \sysFg~\cite{igarashiAl:icfp2017}. Rather, consistency in \gsf coincides with that of \csa~\cite{xieAl:esop2018}.

Lifting type functions such as  $\dom$ requires abstraction: a lifted function is the abstraction of the results of applying the static function to all the denoted static types~\citep{garciaAl:popl2016}:
\begin{restatable}[Consistent lifting of functions]{definition}{Consistentliftingoffunctions}
\label{def:fun-lift}
Let $F_n$ be a function of type $\Type^n -> \Type$. Its consistent lifting  
$\consistent{F_n}$, of type $\GType^n -> \GType$, is defined as: $\consistent{F_n}(\overline{\cT}) = \abst{\set{ F_n(\overline{T}) | \overline{T} \in \overline{\conc{\cT}}}} $
\end{restatable}

The abstract interpretation framework allows us to prove the following definitions:
\begin{restatable}[Consistent type functions]{proposition}{Consistenttypefunctions}
\label{Consistenttypefunctions} The definitions of 
$\cdom$, $\ccod$, $\cinsta$, and $\cprojts_i$ given in Fig.~\ref{fig:gsf-statics} are consistent liftings, as per Def.~\ref{def:fun-lift}, of the corresponding functions from Fig.~\ref{fig:spfl}.
\end{restatable}

The gradual typing rules of \gsf (Figure~\ref{fig:gsf-statics}) are obtained by replacing type predicates and functions with their corresponding liftings. Note that in (Gapp), the premise $\ceqrules{\consistent{\dom}(\cT_1)}{\cT_2}$ is a compositional lifting of the corresponding premise in (Tapp), as justified by \citet{garciaAl:popl2016}. Of particular interest here is the fact that a term of unknown type can be optimistically treated as a polymorphic term and hence instantiated, yielding $\?$ as the result type of the type application ($\cinsta(\?, \cT') = \?$). In contrast, a term of function type, even imprecise, cannot be instantiated because the known top-level constructor does not match (\eg~$\cinsta(\? -> \?, \cT')$ is undefined).

\subsection{Static Properties of \gsf}
\label{sec:static-proprties}

As established by \citet{siekTaha:sfp2006} in the context of simple types, we can prove that the \gsf type system is equivalent to the \SPFL type system on fully-static terms. We say that a gradual type is static if the unknown type does not occur in it, and a term is static if it is fully annotated with static types. Let $|-_S$ denote the typing judgment of \SPFL.\footnote{As usual, the propositions here are stated over closed terms, but are proven as corollaries of statements over open terms.}

\begin{restatable}[Static equivalence for static terms]{proposition}{Staticequivalenceforstaticterms} 
\label{prop:static-eq}
Let $t$ be a static term and $\cT$ a static type ($\cT = T$). We have 
$|-_S t : T$ if and only if $|- t : T$
\end{restatable}

The second important property of the static semantics of a gradual language is the static gradual guarantee, which states that typeability is monotonic with respect to precision~\citep{siekAl:snapl2015}. 

Type precision (Def.~\ref{def:precision}) extends  to {\em term} precision. A term $t$ is more precise than a term $t'$ if they both have the same structure and $t$ is more precisely annotated than $t'$\iffullv{ (\S\ref{asec:GSFStaticProperties})}. The static gradual guarantee ensures that removing type annotations does not introduce type errors (or dually, that gradual type errors cannot be fixed by making types more precise).

\begin{restatable}[Static gradual guarantee]{proposition}{Staticgradualguarantee}
\label{prop:Staticgradualguarantee}
Let $t$ and $t'$ be closed \gsf terms such that $t \gprec t'$ and $|- t : \cT$.
Then $|- t' : \cT'$ and $\cT \gprec \cT'$.
\end{restatable}

\section{\gsf: Evidence-Based Dynamics}
\label{sec:gsf-dynamics}

We now turn to the dynamic semantics of \gsf. As anticipated, this is where the complexity of gradual parametricity manifests. Still, in addition to streamlining the design of the static semantics, AGT provides effective (though incomplete) guidance for the dynamics. 
In this section, we first briefly recall the main ingredients of the AGT approach to dynamic semantics, namely {\em evidence} for consistent judgments and {\em consistent transitivity}. We then describe the reduction rules of \gsf by treating evidence as an abstract datatype. This allows us to clarify a number of key operational aspects before turning in \S\ref{sec:evidence} to the details of the representation and operations of evidence that enable \gsf to satisfy parametricity while adequately tracking type instantiations.

\subsection{Background: Evidence-Based Semantics for Gradual Languages}
\label{sec:agt-ev}

For obtaining the dynamic semantics of a gradual language, AGT augments a consistent judgment (such as consistency or consistent subtyping) with the {\em evidence} of {\em why} such a judgment holds. Then, reduction mimics proof reduction of the type preservation argument of the static language, combining evidences through steps of {\em consistent transitivity}, which either yield more precise evidence, or fail if the evidences to combine are incompatible. A failure of consistent transitivity corresponds to a cast error in a traditional cast calculus~\citep{garciaAl:popl2016}.

Consider the gradual typing derivation of $(\lambda x:\?. x+1)\; \false$. In the inner typing derivation of the function, the consistent judgment $\? \sim \Int$ supports the addition expression, and at the top-level, the judgment $\Bool \sim \?$ supports the application of the function to $\false$. When two types are involved in a consistent judgment, we {\em learn} something about each of these types, namely the justification of {\em why} the judgment holds. This justification can be captured by a pair of gradual types, $\ev = \newev{\cT_1,\cT_2}$, which are at least as precise as the types involved in the judgment~\citep{garciaAl:popl2016}. (Throughout this article, we use blue color for evidence $\ev$ to enhance readability of the structure of terms.)

\begin{small}
$$ \ev \Vdash \cT_1 \sim \cT_2 \iff
\ev \gprec \abst[2]{\{\pr{T_1,T_2} | T_1 \in \conc{\cT_1}, T_2 \in \conc{\cT_2}, T_1 = T_2\}}$$
\end{small}

\ie~if evidence $\newev{\cT'_1,\cT'_2}$ justifies the consistency judgment $\cT_1 \sim \cT_2$, then $\cT'_1 \gprec \cT_1$ and $\cT'_2 \gprec \cT_2$.
For instance, by knowing that $\? \sim \Int$ holds, we learn that the first type can only possibly be $\Int$, while we do not learn anything new about the right-hand side, which is already fully static. Therefore the evidence of that judgment is $\ev[1] = \newev{\Int,\Int}$. Similarly, the evidence for the second judgment is $\ev[2] = \newev{\Bool,\Bool}$. Types in evidence can be gradual, \eg~$\newev{\? -> \?,\? -> \?}$ justifies that $(\? -> \?) \sim \?$.
Note that with the lifting of simple static type equality, both components of the evidence always coincide, so evidence can be represented as a single gradual type. However, type equality in \SPFL is more subtle (\S\ref{sec:spfl-lang}), so the general presentation of evidence as pairs is required.

At runtime, reduction rules need to {\em combine} evidences in order to either justify or refute a use of transitivity in the type preservation argument. In our example, we need to combine $\ev[1]$ and $\ev[2]$ in order to (try to) obtain a justification for the transitive judgment, namely that $\Bool \sim \Int$. The combination operation, called {\em consistent transitivity} $\trans{=}$, determines whether two evidences support the transitivity: here, 
$\ev[2] \trans{=}\ev[1]=\newev{\Bool,\Bool}\trans{=}\newev{\Int,\Int}$ is undefined, so a runtime error is raised.

The evidence approach is very general and scales to disciplines where consistent judgments are not symmetric, involve more complex reasoning, and even other evidence combination operations~\citep{garciaAl:popl2016,lehmannTanter:popl2017}. All the definitions involved are justified by the abstract interpretation framework. Also, both type safety and the dynamic gradual guarantee become straightforward to prove. In particular, the dynamic gradual guarantee follows directly from the monotonicity (in precision) of consistent transitivity. In fact, the generality of the approach even admits evidence to range over other abstract domains; for instance, for gradual security typing with references, evidence is defined with {\em label intervals}, not gradual labels~\citep{toroAl:toplas2018}.

\subsection{Reduction for \gsf}
\label{sec:gsf-reduction}

\begin{figure}[hp]
\vspace{1em}
\begin{small}
  \begin{displaymath}
  \begin{array}{rcll}
    t &::=& v | \pair{t}{t} | x | \cast{\ev}{t} :: \cT | \op{\vectorOp{t}} | 
      t \; t | t \; [ \cT ] | \proj{i}[t]  & (\text{terms})\\
    v & ::= & \cast{\ev}{u} :: \cT & (\text{values})\\
    u & ::= & \const | \lambda x : \cT. t  | \Lambda X. t | \pair{u}{u} & (\text{raw values})\\
    \end{array}   
  \end{displaymath}
   \end{small}
  \begin{small}
   \begin{flushleft}
  \framebox{$\EnvSG \tu : \cT$}~\textbf{Well-typed terms} (for conciseness, $s$ ranges over both $t$ and $u$)
  \end{flushleft}
  \begin{mathpar}
    \inference[($E$\const)]{\ftype(\const) = \basetype & \swellGamma
    }{ \staticgJ{\const} \initu{\basetype}}
     \and
    \inference[($E\lambda$)]{\staticgJ[\store; \Delta; \Gamma, x : \cT]{\iterm{\cT_2}} \inTermT{\cT'}
    }{ \staticgJ{\lambda x:{\cT}.\iterm{\cT}} \initu{\cT -> \cT'}}
     \and
    \inference[($E\Lambda$)]{\staticgJ[{\store;\Delta, X}]{\iterm{\cT}} \inTermT{\cT}& \swellGamma}{ \staticgJ{\Lambda X.\iterm{\cT}} \initu{\forall X.\cT}}
     \and
     \inference[($E$pair)]{\staticgJ{\tu_1} \inTermT{\cT_{1}} &&  \staticgJ{\tu_2} \inTermT{\cT_{2}}}
      {\staticgJ{  \pair{\tu_1}{\tu_2}} \initu{\cT_1 \pairsy \cT_{2}}}
     \and
     \inference[($E$x)]{
    x:\cT\in\Gamma & \swellGamma}{\staticgJ{x} \inTermT{\cT}}
    \and
    \inference[($E$asc)]{\staticgJ{\tu} \inTermT{\cT} &&  
    \Gbox{\ev \Vdash \ceqrulessimpl{\cT}{\cT'}}
    }{\staticgJ{\Gbox{\ev \tu :: \cT'}} \inTermT{\cT'}}
    \and
    \inference[($E$op)]{\staticgJ{\vectorOp{\iterm{}}} \inTermT{\vectorOp{\B}} && \ftype(\aop) = \vectorOp{\B} -> \B'}
      {\staticgJ{\op{\vectorOp{\iterm{}}}} \inTermT{\B'}}
     \and
    \inference[($E$app)]{\staticgJ{\iterm{\cT_{11} -> \cT_{12}}_1} \inTermT{\cT -> \cT'} &&  \staticgJ{\iterm{\cT_{11}}_2} \inTermT{\cT}}
      {\staticgJ{\iterm{\cT_{11} -> \cT_{12}}_1 \; \iterm{\cT_{11}}_2} \inTermT{\cT'}}
     \and
    \inference[($E$app$\cT$)]{\staticgJ{\iterm{\forall X. \cT}} \inTermT{\forall X.\cT} & \gtwf{\cT'} }
    {\staticgJ{\iterm{\forall X. \cT_1} \; [\cT']} \inTermT{ \cT[\cT'/X]}}
    \and
    \inference[($E$pairi)]{\staticgJ{\iterm{}} \inTermT{\pairtype{\cT_1}{\cT_2}}}
    {\staticgJ{\proj{i}[t]} \inTermT{\cT_i}}
  \end{mathpar}
\end{small} 
\vspace{1em}
\begin{small}  
\begin{flushleft}
\framebox{$\conf{t} \nred \conf{t} \text{ or } \error$}
~\textbf{Notion of reduction}
\end{flushleft}
\begin{displaymath}
\begin{array}{p{0.6cm}rcl}
\text{($R$asc})& \conf{\cast{\ev[2]}(\cast{\ev[1]}{u} :: \cT_1) :: \cT_2} &\nred&
  \begin{cases}
  \conf{\cast{(\ev[1] \trans{=} \ev[2])}{u} :: \cT_2 }\\
  \error \qquad \text{if not defined}
  \end{cases}
  \\ 
\text{($R$op)} &   \conf {\op{\vectorOp{{\ev  u :: G}}}} &\nred& 
  \conf{ \ev[\basetype]\; \redop{\vectorOp{u}} :: \basetype}
  \quad\text{where } \basetype \triangleq \cod(\ftype(\aop))
  \\
\text{($R$app)}&   \conf{(\cast{\ev[1]}{
              (\lambda x:\cT_{11}.t) :: \cT_{1} -> \cT_{2}})\;
          (\cast{\ev[2]}{u} :: \cT_1)}
  &\nred&
          \begin{block}
          \begin{cases}
          \conf{\cast{\invcod(\ev[1])}{
              (  t[
              \cast{(\ev[2] \trans{=} \invdom(\ev[1]))}
               {u} :: 
              \cT_{11})/x])} :: \cT_{2}} 
              \\
          \error \qquad \text{if not defined}
          \end{cases}
          \end{block}
  \\
  \text{($R$pair)}&   \conf{{\pair{\cast{\ev[1]}{u_1} :: \cT_1}{\cast{\ev[2]}{u_2} :: \cT_2}}} 
  &\nred& 
  \conf{\cast{(\pairtype{\ev[1]}{\ev[2]})}{\pair{u_1}{u_2}} :: \pairtype{\cT_1}{\cT_2}}
  \\
\text{($R$proj$i$)}&   \conf{\tproj{i}[\cast{\ev}{\pair{u_1}{u_2}} :: \pairtype{\cT_1}{\cT_2}]} 
  &\nred& 
  \conf{\cast{\evproj{i}[\ev]}{u_i} :: \cT_i}\\
\text{($R$app$\cT$)}&   \conf{(\cast{\ev}{
              \Lambda X. t :: {\forall X. \cT}})\;
          [\cT']}
  &\nred&
          \conf[\store']{
            \cast{\evout}{(\cast{\evinst{\ev}{\evlift{\alpha}}}
            \substTermPaper{X}{\evlift{\alpha}}{t} :: \cT[\alpha/X])}} :: \cT[\cT'/X] 
            \\
            & & &
            \text{where } \store' \triangleq \store, \alpha := {\cT'} \text{ for some } \alpha \notin \dom(\store) \\
            & & & \text{and } \evlift{\alpha} = \evliftname_{\store'}(\alpha)\\
\end{array} 
\end{displaymath}
\end{small} 
\begin{small}  
\begin{flushleft}
\framebox{$\conf{t} \red \conf{t} \text{ or } \error$}
~\textbf{Evaluation frames and reduction}
\end{flushleft}
\begin{equation*}
      \begin{array}{rcl}
      f & ::= & {{\cast{\ev}{ \lcorchete] }::\cT}} | \op{\vectorOp{v}, [] , \vectorOp{t}} | [] \ t | v \ [] | [] \ [\cT] | \pair{[]}{t}| \pair{v}{[]}|\tproj{i}[{[]}]
      \end{array}
  \end{equation*} 
\begin{mathpar}
\inference[($R-->$)]{\conf{t} \nred \conf[\store']{t'}}{ \conf{t} \red \conf[\store']{t'}}
\and
\inference[($Rf$)] {\conf{ t } \red \conf[\store']{ t'}}{ \conf {f[t]} \red \conf[\store']{ f[t']}}\\
\inference[($R$err)]{\conf{t} \nred \error}{ \conf{t} \red \error}\and 
\inference[($Rf$err)]{\conf{t} \red \error}{ \conf{f[t]} \red \error}
\end{mathpar}
\end{small}
 \caption{\glangev: Syntax, Static and Dynamic Semantics}
  \label{fig:gsfe}
\end{figure}

In order to denote reduction of (evidence-augmented) gradual typing derivations, \citet{garciaAl:popl2016} use {\em intrinsic} terms as a notational device; while appropriate, the resulting description is fairly hard to comprehend and unusual, and it does implicitly involve a (presentational) transformation from source terms to their intrinsic representation. 

In this work, we simplify the exposition by avoiding the use of intrinsic terms; instead, we rely on a type-directed, straightforward translation that inserts explicit ascriptions everywhere consistency is used---very much in the spirit of the coercion-based semantics of subtyping~\citep{pierce:tapl}.
For instance, the small program of \S\ref{sec:agt-ev} above, $(\lambda x:\?. x+1)\; \false$,  is translated to:
$$(\ev[\? -> \Int](\lambda x:\?. (\ev[1]x :: \Int) + (\ev[\Int]1 :: \Int)) :: \? -> \Int)\; (\ev[2](\ev[\Bool]\false :: \Bool) :: \?)$$
where $\ev[\cT]$ is the evidence of the reflexive judgment $\cT \sim \cT$  (\eg$\ev[\Int]$ supports $\Int \sim \Int$). Evidences $\ev[1]$ and $\ev[2]$ are the ones from \S\ref{sec:agt-ev}. Such initial evidences are computed by means of an {\em interior} function, given by the abstract interpretation framework~\citep{garciaAl:popl2016}: in this setting, the interior\iffullv{ (\S\ref{asec:Interior})} coincides with the precision meet (\S\ref{sec:syntax-syntactic-meaning-gradual-types}). The definition
of the type-preserving translation is direct\iffullv{ (\S\ref{asec:translation})}.

Despite this translation, we do preserve the essence of the AGT dynamics approach in which evidence and consistent transitivity drive the runtime monitoring aspect of gradual typing. Furthermore, by making the translation explicitly ascribe all base values to their base type, we can present a uniform syntax and greatly simplify reduction rules compared to the original AGT exposition. This presentation also streamlines the proofs by reducing the number of cases to consider.

Figure~\ref{fig:gsfe} presents the syntax and semantics of \gsfe, a simple variant of \gsf in which all values are ascribed, and ascriptions carry evidence. Key changes with respect to Figure~\ref{fig:gsf-statics} are highlighted in gray. Here, we treat evidence as a pair of elements of an {\em abstract} datatype; we define its actual representation (and operations) in the next section. 

Because the translation from \gsf to \gsfe introduces explicit ascriptions everywhere consistency is used, the only remaining use of consistency in the typing rules of \gsfe is in the rule (Easc). The evidence of the term itself supports the consistency judgment in the premise. All other rules require types to match exactly; the translation inserts ascriptions to ensure that top-level constructors match in every elimination form. 

The notion of reduction for \gsfe terms deals with evidence propagation and composition with consistent transitivity. 
Rule ($R$asc) specifies how an ascription around an ascribed value reduces to a single value if consistent transitivity holds: $\ev[1]$ justifies that $\cT_u \sim \cT_1$, where $\cT_u$ is the type of the underlying simple value $u$, and $\ev[2]$ is evidence that $\cT_1 \sim \cT_2$. The composition via consistent transitivity, if defined, justifies that $\cT_u \sim \cT_2$; if undefined, reduction steps to $\error$. 
Rule ($R$op) simply strips the underlying simple values, applies the primitive operation, and then wraps the result in an ascription, using a canonical base evidence $\ev[\basetype]$ (which trivially justifies that $\basetype \sim \basetype$).
Rule ($R$app) combines the evidence from the argument value $\ev[2]$ with the domain evidence of the function value $\invdoma{\ev[1]}$ in an attempt to transitively justify that $\cT_u \sim \cT_{11}$. Failure to justify that judgment, as in our example in \S\ref{sec:agt-ev}, produces $\error$.
The return value is ascribed to the expected return type, using the codomain evidence of the function $\invcoda{\ev[1]}$. 
Rule ($R$pair) produces a pair value when the subterms of a pair have been reduced to values themselves, using the product operator on evidences $\evc{\pairtype{\ev[1]}{\ev[2]}}$. 
This rule is necessary to enforce a uniform presentation of all values as ascribed values, which simplifies technicalities.
Dually, Rule ($R$proj$i$) extracts a component of a pair and ascribes it to the projected type, using the corresponding evidence obtained with $\evproj{i}[\ev]$ (not to be confused with $\proj{i}[\ev]$, which refers to the first projection of evidence, itself a metalanguage pair).

Apart from the presentational details, the above rules are standard for an evidence-based reduction semantics. Rule ($R$app$\cT$) is {\em the} rule that specifically deals with parametric polymorphism, reducing a type application. This is where most of the complexity of gradual parametricity concentrates.
Observe that there are two ascriptions in the produced term: 
\begin{itemize}
\item The {\em inner} ascription (to $\cT[\alpha/X]$) is for the body of the polymorphic term, asserting that substituting a fresh type name $\alpha$ for the type variable $X$ preserves typing. The associated evidence $\evinst{\ev}{\evlift{\alpha}}$ is the result of instantiating $\ev$ (which justifies that the actual type of $\Lambda X.t$ is consistent with $\forall X. \cT$) with the fresh type name, hence justifying that the body after substitution is consistent with $\cT[\alpha/X]$. 
\item The {\em outer} ascription asserts that $\cT[\alpha/X]$ is consistent with $\cT[\cT'/X]$, witnessed by evidence $\evout$. 
We define $\evout$ 
in \S\ref{sec:ref-evidence} below, once the representation of evidence is introduced. 
\end{itemize}

The use of $\evlift{\alpha}$ is a technicality: because so far we treat evidence as an abstract datatype from an as-yet-unspecified domain, say pairs of $\EElemType$, we cannot directly use gradual types ($\GType$) inside evidences. The connection between $\GType$ and $\EElemType$ is specified by lifting operations,
$\evliftname_{\store} : \GType  -> \EElemType$ and $\evunliftname : \EElemType -> \GType$. (In standard AGT~\citep{garciaAl:popl2016} the lifting is simply the identity, \ie~$\EElemType = \GType$.)
Because type names have meaning related to a store, the lifting is parameterized by the store $\store$. Term substitution is mostly standard: it uses $\evunliftname$ to recover $\alpha$, and is extended to substitute recursively in evidences.
Substitution in evidence, also triggered by evidence instantiation, is simply component-wise substitution on evidence types.

Finally, the evaluation frames and associated reduction rules in Figure~\ref{fig:gsfe} are straightforward; in particular ($R$err) and ($\mathit{Rf}$err) propagate
 $\error$ to the top-level. 

\section{Evidence for Gradual Parametricity}
\label{sec:evidence}

We now turn to the actual representation of evidence. We first explain in \S\ref{sec:evfail} why the standard representation of evidence as pair of gradual types is insufficient for gradual parametricity. We then introduce the  refined representation of evidence to enforce parametricity (\S\ref{sec:ref-evidence}), and basic properties of the language. 
Richer properties of \gsf are discussed in \S\ref{sec:gsf-param}, \S\ref{sec:gsf-dgg} and \S\ref{sec:free-theorems}.

\subsection{Simple Evidence, and Why It Fails}
\label{sec:evfail}

In standard AGT~\citep{garciaAl:popl2016}, evidence is simply represented as a pair of gradual types, \ie~$\EElemType = \GType$. Consistent transitivity is defined through the abstract interpretation framework. The definition for simple types is as follows ($\ev \Vdash J$ means $\ev$ supports the consistent judgment $J$):
\begin{definition}[Consistent transitivity]\label{def:ct} 
Suppose $\isjudgment{\ev[ab]}{\cT_a}{\cT_b}$ and $\isjudgment{\ev[bc]}{\cT_b}{\cT_c}$. Evidence for consistent transitivity is deduced as $\isjudgment{(\ev[ab] \trans{=} \ev[bc] )}{\cT_a}{\cT_b}$, where:\\[0.5em]
$\braket{{\cT_1, \cT_{21}}} \trans{=} \braket{{\cT_{22}, \cT_{3}}} =
\as^2(\{\pr{T_1, T_3} \in \conc{\cT_1}\times\conc{\cT_3}| \exists T_2 \in \conc{\cT_{21}} \cap \conc{\cT_{22}},
T_1 = T_2 \land T_2 = T_3 \})
$
\end{definition}
In words, if defined, the evidence that supports the transitive judgment is obtained by abstracting over the pairs of static types denoted by the outer evidence types ($\cT_1$ and $\cT_3$) {\em such that} they are connected through a static type common to both middle evidence types ($\cT_{21}$ and $\cT_{22}$). This definition can be proven to be equivalent to an inductive definition that proceeds in a syntax-directed manner on the structure of types~\citep{garciaAl:popl2016}.

Consistent transitivity satisfies some important properties. First, it is associative. Second, the resulting evidence is more precise than the outer evidence types, reflecting that during evaluation, typing justification only gets more precise (or fails). Violating this property breaks type safety. The third property is key for establishing the dynamic gradual guarantee~\citep{garciaAl:popl2016}. 
\begin{restatable}{lemma}{ctopt}(Properties of consistent transitivity).
\label{lm:ctopt}\label{lm:ctassoc} \label{lm:ctmon}  \mbox{}\\
(a) Associativity. $(\ev[1]\trans{=} \ev[2]) \trans{=} \ev[3] =
\ev[1]\trans{=} (\ev[2] \trans{=} \ev[3])$, or both are undefined.\\
(b) Optimality. If $\ev = \ev[1]\trans{=} \ev[2]$ is defined, then $\proj{1}[\ev] \sqsubseteq \proj{1}[\ev[1]]$ and
$\proj{2}[\ev] \sqsubseteq \proj{2}[\ev[2]]$.\\
(c) Monotonicity. If $\ev[1] \gprec \ev[1]'$ and $\ev[2] \gprec \ev[2]'$ and $\ev[1] \trans{} \ev[2]$ is defined, then 
$ \ev[1]\trans{=}\ev[2]  \gprec \ev[1]'\trans{=}\ev[2]'$.
\end{restatable}

Unfortunately, adopting gradual types for evidence types and simply extending the consistent transitivity definition to deal with \gsf types and consistency judgments yields a gradual language that breaks parametricity.\footnote{The obtained language is type safe, and satisfies the dynamic gradual guarantee. This novel design could make sense to gradualize impure polymorphic languages, which do not enforce parametricity. }
To illustrate, consider this simple program:
\begin{lstlisting}
((*$\Lambda$*)X.((*$\lambda$*)x:X.let y:? = x in let z:? = y in z + 1)) [Int] 1
\end{lstlisting}
The function is not parametric because it ends up adding $1$ to its argument, although it does so after two intermediate bindings, typed as $\?$. Without further precaution, the parametricity violation of this program would not be detected at runtime. Assume that the type application generates the fresh name $\alpha$, bound to $\Int$ in the store. 
For justifying that \lstinline{x} can flow to \lstinline{y} (the let-binding is equivalent to a function application), we need evidence for ${\Int}\sim{\?}$ by consistent transitivity between the evidences $\pr{\Int, \alpha}$, which justifies $\Int \sim \alpha$, and $\pr{\alpha, \alpha}$, which justifies $\alpha \sim \?$. Observe that conversely to the simply-typed setting, both components of evidence are not necessarily equal. The evidence $\pr{\alpha, \alpha}$ is obtained by substituting $\alpha$ for $X$ in the initial evidence $\pr{X,X}$ for $X \sim \?$.
 Using the definition of consistent transitivity (Def.~\ref{def:ct}), $\pr{\Int, \alpha} \trans{=} \pr{\alpha, \alpha} = \pr{\Int, \alpha}$.
Similarly, for justifying the flow of \lstinline{y} to \lstinline{z}, the previous evidence must be combined with $\pr{\?, \?}$, which  justifies $\?\sim\?$. By Def.~\ref{def:ct}, $\pr{\Int, \alpha} \trans{=} \pr{\?, \?} = \as^2(\set{\pr{\Int, \Int}, \pr{\Int, \alpha}}) = \pr{\Int, \?}$. This evidence can subsequently be used to produce evidence to justify that the addition is well-typed, since $\pr{\Int, \?} \trans{=} \pr{\Int, \Int} = \pr{\Int, \Int}$. Therefore the program produces $2$, without errors: parametricity is violated. 

\subsection{Refining Evidence}
\label{sec:ref-evidence}
For gradual parametricity, evidence must do more than just ensure type safety. It needs to safeguard the sealing that type variables are meant to represent, also taking care of unsealing as necessary. First of all, we need to define evidence to adequately represent consistency judgments of \gsf.

\myparagraph{Evidence Types}
We define {\em evidence types}, $\cE \in \EType$, to be an enriched version of gradual types:
     \begin{displaymath}
        \begin{array}{rcll}
          \cE & ::= & \basetype | \cE -> \cE| \forall X. \cE| \pairtype{\cE}{\cE}| \Gbox{\richtype{\cE}}| X | \?  \\
        \end{array} 
      \end{displaymath}  
\SPFL equality judgments, and hence \gsf consistency judgments, are relative to a store. It is therefore not enough to use type names in evidence: we need to keep track of their associated types in the store. An evidence type name $\richtype{\cE}$ therefore captures the type associated to the type name $\alpha$ through the store. For instance, evidence that a variable has a polymorphic type $X$ is initially $\pr{X,X}$. When $X$ is instantiated, say to $\Int$, and a fresh type name $\alpha$ is introduced, 
the evidence becomes $\pr{\richtype{\Int},\richtype{\Int}}$.
An evidence type name does not only record the end type to which it is bound, but the whole path. For instance, $\richtype{\beta^\Int}$ is a valid evidence type name that embeds the fact that $\alpha$ is bound to $\beta$, which is itself bound to $\Int$.

Note that as a program reduces, evidence can not only become more precise than statically-used types, but also than the global store. For instance, it can be the case that $\alpha := \?$ in the global store $\store$, but that locally, the evidence for $\alpha$ has gotten more precise, such as $\richtype{\Int}$. 
Formally, a type name is enriched with its transitive bindings in the store, $\enrich{\alpha}{\store} = \richtype{\enrich{\store[](\alpha)}{\store}}$. 
Unlifting simply forgets the additional information: $\evunlift(\richtype{\cE}) = \alpha$.
In all other cases, both operations recur structurally.

It is crucial to understand the intuition behind the {\em position} of type names in a given evidence. The position of $\richtype{\cE}$ in an evidence can correspond to a {\em sealing}, an {\em unsealing}, or neither. If $\richtype{\cE}$ is {\em only} on the right side, \eg~$\pr{\Int, \richtype{\Int}}$, then the evidence is a sealing (here, of $\Int$ with $\alpha$). Dually, if $\richtype{\cE}$ is {\em only} on the left side, \eg~$\pr{\richtype{\Int}, \Int}$, the evidence is an unsealing (here, of $\Int$ from $\alpha$). Sealing and unsealing evidences arise through reduction, as will be illustrated later in this section.

\begin{figure}[t]
  \begin{small}
\begin{flushleft}
{\footnotesize  (unsl)\hspace{16em}(idL)\hspace{13em}(sealL)}
\end{flushleft}
\begin{mathpar}
\inference[]{\pr{\cE_1, {\cE_2}} \transc \pr{{\cE_3}, \cE_4} = \pr{\cE'_1, \cE'_2}}
      {\pr{\cE_1, \richtype{\cE_2}} \transc \pr{\richtype{\cE_3}, \cE_4} = \pr{\cE'_1, \cE'_2}}
      \and
\inference[]{}{\pr{\cE, \cE} \transc \pr{\?, \?} = \pr{\cE, \cE}}
      \and 
\inference[]{
        \pr{\cE_1, {\cE_2}} \transc \pr{{\cE_3}, \cE_4} = \pr{\cE'_1, \cE'_2}}
      {\pr{\cE_1, \cE_2} \transc \pr{\cE_3, \richtype{\cE_4}} = \pr{\cE'_1, \richtype{\cE'_2}}}
      \and
\inference[(func)]{
        \pr{\cE_{41}, \cE_{31}}  \trans{} \pr{\cE_{21}, \cE_{11}}  = \pr{\cE_{3}, \cE_{1}} &
        \pr{\cE_{12}, \cE_{22}} \trans{} \pr{\cE_{32}, \cE_{42}} = \pr{\cE_{2}, \cE_{4}}
      }
      {\pr{\cE_{11} \barr \cE_{12}, \cE_{21} \barr \cE_{22}} \trans{=} \pr{\cE_{31} \barr \cE_{32}, \cE_{41} \barr \cE_{42}} = 
      \pr{\cE_{1} \barr \cE_{2},\cE_{3} \barr \cE_{4}}}
  \and
\inference[(func$\?$L)]{
        \pr{\cE_{1}\barr \cE_{2}, \cE_{3} \barr \cE_{4}} \trans{=} \pr{\unkL \barr \unkL, \unkR \barr \unkR} = 
        \pr{\cE'_1 \barr \cE'_2, \cE'_3 \barr \cE'_4}
      }
      {\pr{\cE_{1} \barr \cE_{2}, \cE_{3} \barr \cE_{4}} \trans{=} \pr{\unkL, \unkR} = 
      \pr{\cE'_1\barr \cE'_2, \cE'_3 \barr \cE'_4}
      }
    \end{mathpar}
  \end{small}
 \caption{Consistent Transitivity (selected rules)}
  \label{fig:gsf-trans}
\end{figure}

\myparagraph{Consistent Transitivity}
With this syntactic enrichment, consistent transitivity can be strengthened to account for sealing and unsealing, ensuring parametricity. Consistent transitivity is defined inductively; selected rules are presented in Figure~\ref{fig:gsf-trans}. 

Rule (unsl) specifies that when a sealing and an unsealing of the same type name meet in the middle positions of a consistent transitivity step, the type name can be eliminated in order to calculate the resulting evidence. For instance, $\pr{\Int, \richtype{\Int}} \trans{=} \pr{\richtype{\?}, \?} = \pr{\Int, \Int} \trans{=} \pr{\?, \?} = \pr{\Int,\Int}$.

As shown in \S\ref{sec:evfail}, it is important for consistent transitivity to not lose precision when combining an evidence with an unknown evidence. To this end, rule (identL)
in Fig.~\ref{fig:gsf-trans} preserves the left evidence. 
Going back to the example of \S\ref{sec:evfail}, we now have $\pr{\Int, \richtype{\Int}} \trans{=} \pr{\unkG{\alpha},\unkG{\alpha}} = \pr{\Int, \richtype{\Int}}$, instead of $\pr{\Int, \?}$. Because $\pr{\Int, \richtype{\Int}} \trans{=} \pr{\Int, \Int}$ is undefined, reduction steps to $\error$ as desired.

Rule (sealL) shows that when an evidence is combined with a sealing, the resulting evidence is also a sealing. This sealing can be more precise, \eg~$\pr{\Int, \Int} \trans{=} \pr{\?, \richtype{\?}} = \pr{\Int, \richtype{\Int}}$.

Figure~\ref{fig:gsf-trans} only shows one structurally-recursive rule, corresponding to the function case (func); consistent transitivity is computed recursively with the domain and codomain evidences. To combine a function evidence with unknown evidence, the unknown evidence is first ``expanded'' to match the type constructor (func?L). There are similar rules for the other type constructors. Also, there are symmetric variants of the above rules---such as (identR) and (sealR)---in which every evidence and every evidence type is swapped. \iffullv{The complete definition is provided in \S\ref{asec:ConsistentTransitivity}.} 
 
Importantly, this refined definition of consistent transitivity preserves associativity and optimality, based on a natural notion of precision for evidence types\iffullv{ (\S\ref{asec:ev-precision})}. It does however break monotonicity, and consequently, the dynamic gradual guarantee (we come back to this in \S\ref{sec:gsf-dgg}).

\myparagraph{Outer Evidence} 
The reduction rule of a type application ($R$app$\cT$) produces an outer evidence $\evout$ that justifies that $\cT[\alpha/X]$ is consistent with $\cT[\cT'/X]$. The precise definition of this evidence is delicate, addressing a subtle tension between the precision required for justifying unsealing when possible, and the imprecision required for parametricity. 
$$
\evout \triangleq \pr{\cEU [\richtype{\cE}], \cEU [\cE']}
\qquad \text{where } \cEU \!=\! \enrich{\unlift{\proj{2}[\ev]}}{\store}, 
\richtype{\cE} \!=\! \enrich{\alpha}{\store'}, \cE'\!=\!\enrich{G'}{\store}
$$
In this definition, $\ev$, $\alpha$, $\cT'$, $\store$, and $\store'$ come from rule ($R$app$\cT$). 
Determining $\cEU$ is the key challenge.
The second evidence type of $\ev$ refines $\forall X.\cT$ by exploiting the fact that the underlying polymorphic value $\Lambda X.t$ is consistent with it;
this extra precision is crucial for unsealing. The roundtrip unlift/lift ``resets'' the sealing information of evidence type names to that contained in the store; this relaxation is crucial for parametricity (to prove the compositionality lemma---\S\ref{sec:gsf-param}).

Note that $\evout$ will never cause a runtime error when combined with the resulting evidence of the parametric term result, because both are necessarily related by precision.

\myparagraph{Illustration} The following reduction trace illustrates all the important aspects of reduction: 

\begin{small}
\begin{mathpar}
\begin{array}{p{2.1em}p{0.8em}lr}
&&(\ev[\forall X.X\barr X](\Lambda X.\lambda x:X.x)::\forall X.X\barr\?) \;[\Int] \; (\ev[\Int]1::\Int) & \text{\footnotesize initial evidence}\\
\footnotesize($R$app$\cT$)&$\red$&
  (\pr{\richtype{\Int}\barr\richtype{\Int},\Int\barr\Int}(\ev[\alpha\barr\alpha]
  (\lambda x:\alpha.x)::\alpha\barr\?)::\Int\barr\?) \; (\ev[\Int]1::\Int) 
  & \text{\footnotesize note the precision of $\evout$}\\
\footnotesize($R$asc)&$\red$&
  (\pr{\richtype{\Int}\barr\richtype{\Int},\Int\barr\Int}
  (\lambda x:\alpha.x)::\Int\barr\?) \; (\ev[\Int]1::\Int)
  & \text{\footnotesize consistent transitivity}\\
\footnotesize($R$app)&$\red$&
  \pr{\richtype{\Int},\Int}(\pr{\Int,\richtype{\Int}}1::\alpha)::\?
  & \text{\footnotesize argument is sealed}\\
\footnotesize($R$asc)&$\red$& 
  \pr{\Int,\Int}1::\?
  & \text{\footnotesize unsealing eliminates $\alpha$}
\end{array} 
\end{mathpar}
\end{small}

Crucially, the initial evidence of the identity function is fully precise, even though it is ascribed an imprecise type. Consequently, in the first reduction step above, $\evout$ is calculated as:

\begin{small}
 \begin{displaymath}
  \evout \triangleq \pr{\cEU [\richtype{\cE}], \cEU [\cE']} = \pr{(\forall X.X\barr X) [\richtype{\Int}] , (\forall X.X\barr X) [\Int] } = \pr{\richtype{\Int} \barr \richtype{\Int}, \Int \barr \Int}
 \end{displaymath}
 \end{small}

The application step ($R$app) then gives rise to sealing and unsealing evidences after deconstructing $\evout$: the inner evidence $\pr{\Int,\richtype{\Int}}$ seals the number $1$ at type $\alpha$, while the outer evidence $\pr{\richtype{\Int},\Int}$ allows the subsequent unsealing in the ascription step ($R$asc). 
As a result, the ascribed identity function yields usable values, because the outer evidence subsequently takes care of unsealing. This addresses the excess of failure reported with \lamB and \sysFc in \S\ref{sec:dynamic-issues}. Note that if the function were not intrinsically precise on its return type, \eg~$\Lambda X.\lambda x:X.(x::\?)$, then initial evidence would likewise be imprecise, and deconstructing $\evout$ would {\em not} justify unsealing the result anymore. 

\subsection{Basic Properties of \gsf Evaluation}
\label{sec:basic-properties-of-gsf-evaluation}
The runtime semantics of a \gsf term are given by first translating the term to \gsfe (noted $|- t  \translate \te : \cT$) and then reducing the \gsfe term. We write $\gsfreds{t}{\valuestore{v}[\store]}$ (resp. $\gsfreds{t}{\error}$) if $|- t  \translate \te : \cT$ and $\conf[\emptyenv] \te \red^{*} \conf[\store]{v}$ (resp. $\conf[\emptyenv] \te \red^{*} \conf[\store]{\error
}$) for some resulting store $\store$. We write $\valuestore{v}[\store]:\cT$ for $\EnvSG[\store][\emptyenv][\emptyenv] v \inTermT{\cT}$.
We write $t \divergesy$ if the translation of $t$ diverges, and $\gsfreds{t}{v}$ when the store is irrelevant.

The properties of \gsf follow from the same properties of \gsfe, expressed using the small-step reduction relation, due to the fact that the translation $\translate$ preserves typing\iffullv{ (\S\ref{asec:translation})}. In particular, \gsf terms do not get stuck, although they might produce $\error$ or diverge:

\begin{restatable}[Type Safety]{proposition}{typesafety}
\label{prop:typesafety}
If $|- t : \cT$ then either $\gsfreds{t}{\valuestore{v}[\store]}$ with $
\valuestore{v}[\store]:\cT$, $\gsfreds{t}{\error}$, or $t \divergesy$.
\end{restatable}

Proposition~\ref{prop:static-eq} established that \gsf typing coincides with \SPFL typing on static terms. A similar result holds considering the dynamic semantics. In particular, static \gsf terms never produce $\error$:

\begin{restatable}[Static terms do not fail]{proposition}{dgequivs}
  \label{prop:dgequivs}
  \mbox{}
  Let $t$ be a static term. 
  If $|- t : T$ then $\neg(\gsfreds{t}{\error})$.
\end{restatable}
This result follows from the fact that all evidences in a static program are static, hence never gain precision; the initial type checking ensures that combination through transitivity never fails.

\section{\gsf and the Dynamic Gradual Guarantee} 
\label{sec:gsf-dgg}

The previous section clarified several aspects of the semantics of \gsf programs, by establishing type safety, and by showing that static terms do not fail. This section studies the dynamic gradual guarantee (DGG)~\citep{siekAl:snapl2015}, also known as graduality~\cite{newAhmed:icfp2018,newAl:popl2020}. 
In a big-step setting, this guarantee essentially says that if $|- t : \cT$ and $\gsfreds{t}{v}$, then for any $t'$ such that $t \gprec t'$, we have  $\gsfreds{t'}{v'}$ for some $v'$ such that $v \gprec v'$. Intuitively: losing precision is harmless, or, reducibility is monotonous with respect to precision.

Unfortunately, in order to enforce parametricity (\S\ref{sec:gsf-param}), and as already alluded to earlier (\S\ref{sec:ref-evidence}), \gsf does not satisfy the DGG. First, we exhibit a counterexample, and identify the non-monotonicity of consistent transitivity as the root cause for this behavior (\S\ref{sec:dgg-violation}). 
Then, in order to better understand the behavior of \gsf programs when losing precision, we study a weaker variant of the DGG---weaker in the sense that it is valid for a stricter notion of precision---first in \gsfe (\S\ref{sec:wdgg-gsfe}) and then in \gsf (\S\ref{sec:wdgg-gsf}). 
The idea of devising a stricter notion of precision for which a variant of the DGG can be satisfied was first explored by \citet{igarashiAl:icfp2017}, even though they leave the proof of such a result for \sysFg as a conjecture. 

\subsection{Violation of the Dynamic Gradual Guarantee in \gsf}
\label{sec:dgg-violation}

To show that \gsf does not satisfy the dynamic gradual guarantee (DGG), it is sufficient to exhibit two terms in \gsf, related by precision, whose behavior contradicts the DGG. Consider the polymorphic identity function $\idx \triangleq \Lambda X. \lambda x:X.x::X$, and an imprecise variant $\idu \triangleq \Lambda X. \lambda x:\?.x::X$. 
Then $\idx\;[\Int]\;1\Downarrow 1$, but $\idu\;[\Int]\;1\Downarrow \error$, despite the fact that $\idx\;[\Int]\;1 \gprec \idu\;[\Int]\;1$.

Conceptually, it is interesting to shed light on what causes such a violation. Recall that \citet{garciaAl:popl2016} prove the DGG for their language using (mostly) the monotonicity of consistent transitivity (Prop~\ref{lm:ctmon} (c)). In fact, while not sufficient, we can prove that monotonicity of consistent transitivity (CT) is a {\em necessary} condition for the DGG to hold:

\begin{restatable}[$\neg$ monotonicity of CT $=>$ $\neg$ DGG]{proposition}{not-monotone-implies-not-dgg}
\label{prop:not-monotone-implies-not-dgg}
Let $\ev[1] \gprec \evp[1]$, $\ev[2] \gprec \evp[2]$, 
$\ev[1] \Vdash \cT_{1} \sim \cT_{2}$,
$\ev[2] \Vdash \cT_{2} \sim \cT_{3}$,
$\evp[1] \Vdash \cT'_{1} \sim \cT'_{2}$,
$\evp[2] \Vdash \cT'_{2} \sim \cT'_{3}$,
where $\cT_i \gprec \cT'_i$.\\
If $\ev[1]\trans{=}\ev[2]  \not\gprec \evp[1] \trans{=}\evp[2]$,
then $\exists t \gprec t'$, such that $t \red v$, $t' \red v'$ such that $v \not\gprec v'$.
\end{restatable}
\begin{proof}
  Let $t \triangleq \ev[2](\ev[1]u::\cT_{2})::\cT_{3}$, and $t' \triangleq  \evp[2](\evp[1]u'::\cT'_{2})::\cT'_{3}$, for some $u \gprec u'$. 
  We know $t \gprec t'$.
  Let $\ev[1]\trans{=}\ev[2] = \ev[12]$ and $\evp[1] \trans{=}\evp[2] = \evp[12]$, then
  $t \red \ev[12]u::\cT_{3}$ and $t' \red \evp[12]u::\cT'_{3}$,
  but as $\ev[12] \not\gprec \evp[12]$ then 
  $\ev[12]u::\cT_{3} \not\gprec \evp[12]u'::\cT'_{3}$ and the result holds.
\end{proof}

\citet{garciaAl:popl2016} study a language without universal types. But in \gsf, because of universal types, there is an additional monotonicity condition that is necessary for the DGG to hold: \emph{monotonicity of evidence instantiation} (EI). Monotonicity of EI states that given two type abstractions related by precision, the new evidences created after type application remain related. Formally:

\begin{restatable}[$\neg$ monotonicity of EI $=>$ $\neg$ DGG]{proposition}{not-monotone-ei-implies-not-dgg}
\label{prop:not-monotone-ei-implies-not-dgg}
Let $\ev[1] \gprec \ev[2]$, $\cT_1 \gprec \cT_2$, $\store_1 \gprec \store_2$, $\alpha := {\cT_1} \in \store_1$, $ \alpha := {\cT_2} \in \store_2$,
\mbox{$\evlift{\alpha_1} = \evliftname_{\store_1}(\alpha)$}, 
\mbox{$\evlift{\alpha_2} = \evliftname_{\store_2}(\alpha)$}, and
$\evinst{\ev[1]}{\evlift{\alpha_1}}$ is defined.\\
If 
$\evinst{\ev[1]}{\evlift{\alpha_1}} \not\gprec \evinst{\ev[2]}{\evlift{\alpha_2}}$, or $\evoutE[\ev][1] \not\gprec \evoutE[\ev][2]$,
then $\exists t \gprec t'$, such that $t \red v$, $t' \red v'$ such that $v \not\gprec v'$.
\end{restatable}
\begin{proof}
  Let $t \triangleq \ev[1](\Lambda X.t_1)::\_\;[\cT_1]$, and $t' \triangleq 
  \ev[2](\Lambda X.t'_1)::\_\;[\cT_2]$, for some $t_1 \gprec t'_1$. 
  We know $t \gprec t'$.
  We know that 
  $\conf[\store_1]{t} \red 
  \conf[\store_1, \alpha := \cT_1]{\evoutE[\ev][1](\evinst{\ev[1]}{\evlift{\alpha_1}}t_2::\_)::\_}$ 
  and 
  $\conf[\store_2]{t'} \red 
  \conf[\store_2, \alpha := \cT_2]{\evoutE[\ev][2](\evinst{\ev[2]}{\evlift{\alpha_2}}t'_2::\_)::\_}$ 
  but as either $\evoutE[\ev][1] \not\gprec \evoutE[\ev][2]$ or
  $\evinst{\ev[1]}{\evlift{\alpha_1}} \not\gprec \evinst{\ev[2]}{\evlift{\alpha_2}}$ 
  then $\evoutE[\ev][1](\evinst{\ev[1]}{\evlift{\alpha_1}}t_2::\_)::\_ \not\gprec \evoutE[\ev][2](\evinst{\ev[2]}{\evlift{\alpha_2}}t'_2::\_)::\_$ and the result holds.
\end{proof}

As mentioned in \S\ref{sec:ref-evidence}, monotonicity of consistent transitivity is broken by the strengthening we impose to enforce parametricity. For instance, consider $\pr{\Int, \richtype{\Int}} \gprec \pr{\Int, \richtype{\Int}}$ and $\pr{\richtype{\Int}, \Int} \gprec \pr{\?, \?}$. By consistent transitivity, $\pr{\Int, \richtype{\Int}} \trans{} \pr{\richtype{\Int},\Int} = \pr{\Int, \Int}$ (rule unsl), and
$\pr{\Int, \richtype{\Int}} \trans{} \pr{\?, \?} = \pr{\Int, \richtype{\Int}}$  (rule idL), but $\pr{\Int, \Int} \not\gprec \pr{\Int, \richtype{\Int}}$.
Therefore the DGG cannot be satisfied as such. 
We later on discuss a tension between our notion of parametricity and the DGG (\S\ref{sec:gsf-param}), but first, we look at how to characterize the set of terms for which loss of precision is indeed harmless in \gsf.

\subsection{A Weak Dynamic Gradual Guarantee for \gsfe}
\label{sec:wdgg-gsfe}

One way to accommodate the dynamic gradual guarantee in languages like \lamB, \gsf, and \sysFg, would be to change the definition of type (and term) precision. This is the approach taken by \citet{igarashiAl:icfp2017}, 
although they do not prove that the DGG holds with this adjusted precision, and leave it as a conjecture.
Dually, if one sticks to the natural notion of precision, as adopted by both \gsf and \csa, and justified by the AGT interpretation, reconciliation might come from considering other forms of parametricity, or perhaps less flexible gradual language designs~\cite{devrieseAl:popl2018}.
Inspired by the approach of \citet{igarashiAl:icfp2017}, we now explore an alternative notion of precision for which the DGG does hold. Conversely to \citet{igarashiAl:icfp2017}, however, we do not intend this alternative precision relation to be the one used to typecheck programs, but only to serve as a technical device to characterize harmless losses of precision. 

Clearly, \gsf terms do not start to behave unexpectedly when losing precision. At most, they may fail and violate the DGG. These failures come from the interaction between polymorphic types and imprecision, which affects runtime sealing with type names. Therefore one might expect the DGG to hold at least for the simply-typed subset of \gsf. Even more, we observe that in \gsf, losing precision {\em extrinsically} (\ie through imprecise type ascriptions) has a different impact on reducibility, compared to losing precision {\em intrinsically} (\ie by modifying the types of binders). Specifically, extrinsic loss of precision is harmless in \gsf when the ascribed term is closed with respect to type variables. This means for instance that a fully-static polymorphic function embedded in a gradual context and used adequately (type-wise) will behave as expected.\footnote{Our prior work on \gsf~\cite{toroAl:popl2019} formalizes exactly this result. Here, we go further and establish more general results, from which this sort of preservation of behavior for ascribed static terms easily follows (\S\ref{sec:free-theorems}).}  

For instance, in \S\ref{sec:dgg-violation}, $\idu$ presents an {\em intrinsic} loss of precision compared to the polymorphic identity function $\idx$, because the term binder changes from having type $X$ to having type $\?$. The runtime error when applying $\idu$ violates the DGG. In contrast, the following function $\idxu = (\Lambda X. \lambda x:X.x::X) :: \forall X.?->X$, which has the same imprecise type as $\idu$, presents an {\em extrinsic} loss of precision compared to $\idx$. 
Interestingly, $\idxu$ behaves just as $\idx$.
  
\myparagraph{Strict precision} 
Armed with the intuition presented above, we define a strict notion of precision, noted $\sprec$, which closely characterizes \gsfe terms for which monotonicity of consistent transitivity holds. While not sufficient, this property is necessary for the DGG to hold, as established in Prop.~\ref{prop:not-monotone-implies-not-dgg}.
Of course, because strict precision $\sprec$ is more restrictive than standard precision $\gprec$, the dynamic gradual guarantee that one may establish with respect to it is {\em weaker}; hereafter, we denote it \sdgg.

We define {\em strict type precision} $\sprec$ in Figure~\ref{fig:sprec}. For brevity, in this section we focus on the core of \gsf and \gsfe and omit both pairs and operators on base types.
Intuitively, $\sprec$ avoids any interference between runtime sealing and loss of precision. More specifically, $\sprec$ is similar to $\gprec$, save for
universal types, type variables and type names: these are not more precise than the unknown type anymore.
For instance,  $\forall X. X -> X \nsprec \forall X. X -> \? \nsprec \?$.
We say $\cT_1$ is ``more strictly precise'' than $\cT_2$ when $\cT_1 \sprec \cT_2$.
The definition of $\sprec$ for evidences is defined as the natural lifting of $\sprec$ to evidence types $\cE$. Crucially, monotonicity of consistent transitivity holds with respect to $\sprec$.

\begin{figure}[t]
\begin{small}
\begin{mathpar}   
  \inference{}{\B \sprec \B} \and
  \inference{}{X \sprec X} \and
  \inference{}{\alpha \sprec \alpha} \and
  \inference{}{\B \sprec \?} \and
  \inference{\functype{\cT_1}{\cT_2} \sprec \functype{\?}{\?}}{\functype{\cT_1}{\cT_2} \sprec \?} \and
  \inference{}{\? \sprec \?} \and
  \inference{\cT_1 \sprec \cT_3 & \cT_2 \sprec \cT_4}{\functype{\cT_1}{\cT_2} \sprec \functype{\cT_3}{\cT_4}} \and
  \inference{\cT_1 \sprec \cT_2}{\forall X. \cT_1 \sprec \forall X. \cT_2} \and
\end{mathpar}
\end{small}
\caption{\gsf: Strict type precision}
\label{fig:sprec}
\end{figure}

\begin{restatable}[$\sprec$-Monotonicity of Consistent Transitivity]{proposition}{monotone-trans-precsM}
\label{prop:monotone-trans-precsM}
  If $\ev[1] \sprec \ev[2]$, $\ev[3] \sprec \ev[4]$, and $\ev[1] \trans{=} \ev[3]$ is defined,  then $\ev[1] \trans{=} \ev[3] \sprec \ev[2] \trans{=} \ev[4]$. 
\end{restatable}

For illustration purposes, let us recall the counterexample to monotonicity presented in \S\ref{sec:dgg-violation}.
Consider $\pr{\Int, \richtype{\Int}} \gprec \pr{\Int, \richtype{\Int}}$ and $\pr{\richtype{\Int}, \Int} \gprec \pr{\?, \?}$. By consistent transitivity, 
$\pr{\Int, \richtype{\Int}} \trans{} \pr{\richtype{\Int},\Int} = \pr{\Int, \Int}$ (rule unsl), and
$\pr{\Int, \richtype{\Int}} \trans{} \pr{\?, \?} = \pr{\Int, \richtype{\Int}}$  (rule idL),
but
$\pr{\Int, \Int} \not\gprec \pr{\Int, \richtype{\Int}}$.
This argument is not longer valid with strict precision, as $\richtype{\Int} \nsprec \?$ and therefore $\pr{\richtype{\Int}, \Int} \nsprec \pr{\?, \?}$.

\myparagraph{\sdgg for \gsfe}
Armed with strict precision, and the fact that consistent transitivity is monotone with respect to it, we can prove a weak dynamic gradual guarantee, denoted \sdgg, for \gsfe. We establish the \sdgg for \gsf in \S\ref{sec:wdgg-gsf}.

First, we define strict type precision for \gsfe terms in Figure~\ref{fig:sprec-term}, which relates two possibly-open terms and their respective types. It is worth noting that types can be related in $\sprec$ or $\gprec$ depending on the rule. The precision judgment $\impR{s_1}{s_2}{\cT_1}{\cT_2}$ uses $\Omega$ to relate variables. $\Omega$ binds a variable to a pair of types related by precision. Rule ($\sprec$x$ _{\ev}$) establishes $\impR{x}{x}{\cT_1}{\cT_2}$ if $x:\cT_1 \gprec \cT_2\in\Omega$, and Rule ($\sprec$$\lambda_{\ev}$) extends $\Omega$ with the annotated types of the functions to relate.

Strict term precision is the natural lifting of strict type precision $\sprec$ to terms, except for types that do not influence evidence in the runtime semantics, namely function argument types and ascription types: for these, we can use the more liberal type precision relation $\gprec$.
For example, Rule~($\sprec$asc$ _{\ev}$) has the premise $\cT_1 \gprec \cT_2$. This means for instance that 
the term $\cast{\ev[\forall X. X -> X]}{(\Lambda X.\lambda x: X. x)} :: \forall X. X -> X$ is more strictly precise than $\cast{\ev[\forall X. X -> X]}{(\Lambda X.\lambda x: X. x)} :: \forall X. \? -> X$ because evidences are the same (and thus related by $\sprec$), whereas the type annotations are related by $\gprec$. These two terms would not be related if we imposed a strict precision relation between ascribed types. Note that these two terms correspond to the translation of the \gsf terms $\idx$ and $\idxu$ discussed previously. In contrast, the translation of $\idu$ is {\em not} related by strict precision because the associated evidence $\ev[\forall X. ? -> X]$ is not related to $\ev[\forall X. X -> X]$.

Rule ($\sprec$appG$ _{\ev}$) states that types involved in a type application must be related by strict precision because they do influence evidence during reduction: after elimination of type abstractions, new evidences are created using these types, and such evidences need to be related as well. Note that this restriction is sufficient to satisfy monotonicity of evidence instantiation, which is needed for the dynamic gradual guarantee (Prop~\ref{prop:not-monotone-ei-implies-not-dgg}).

Finally, we need to strengthen the relation with an additional rule ($\sprec$Masc$ _{\ev}$) to account for \gsfe terms that are the result of the elaboration from \gsf. This will be important to scale the \sdgg from \gsfe to \gsf below. Recall that, as explained briefly in \S\ref{sec:gsf-dynamics}, the translation from \gsf to \gsfe introduces evidences to ensure that \gsfe terms are well-typed. 
In particular, the translation uses {\em type matching} $\similar{}{}$~\cite{ciminiSiek:popl2016} (repeated here in Figure~\ref{fig:sprec-term}) to ascribe subterms of type $\?$ in elimination positions to the corresponding top-level type constructor (\eg~$\forall X.\?$ for a type application, or $\?->\?$ for a function application). For subterms of a more precise type, type matching is the identity. When an actual matching expansion occurs, the corresponding evidence is generated (\eg~$\ev[\forall X.\?]$, or $\ev[\?->\?]$). Such evidences are related by $\gprec$, but not necessarily by $\sprec$. Rule ($\sprec$Masc$ _{\ev}$) accounts for the case where they are not.

Figure~\ref{fig:sprec-term} also defines strict type precision for \gsfe stores and configurations. A store is more strictly precise than another if it binds each type name
to a more strictly precise type. Finally, a configuration is more strictly precise than another if the store and term components are more strictly precise, and the terms are well-typed with their respective stores.

\begin{figure}[t]
\begin{small}
\begin{flushleft}
\framebox{$\impR{s}{s}{\cT}{\cT}$}
~\textbf{Strict term precision} (for conciseness, $s$ ranges over both $t$ and $u$) \\
\;
\end{flushleft}
  \begin{mathpar}
    \inference[($\sprec$$\const_{\ev}$)]{\ftype(\const) = \basetype
    }{\impR{b}{b}{\B}{\B}}
    \and
    \inference[($\sprec$$\lambda_{\ev}$)]{\impR[\Omega, x : \cT_1 \gprec \cT_2]{t_1}{t_2}{\cT'_1}{\cT'_2} & \cT_1 \gprec \cT_2
    }{\impR{\lambda x:\cT_1.t_1}{\lambda x:\cT_2.t_2}{\functype{\cT_1}{\cT'_1}}{\functype{\cT_2}{\cT'_2}}}
     \and
    \inference[($\sprec$$\Lambda_{\ev}$)]{\impR{t_1}{t_2}{\cT_1}{\cT_2}}
    {\impR{\Lambda X.t_1}{\Lambda X.t_2}{\forall X.\cT_1}{\forall X.\cT_2}} \and
    \inference[($\sprec$x$ _{\ev}$)]{x:\cT_1 \gprec \cT_2\in\Omega 
    }{\impR{x}{x}{\cT_1}{\cT_2}} \and
    \inference[($\sprec$app$ _{\ev}$)]{
      \impR{t_1}{t_2}{\cT'_1 -> \cT_1}{\cT'_2 -> \cT_2} & 
      \impR{t'_1}{t'_2}{\cT'_1}{\cT'_2}
    }{\impR{t_1\;t'_1}{t_2 \; t'_2}{\cT_1}{\cT_2}} \and
    \inference[($\sprec$appG$ _{\ev}$)]{ \impR{t_1}{t_2}{\forall X.\cT_1}{\forall X.\cT_2} & \cT'_1 \sprec \cT'_2
    }{\impR{t_1\; [\cT'_1]}{t_2\; [\cT'_2]}{\cT_1 [\cT'_1/X]}{\cT_2 [\cT'_2/X]}} \and
    \inference[($\sprec$asc$ _{\ev}$)]{\ev[1] \sprec \ev[2] &  \impR{s_1}{s_2}{\cT'_1}{\cT'_2}  & \cT_1 \gprec \cT_2
    }
    {\impR{\cast{\ev[1]}{s_1} :: \cT_1}{\cast{\ev[2]}{s_2} :: \cT_2}{\cT_1}{\cT_2}} \and
    \inference[($\sprec$Masc$ _{\ev}$)]{\ev[\cT_1] \nsprec \ev[\cT_2] & \impR{t_1}{t_2}{\cT'_1}{\cT'_2} & \cT_1 \gprec \cT_2 & \similar{\cT'_1}{\cT_1} & \similar{\cT'_2}{\cT_2}
    }{\impR{\cast{\ev[\cT_1]}{t_1} :: \cT_1}{\cast{\ev[\cT_2]}{t_2} :: \cT_2}{\cT_1}{\cT_2}} 
     \end{mathpar}
     \begin{flushleft}
        \framebox{$\similar{\cT}{\cT}$}
        ~\textbf{Type matching} 
        \end{flushleft}  
      \begin{mathpar}
        \similar{\cT_1 -> \cT_2}{\cT_1 -> \cT_2} \and
        \similar{\forall X. \cT}{\forall X. \cT} \and
        \similar{\?}{\? -> \?} \and
        \similar{\?}{\forall X. \?} \and
      \end{mathpar}
      \\
        \;
       \begin{flushleft}
        \framebox{$\store |- t \sprec \store |- t$}
        ~\textbf{Configuration precision} \\
        \;
        \end{flushleft}  
      \begin{mathpar}  
\inference{\store_1 \sprec \store_2 & \impR[]{t_1}{t_2}{\cT_1}{\cT_2} & \store_1 |- t_1: \cT_1 & \store_2 |- t_2: \cT_2}
{\store_1 |- t_1 \sprec \store_2 |- t_2}\and
\inference{\forall \alpha \in \dom(\store_1), \store_1(\alpha) \sprec \store_2(\alpha)}{ \store_1 \sprec \store_2}
  \end{mathpar}  
\end{small}
\caption{\gsfe: Strict term, store and configuration precision}
\label{fig:sprec-term}
\end{figure}

The \sdgg holds for \gsfe: given two configurations related by strict precision, small-step reduction ($\red$) of the most precise one implies that of the less precise one. Alternatively, if the first configuration is already a value, then so is the second.
\begin{restatable}[Small-step \sdgg for \gsfe]{proposition}{dggLemma}
\label{prop:dggLemma}
Suppose $\store_1 |- t_1 \sprec \store_2 |- t_2$.
\begin{enumerate}[label=\alph*.]
  \item \label{case:caseBB}  If $\conf[\store_1]{t_1} \red \conf[\store'_1]{t'_1}$, then $\conf[\store_2]{t_2} \red \conf[\store'_2]{t'_2}$, and we have $\store'_1 |- t'_1 \sprec \store'_2 |- t'_2$.
  \item \label{case:caseCC} If $t_1 = v_1$, then $t_2 = v_2$.
\end{enumerate}
\end{restatable}

\noindent Next, we exploit this small-step \sdgg result for \gsfe in order to establish the \sdgg for \gsf.

\subsection{A Weak Dynamic Gradual Guarantee for \gsf}
\label{sec:wdgg-gsf}

\begin{figure}[t]
\begin{small}
\begin{flushleft}
\framebox{$\impRu{v}{v}{\cT}{\cT}$}
~\textbf{Strict value precision}
\end{flushleft}
  \begin{mathpar}
    \inference[($\sprec$$\const$)]{  \ftype(\const) = \basetype
    }{\impRu{b}{b}{\B}{\B}}
    \and
    \inference[($\sprec$$\lambda$)]{\impR[\Omega, x : \cT_1 \gprec \cT_2]{t_1}{t_2}{\cT'_1}{\cT'_2} & \cT_1 \gprec \cT_2
    }{\impRu{(\lambda x:\cT_1.t_1)}{(\lambda x:\cT_2.t_2)}{\functype{\cT_1}{\cT'_1}}{\functype{\cT_2}{\cT'_2}}}
     \and
    \inference[($\sprec$$\Lambda$)]{\impR{t_1}{t_2}{\cT_1}{\cT_2}}
    {\impRu{(\Lambda X.t_1)}{(\Lambda X.t_2)}{\forall X.\cT_1}{\forall X.\cT_2}}
     \end{mathpar}
\\
\;
\begin{flushleft}
\framebox{$\impR{t}{t}{\cT}{\cT}$}
~\textbf{Strict term precision} \\
\;
\end{flushleft}
     \begin{mathpar}
    \inference[($\sprec$x)]{x:\cT_1 \gprec \cT_2\in\Omega
    }{\impR{x}{x}{\cT_1}{\cT_2}} \and
    \inference[($\sprec$v)]{\impRu{v_1}{v_2}{\cT_1}{\cT_2} & 
    \cT_1 \sprec \cT_2
    }{\impR{v_1}{v_2}{\cT_1}{\cT_2}} \and
    \inference[($\sprec$ascv)]{\impRu{v_1}{v_2}{\cT'_1}{\cT'_2} & \cT'_1 \meet \cT_1 \sprec \cT'_2 \meet \cT_2 & \cT_1 \gprec \cT_2
    }{\impR{v_1 :: \cT_1}{v_2 :: \cT_2}{\cT_1}{\cT_2}} \and
    \inference[($\sprec$asct)]{\impR{t_1}{t_2}{\cT'_1}{\cT'_2} & \cT'_1 \meet \cT_1 \sprec \cT'_2 \meet \cT_2 & \cT_1 \gprec \cT_2 & t_1, t_2 \neq v
    }{\impR{t_1 :: \cT_1}{t_2:: \cT_2}{\cT_1}{\cT_2}} \and
    \inference[($\sprec$app)]{
      \impR{t_1}{t_2}{\cT_1}{\cT_2} & 
      \impR{t'_1}{t'_2}{\cT'_1}{\cT'_2} &
      \cT'_1 \meet \cdom(\cT_1) \sprec  \cT'_2 \meet \cdom(\cT_2)
    }{\impR{t_1\;t'_1}{t_2 \; t'_2}{\ccod(\cT_1)}{\ccod(\cT_2)}}
     \and
    \inference[($\sprec$appG)]{ \impR{t_1}{t_2}{\cT_1}{\cT_2} & \cT'_1 \sprec \cT'_2 
    }{\impR{t_1\; [\cT'_1]}{t_2\; [\cT'_2]}{\cinsta(\cT_1, \cT'_1)}{\cinsta(\cT_2, \cT'_2)}}
  \end{mathpar}  
\end{small}
  \caption{\gsf: Strict term precision}
  \label{fig:weak-dgg-source}
\end{figure}

It would be helpful for programmers to be able to reason about \gsf terms directly in order to understand the (variant of the) dynamic gradual guarantee that is satisfied.
To do so, we first have to define when two \gsf terms are in the strict precision relation. We could simply say that two \gsf terms are related by $\sprec$ if their translations to \gsfe are related by $\sprec$. Unfortunately, with this definition it would be hard for programmers to get an intuition about when two terms are related by strict precision, as it would require understanding the elaborations to \gsfe. 

Instead, we design a strict type relation $\sprec$ for \gsf terms syntactically. Like for \gsfe, just lifting of $\sprec$ for \gsf terms would be too conservative. 
For instance, we could not relate
$(\Lambda X.\lambda x:X. x) :: \forall X. X -> X$ and $(\Lambda X.\lambda x:X. x) :: \forall X. ? -> X$ because $\forall X. X -> X \not\sprec \forall X. ? -> X$, although as discussed before, this extrinsic imprecision is harmless (and the translated terms {\em are} related by strict precision in \gsfe). Figure~\ref{fig:weak-dgg-source} defines a strict precision relation for \gsf that soundly reflects strict precision for \gsfe, and can account for the translation of \gsf terms to \gsfe.

Most of the rules are straightforward. We use metavariable $v$ in \gsf to range over constants, functions and type abstractions, and use $\seprecv$ to relate them.
Relation $\seprecv$ does not require the types of values to be related in either $\sprec$ or $\gprec$; this is specified in rules ($\sprec$v) and ($\sprec$ascv).
Rule ($\sprec$v) demands that the intrinsic types of the values be related in $\sprec$.
In contrast, rule ($\sprec$ascv) is more permissive, establishing that the intrinsic types can be in $\gprec$---but only if the values have ascriptions such that their meet (\ie~initial evidences) are in $\sprec$. This allows capturing some intrinsic losses of precision, whenever the surrounding type information ensures that the associated evidence will be related by $\sprec$. For instance, $(\Lambda X.\lambda x:X. x :: X) :: \forall X. X -> X \sprec (\Lambda X.\lambda x:\?. x :: X) :: \forall X. X -> X$ at the corresponding type.

Rule ($\sprec$asct) uses the same technique to be as permissive as possible: it only requires $\cT_1 \gprec \cT_2$, but requires the meets of the types involved in the ascriptions to be related by $\sprec$.
Likewise, Rule ($\sprec$app) requires the meets
of the function argument types and the actual argument types to be related by $\sprec$.
Finally, Rule ($\sprec$appG) follows the \gsfe precision rule for type instantiation and uses $\sprec$ to relate the instantiation types.

Strict term precision for \gsf is sound with respect to strict term precision of the translated terms in \gsfe:

\begin{restatable}{proposition}{equivalencesprecelab}
\label{prop:equivalencesprecelab}
If $\impR[]{t_1}{t_2}{\cT_1}{\cT_2}$ and 
$|- t_i \translate t'_i  : \cT_i$,
 then $\impR[]{t'_1}{t'_2}{\cT_1}{\cT_2}$.
\end{restatable}

Finally, Proposition \S\ref{prop:equivalencesprecelab} and \S\ref{prop:dggLemma}, allow us to establish the strict dynamic gradual guarantee for \gsf programs.

\begin{restatable}[\sdgg]{theorem}{wdggMGSF}
\label{theorem:wdggMGSF}
 Suppose $\impR[]{t_1}{t_2}{\cT_1}{\cT_2}$ and $|- t_1: \cT_1$.
 \begin{enumerate}[label=\alph*.]
  \item If $\gsfreds{t_1}{v_1}$, then $\gsfreds{t_2}{v_2}$ and $\impR[]{v_1}{v_2}{\cT_1}{\cT_2}$. \\
   If $t_{1}\divergesy$ then $t_{2}\divergesy$.
  \item If $\gsfreds{t_2}{v_2}$, then $\gsfreds{t_1}{v_1}$ and $\impR[]{v_1}{v_2}{\cT_1}{\cT_2}$, or $\gsfreds{t_1}{\error}$.\\
   If $t_{2}\divergesy$, then $t_{1}\divergesy$ or $\gsfreds{t_1}{\error}$.
 \end{enumerate}
\end{restatable}

\section{\gsf: Gradual Parametricity}
\label{sec:gsf-param}

In this section, we first discuss two different notions of parametricity for gradual languages that have been developed in the literature (\S\ref{sec:parametricities}), in order to situate the notion of gradual parametricity for \gsf (\S\ref{sec:parametricity-in-GSF}). Then, we show in \S\ref{sec:gdd-violation} that this notion of parametricity is incompatible with the \gdgg. This tension is established solely driven by the definition of parametricity, and not by monotonicity of consistent transitivity (\S\ref{sec:dgg-violation}). This suggests that the incompatibility is shared by other languages with essentially the same notion of gradual parametricity, for which the dynamic gradual guarantee has so far been left as an open question.
Finally, we explore gradual free theorems in \gsf based on examples discussed in the literature, using both gradual parametricity and the \sdgg in order to establish such results (\S\ref{sec:free-theorems}).

\subsection{On Gradual Parametricities}
\label{sec:parametricities}
The notion of parametricity established by \citet{reynolds:83} is usually defined by interpreting types as {\em binary logical relations}. The fundamental property of such a relation (also known as the {\em abstraction theorem}) states that a well-typed term is related to itself at its type. 
This technique is fairly standard and can be summarized as follows.
The logical relation is defined using two mutually-defined interpretations: one for values and one for computations.
For simplicity and uniformity, throughout this section we use notation 
$(v_1, v_2) \in \setv{\cT}$ when values $v_1$ and $v_2$ are related values at type $\cT$ under environment $\rho$, and notation $(t_1, t_2) \in \sett{\cT}$ when terms $t_1$ and $t_2$ are related computations at type $\cT$ under environment $\rho$. An environment $\rho$, which maps a type variable to two types and a relation, is used to relate values at abstract types as explained below. Let us briefly go through the definitions.
Two base values are related if they are the same: 
$$\setv{\B} = \{(v,v) \in \atomterm{B}\}$$
where $\atomterm{T}  = \{(t_1,t_2) | t_1 : \tpesub(T) \land  t_2 : \tpesub(T)\}$.
Two functions are related if given two related argument the application yield related computations:
$$\setv{T_1 -> T_2} = \{(v_1,v_2) \in \atomterm{T_1 -> T_2} | \forall (v'_1,v'_2) \in \setv{T_1}. (v_1\;v'_1, v_2\;v'_2) \in \sett{T_2} \}$$
Two type abstractions are related if their instantiations to two arbitrary types yield related computations for any given relation between the instantiated types:
\begin{multline*}
\setv{\forall X. T} = \{(v_1,v_2) \in \atomterm{\forall X. T} | \\
\forall T_1,T_2, \forall R \in \mathsf{Rel}[T_1,T_2]. (v_1\;[T_1], v_2\;[T_2]) \in \sett[\rho,X \mapsto (T_1,T_2,R)]{T} \}
\end{multline*}
This relation allows us to relate values at abstract types: two values are related at an abstract type $X$, if they are in the relation for $X$:
$$\setv{X} = R \text{ where } \rho(X) = (T_1,T_2,R)$$
Finally, two computations are related if they reduce to two related values.
$$\sett{T} = \{(t_1,t_2)  \in \atomterm{T} | t_1 \red^{*} v_1 => (t_2 \red^{*} v_2 \land (v_1, v_2) \in \setv{T}) \}$$

This definition of parametricity for the statically-typed polymorphic lambda calculus is standard and uncontroversial. Conversely, parametricity for gradual languages---or {\em gradual parametricity}---is a novel concept around which different efforts have been developed, yielding different notions. The subtle differences in interpretation come from the specificities of gradual typing, namely the potential for runtime errors due to type imprecision. Much of it is linked to the mechanism used to enforce type abstraction.
Apart from \gsf, gradual parametricity has only be proven for \lamB~\cite{ahmedAl:icfp2017} and \polyG~\cite{newAl:popl2020}, under two fairly different interpretations. Technically, both are defined using logical relations that are fairly standard, except for three important cases: polymorphic types, type variables, and of course, the unknown type.
We now briefly review and compare both approaches.

\myparagraph{Gradual parametricity in \lamB}
Building on prior work by~\citet{matthewsAhmed:esop2008}, \lamB uses runtime type generation to reduce type applications, and a form of automatic (un)sealing is introduced via casts and type names during reduction. For instance, the casted value $1: \Int \Rightarrow \alpha$ represents an integer sealed with type name $\alpha$.
Two values are related at a type name $\alpha$, if both values are casted to $\alpha$ and belong to the relation of $\alpha$. 
Because type names have indefinite dynamic extent, the logical relation is now ternary, adding {\em worlds} that hold (among other data) the association between type names and the chosen relation between instantiation types. For instance $\setv{\B} = \{(W, v,v) \in \atomterm{B}\}$.
The lexical environment $\rho$ now binds type variables to type names.
For instance, $(W, 1: \Int \Rightarrow \alpha, 2: \Int \Rightarrow \alpha) \in \setv[\rho]{\alpha} = \setv[\rho]{X}$, when $\rho(X)=\alpha$
and $W(\alpha).R = \{(1,2)\}$.

As sealing and unsealing is introduced automatically at runtime, to reason parametrically about type abstractions, \lamB does not directly relate type applications as computations. 
For instance, consider term $\Lambda X.\lambda x: X. x$, which is related with itself $(W, \Lambda X.\lambda x: X. x, \Lambda X.\lambda x: X. x) \in$  
\mbox{$\setv[\emptyset]{\forall X. X -> X}$}. 
If we instantiate these values with $\Int$, choosing relation $\{(1,2)\}$, then 
as  $W.\Sigma_i \triangleright (\Lambda X.\lambda x: X. x)\;[\Int] \red  W.\Sigma_i, \alpha := \Int \triangleright (\lambda x: \alpha. x) : \alpha -> \alpha \Rightarrow \Int -> \Int$,
according to the classic definition of parametricity, it must be the case that
$$(W', (\lambda x: \alpha. x) : \alpha -> \alpha \Rightarrow \Int -> \Int, (\lambda x: \alpha. x) : \alpha -> \alpha \Rightarrow \Int -> \Int) \in \setv[\rho]{X-> X}$$
for a future world $\W'$ (notation $\W' \futureW \W$) such that $W'(\alpha).R = \{(1,2)\}$.
The problem here is that
because $(W', 1: \Int \Rightarrow \alpha,2: \Int \Rightarrow \alpha) \in \setv[\rho]{X}$, then the following applications should be related computations
\begin{align*}
& (W', ((\lambda x: \alpha. x) : \alpha -> \alpha \Rightarrow \Int -> \Int) \;(1: \Int \Rightarrow \alpha),\\ 
& ~((\lambda x: \alpha. x): \alpha -> \alpha \Rightarrow \Int -> \Int)\;(2: \Int \Rightarrow \alpha)) \in \sett[\rho]{X}
\end{align*}
But these application expressions do not type check! Furthermore, reducing $(1 : \Int \Rightarrow \alpha) : \Int \Rightarrow \alpha$ would always yield an error.

Therefore, instead of relating the two type application expressions as computations, \lamB relates only the bodies of the type abstractions {\em after} the type applications have been performed (highlighted in gray), in an extended world:
\begin{multline*}
\setv{\forall X. \cT} = \{(W, v_1,v_2) \in \atomterm{\forall X. \cT} | \forall \cT_1,\cT_2, \forall R \in \mathsf{Rel}[\cT_1,\cT_2]. \forall W' \futureW W \\
W'.\Sigma_1 \triangleright 
v_1\;[\cT_1] \red 
W'.\Sigma_1, \alpha: \cT_1 \triangleright  
(t_1 : \_ \Rightarrow \_)  \;\land\\
W'.\Sigma_2 \triangleright 
v_2\;[\cT_2] \red 
W'.\Sigma_2, \alpha: \cT_2 \triangleright 
(t_2 : \_ \Rightarrow \_) \;\land\; \\
(W' \extworld (\alpha, \cT_1, \cT_2,  R), \Gbox{t_1, t_2}) \in \sett[\rho \lcorchete X \mapsto \alpha \rcorchete]{\cT} \}
\end{multline*}
After both type applications take a step, only the inner terms $t_1$ and $t_2$ are related. Observe how this definition ``strips out'' the outermost casts (whose source and target types are omitted here). These casts automatically realize (un)sealing.

This technique makes it possible to reason about two values that expect already-sealed pair of related values. In the previous example, we know that if 
$W.\Sigma_i \triangleright (\Lambda X.\lambda x: X. x)\;[\Int] \red 
W.\Sigma_i, \alpha := \Int \triangleright (\lambda x: \alpha. x) :$ 
\mbox{$\alpha -> \alpha \Rightarrow \Int -> \Int$}
then 
$(W', \lambda x: \alpha. x, \lambda x: \alpha. x) \in \setv[\rho]{X -> X}$, where $W'$ is the extended future world. Therefore we can deduce that
$(W', (\lambda x: \alpha. x)\;(1 : \Int \Rightarrow \alpha), (\lambda x: \alpha. x)\;(2 : \Int \Rightarrow \alpha)) \in \sett[\rho]{X}$.

To reason about related type applications as computations, and not by considering the inner terms only, one needs to use a {\em compositionality} lemma. This lemma relates two values after the unsealing of some type name $\alpha$. It is important to notice that to apply this lemma, $\alpha$ must be {\em synchronized} in $W$, \ie~bound to the same type on both sides ($\cT_1=\cT_2=\cT$) and to the value relation of that type ($W(\alpha).R = \setv[\rho]{\cT}$).
This synchronization requirement is needed to prevent the unsealing of unrelated values such as $(W', 1 : \Int \Rightarrow \alpha, 2 : \Int \Rightarrow \alpha)$. Otherwise, after unsealing, we would have $(W', 1, 2) \in \setv[\emptyset]{\Int}$, which is clearly false. 

\myparagraph{Gradual parametricity in \polyG} \citet{newAl:popl2020} recently developed another approach to gradual parametricity, which has the benefit of avoiding the convoluted treatment of type applications described above, and doing without type names altogether. In doing so, the notion of gradual parametricity they present is faithful to Reynolds original presentation. Note however that this comes at a cost: the {\em syntax} of \polyG departs importantly from \sysF, by requiring all sealing and unsealing to happen explicitly in the term syntax, with outward scoping of type variables:
$$
\framebox{\text{\sysF}}\; ((\Lambda X. \lambda x:X. x)\;[\Int]\;1)+1 \qquad
\framebox{\text{\polyG}}\; \mathsf{unseal}_X ((\Lambda X. \lambda x:X. x)\;[X=\Int]\;(\mathsf{seal}_X 1)) + 1
$$

Technically, gradual parametricity for \polyG is established by first translating \polyG to an intermediate language \polyC and finally to \CBPV, a variant of Levy's Call-by-Push-Value~\cite{levy:tlca1999} with open sums to encode the unknown type.
The logical relation of parametricity is defined for \CBPV, and differs importantly from that of \lamB. In fact, it is very close to the standard presentation used in static parametric languages. In particular, to relate type abstractions,
the definition requires type applications (and not some inner terms) to be related as computations, as expected in the standard treatment of parametricity. 
Crucially, this is possible only because type application never incurs automatic insertion of casts to seal/unseal values, as would happen in \lamB, 
because in this approach sealing and unsealing are explicit in the syntax of terms.

Let us revisit the example discussed above for \lamB, now in \polyG. Extrapolating the logical relation to be defined over \polyG directly, we know that
$(\Lambda X.\lambda x: X. x, \Lambda X.\lambda x: X. x) \in \setv{\forall X. X -> X}$. Therefore after type application, and following the definition of the relation for functions, we know that
$((\lambda x: \Int. x)\; 1,
(\lambda x: \Int. x)\; 2) \in \sett[\rho]{X}
$.
The resulting values are related values because $(1,2) \in \setv[\rho]{X}$.

\myparagraph{Comparing parametricities}
The notion of gradual parametricity of \polyG is stronger than that of \lamB, as it directly embodies the kind of parametric reasoning that one is used to in static languages. While \lamB ensures a form of gradual parametricity, this notion is weaker, because given two related type abstractions and two arbitrary (possibly different) types, we cannot directly reason about both corresponding type applications: we can only directly reason about the body of the type abstractions after the type applications have reduced. 

While the notion of gradual parametricity of \polyG is superior, as already mentioned, it is enabled by sacrificing the syntax of \sysF. In this work, we are interested in gradualizing \sysF, and studying the properties we can get, rather than in designing a different static source language in order to accommodate the desired reasoning principles. This led us to embrace runtime sealing through type names, as in \lamB, and consequently, to aspire to a weaker notion of gradual parametricity than that of \polyG.
We do believe that both approaches are fully valuable and necessary to understand the many ways in which gradual typing can embrace such an advanced typing discipline.

In particular, as illustrated by \citet{newAl:popl2020}, the weaker notion of parametricity adopted in \gsf can lead to behavior that breaks the (strong notion of) parametricity enjoyed by \polyG. Note that this can however only occur when manipulating values of {\em imprecise} polymorphic types; for values of static types, the reasoning principles of standard parametricity do apply. 
The example they present starts from the value
$v \triangleq (\Lambda X. \lambda x: X. \ttt) :: \forall X. ? -> \Bool$. Although $v$ is related to itself at type  $\forall X. ? -> \Bool$, 
two {\em different} instantiations (to $\Int$ and $\Bool$, respectively) are not related computations, \ie $(W, v\;[\Int], v\;[\Bool]) \not\in \sett{\? -> \Bool}$.
This is because, given two related arguments at type $\?$ such as $\ev[\Int]1::\?$ twice,
$v\;[\Int]\;(\ev[\Int]1::\?)$ reduces to $\ttt$, whereas 
$v\;[\Bool]\;(\ev[\Int]1::\?)$ reduces to an error. 
The logical relation only tells us that after instantiation, terms
$\ev[\richtype{\Int} -> \Bool](\lambda x: \alpha. \ttt) :: \? -> \Bool$
and 
$\ev[\richtype{\Bool} -> \Bool](\lambda x: \alpha. \ttt) :: \? -> \Bool$
are indeed related at type $\? -> \Bool$.
In this case, if we try to apply both functions to $\ev[\Int]1::\?$, both programs fail.
The only arguments that can be passed such that both applications succeed are related 
sealed values at type $\?$, such as 
$(W, \cast{\pr{\Int, \richtype{\Int}}}{1}::\?, \cast{\pr{\Int, \richtype{\Bool}}}{\ttt}::\?) \in \setv{\?}$ 
(assuming an appropriate relation for $\alpha$). Note that we cannot use the compositionality lemma to relate the type applications as computations, because $\alpha$ is not synchronized in this case (it is bound to $\Int$ in one term and to $\Bool$ in the other). The compositionality lemma can only be used to related terms such as $t\;[\cT]$ with themselves.

Finally, note that the fact that $v\;[\Bool]\;(1::\?)$ reduces to an error in \gsf
points to a wider point in the design space of gradually-typed languages: how {\em eagerly} should type constraints be checked? Indeed, $v\;[\Bool]$ is $\lambda x:\Bool.\ttt$, whose application to an underlying $\Int$ value is ill-typed and can legitimately be expected to fail. In that respect, \gsf follows GTLC~\cite{siekTaha:sfp2006,siekAl:snapl2015}, in which $(\lambda x:\Bool.\ttt)\;(1::\?)$ also fails with a runtime cast error. This eager form of runtime type checking likewise follows from the Abstracting Gradual Typing methodology as formulated by \citet{garciaAl:popl2016}. An interesting perspective would be to study a lazy variant of AGT, and whether it recovers properties of alternative approaches~\cite{newAhmed:icfp2018}.

It is interesting to observe that no runtime error is raised in \lamB for this example, despite the fact that the parametricity logical relation is essentially the same as that of \gsf. The difference comes from the runtime semantics of \lamB: as we have illustrated in \S\ref{sec:design}, \lamB does not track the type instantiations that occur on imprecise types. This means that the underlying typing violation observed by \gsf, which manifests as a runtime error, is not noticed in \lamB. Therefore, this example highlights yet another point of tension in the design space of \sysF-based gradual languages.

\subsection{Gradual Parametricity in \gsf}
\label{sec:parametricity-in-GSF}

\begin{figure}[hp]
\begin{small}
    \begin{flalign*}
    \begin{array}{@{}>{\displaystyle}l@{}>{\displaystyle{}}l@{}}
\setv{\basetype} &=\quad \{\lgrpp{\W}{v}{v}[\basetype]\} \\
\lgrvm{\W}{v_1}{v_2}{\cT_1 -> \cT_2}{
        \forall \W' \futureW \W.
        \forall v'_1, v'_2.\\ 
        &\qquad \quad (\W', v'_1, v'_2) \in  \setv{\cT_1} => 
       \lgrp{\W'}{{v_1}\; {v'_1}}{{v_2}\; {v'_2}} \in \sett{\cT_2}
      }\\
\lgrvm{\W}{v_1}{v_2}{\cT_1 \pairsy \cT_2}{\\ 
        & \qquad \quad
       \lgrp{\W}{\tproj{1}[v_1]}{\tproj{1}[v_2]} \in \sett{\cT_1} \mland   
        \lgrp{\W}{\tproj{2}[v_1]}{\tproj{2}[v_2]} \in \sett{\cT_2}
      }\\ 
\lgrvm{\W}{v_1}{v_2}{\forall X.\cT}{
        \forall \W' \futureW \W. \forall t_1, t_2, \cT_1, \cT_2, \alpha, \ev[1], \ev[2].\\ 
        & \qquad \quad  \forall R \in \reln[{\getIdxp}]{\cT_1}{\cT_2}. \\
        & \qquad \qquad \quad 
        (\twfsimpl{\getStorep[1]}{\cT_1} \land \twfsimpl{\getStorep[2]}{\cT_2}  \land\\ 
        & \qquad \qquad \quad
        \;\;\storeeval[\getStorep[1]] {v_1}  [\cT_1] \red \storeeval[\getStorep[1],\alpha := \cT_1] \cast{\ev[1]}{t_1} :: \tpesub(\cT) [\cT_1 / X] \mland \\
        & \qquad \qquad \quad
        \;\;\storeeval[\getStorep[2]] {v_2}  [\cT_2] \red \storeeval[\getStorep[2],\alpha := \cT_2] \cast{\ev[2]}{t_2} :: \tpesub(\cT) [\cT_2 / X])  => \\
        & \qquad \qquad \quad 
        \downstep \lgrp{\W' \extworld (\alpha, \cT_1, \cT_2,  R)}{t_1}{t_2} \in  \sett[\tpesub \lcorchete X \mapsto \alpha \rcorchete]{\cT}
  }\\
\setv{X} &=\quad \setv{\tpesub(X)}\\
\lgrvmalpha{\W}{\cast{\newev{\cE_{11}, \richtype{\cE_{12}}}[\key_1 ][\lock_1\cup \{\richtype{\cE_{12}}\}]}{u_1} :: \alpha}{\cast{\newev{\cE_{21}, \richtype{\cE_{22}}}}{u_2} :: \alpha}{\alpha}{ \\
      & \qquad \quad
      \lgrp{\W}{\cast{\newev{\cE_{11}, \cE_{12}}[\key_1][\lock_1]}{u_1} :: \getStore[1](\alpha)}{\cast{\newev{\cE_{21}, \cE_{22}}[\key_2][\lock_2]}{u_2} :: \getStore[2](\alpha)} \in \getRel(\alpha)}\\
\lgrvmalpha{\W}{\cast{\ev[1]}{u_1} :: \?}{
        \cast{\ev[2]}{u_2} :: \?}{\?}{ 
        \unknowify{\pi_2(\ev[i])} = \cT~\land  
        \\
    & \qquad \quad
        (\W, \cast{\ev[1]}{u_1} :: \cT, \cast{\ev[2]}{u_2} :: \cT) \in \setv{\cT} 
      }\\[0.2em]
\hline\\[-2.3ex]
\lgrtm{\W}{t_1}{t_2}{\cT}{
        \forall i < \getIdx, (\forall \Wstore_1, v_1. \ \storeeval[\getStore[1]] t_1 \red^i \storeeval[\Wstore_1] v_1 =>  \\
        & \qquad \quad
        \exists \W' \futureW  \W , v_2. \ \storeeval[\getStore[2]] t_2 \red^{*}\storeeval[\getStorep[2]] v_2 \land \getIdxp + i = \getIdx \land {} \\
        & \qquad \quad
        \getStorep[1] = \Wstore_1 \land \lgrp{\W'}{v_1}{v_2} \in \setv{\cT}) \mland \\
        & \qquad \quad
        (\forall \Wstore_1. \storeeval[\getStore[1]] t_1 \red^i \error => \exists \Wstore_2. \storeeval[\getStore[2]] t_2 \red^{*} \error) 
        }\\[0.2em]
\hline\\[-2.3ex]
\sets[\emptyenv] & = \quad \world\\
\sets[\Wstore, \alpha := \cT] & = \quad 
        \begin{block}
        \sets \cap \{\W \in \world| \getStore[1](\alpha) = \cT \land \getStore[2](\alpha) = \cT \mland \\ |- \getStore[1] \land |- \getStore[2] \land \getRel(\alpha) = \acotrel{\setv[\emptyset]{\cT}}{\getIdx}\}\\[0.2em]
        \end{block}\\
\hline\\[-2.3ex]
\setd[\emptyenv] & = \quad \set{(\W, \emptyset)| \W \in \world}\\
\setd[\Delta, X] & = \quad \set{(\W, \tpesub[X \mapsto \alpha])| (\W, \tpesub) \in \setd \land \alpha \in \dom(\W.\krel)}\\[0.2em]
\hline\\[-2.3ex]
\setg{\emptyenv} & = \quad \set{(\W, \emptyset)| \W \in \world}\\
\setg{\Gamma, x:{\cT}} & = \quad \set{(\W, \gamma[x \mapsto (v_1, v_2)])| (\W, \gamma) \in \setg{\Gamma} \land (\W, v_1, v_2) \in \setv{\cT}}\\[0.2em]
\hline\\[-2.3ex]
\logaproxg{t_1}{t_2}{\cT} & \delequal \quad 
        \begin{block}
       \staticgJ{t_1} \inTermT{\cT} \land \staticgJ{t_2} : \TermT{\cT} \land  \forall \W \in \sets, \tpesub, \gamma. \\
        ( (\W, \tpesub) \in \setd \land (\W, \gamma) \in \setg{\Gamma})  => (\W, \tpesub(\gamma_1(t_1)), \tpesub(\gamma_2(t_2))) \in \sett{\cT}
        \end{block}\\
\logeqg{t_1}{t_2}{\cT}  & \delequal \quad \logaproxg{t_1}{t_2}{\cT} \land \logaproxg{t_2}{t_1}{\cT}\\[0.2em]
\hline
\end{array} 
\end{flalign*}

\begin{align*}
\atom{\cT_1}{\cT_2} = & \{ (\W, t_1, t_2)| \getIdx < n \land \W \in \world  \land \staticgJ[\getStore[1], \cdot, \cdot]{t_1}  \inTermT{\cT_1} \land \staticgJ[\getStore[2], \cdot, \cdot]{t_2}  \inTermT{\cT_2}\} \\
\atomv{\cT_1}{\cT_2} = & \{(\W, v_1, v_2) \in \atom{\cT_1}{\cT_2}\}\quad\quad
\atomterm{\cT} =   \displaystyle\cup_{n \geq 0}\{ (\W, t_1, t_2) \in \atom{\tpesub(\cT)}{\tpesub(\cT)}\}\\
\atomvalue{\cT} = & \{ (\W, v_1, v_2) \in \atomterm{\cT}| 
\unlift{\pi_2(\getev{v_1})} = \unlift{\pi_2(\getev{v_2})} \}\\
\world = & \displaystyle\cup_{n \geq 0}{\worldn}\\
\worldn = & \{(j, \store_1, \store_2, \krel) \in \nat \times \tnstore \times \tnstore \times (\TypeName -> \relj)| \\ &
j < n \ \land |- \store_1 \ \land |- \store_2 \land \forall \alpha \in \dom(\krel). \krel(\alpha) \in \reln[j]{\store_1(\alpha)}{\store_2(\alpha)}\}\\
\reln{\cT_1}{\cT_2} = & 
\begin{block}
\{ R \in \atomv{\cT_1}{\cT_2}| \forall(\W, v_1, v_2) \in R. \forall \W' \futureW \W. (\W', v_1, v_2) \in R\}
\end{block}\\
\acotrel{R}{n} = & \set{(W, e_1, e_2) \in R| \getIdx \leq n}\quad\quad
\acotrel{\krel}{n} = \set{\alpha \mapsto \acotrel{R}{n}| \krel(\alpha) = R}\\
\krel' \futureW \krel \delequal & \forall \alpha \in \dom(\krel). \krel'(\alpha) = \krel(\alpha)\\
\W' \futureW \W \delequal &  \getIdxp \leq \getIdx \land \getStorep[1] \supseteq \getStore[1] \land \getStorep[2]\supseteq \getStore[2] \land \getRelp \futureW \acotrel{\getRel}{\getIdxp} \land \W', \W \in \world\\
\downstep W = & (j, \getStore[1], \getStore[2], \acotrel{\getRel}{j}) \quad \text{where } j = \getIdx[j]-1
\end{align*}
\end{small}
\vspace{-1em}
\caption{Gradual logical relation and auxiliary definitions}
\label{fig:glgrv2}\label{fig:lgr}
\end{figure}

We now turn to the technical details of gradual parametricity in \gsf. As explained above, we follow \lamB~\cite{ahmedAl:icfp2017} for our notion of gradual parametricity, due to the use of runtime type name generation for sealing, and the \sysF syntax that requires automatic insertion of (un)sealing evidences at runtime. 

We establish parametricity for \gsf by proving parametricity for \gsfe.
Specifically, we define a {\em step-indexed} logical relation for \gsfe terms, closely following the relation for \lamB. In the following, we focus on the few differences with the \lamB relation, essentially dealing with evidences.
The relation (Figure~\ref{fig:glgrv2}) is defined on tuples $(\W, t_1, t_2)$ that denote two related terms $t_1,t_2$ in a world $\W$. A world is composed of a step index $j$, 
two stores $\store[1]$ and $\store[2]$ used to typecheck and evaluate the related terms, and a mapping $\krel$, which maps type names to relations $R$, used to relate sealed values. The components of a world are accessed through a dot notation, \eg $W.j$ for the step index. 

The interpretations of values, terms, stores, name environments, and type environments are mutually defined, using the auxiliary definitions at the bottom of Figure~\ref{fig:glgrv2}. As usual, the value and term interpretations are indexed by a type and a type substitution $\tpesub$.
We use $\atom{\cT_1}{\cT_2}$ to denote a set of pair of terms of type $\cT_1$ and $\cT_2$, and worlds with a step index less than $n$. We write $\atomv{\cT_1}{\cT_2}$ to restrict that set to values, and $\atomterm{\cT}$ to denote a set of terms of the same type after substitution. The $\atomvalue{\cT}$ variant is similar to $\atomv{\cT_1}{\cT_2}$ but restricts the set to values that have, after substitution, equally precise evidences (the equality is after unlifting because two sealed values may be related under different instantiations). 
$\reln{\cT_1}{\cT_2} $ defines the set of relations of values of type $\cT_1$ and $\cT_2$. We use $\acotrel{R}{n}$ and $\acotrel{\krel}{n}$ to restrict the step index of the worlds to less than $n$. Finally, $\krel' \futureW \krel$ specifies that $\krel'$ is a future relation mapping of $\krel$ (and extension), and similarly $\W' \futureW \W$ expresses that $\W'$ is a future world of $\W$. The $\downstep$ operator lowers the step index of a world by 1.

The logical interpretation of terms of a given type enforces a ``termination-sensitive'' view of parametricity: if the first term yields a value, the second must produce a related value at that type; if the first term fails, so must the second. Note that 
$\atomvalue{\cT}$ requires the second component of the evidence of each value to have the same precision in order to enforce such sensitivity. Indeed, if one is allowed to be more precise than the other, then when later combined in the same context, the more precise value may induce failure while the other does not. 

Two base values are related if they are equal. Two functions are related if their application to related values yields related results. Two type abstractions are related if given any two types and any relation between them, the instantiated terms (without their unsealing evidence) are also related in a world extended ($\extworld$) with $\alpha$, the two instantiation types $\cT_1$ and $\cT_2$ and the chosen relation $R$ between sealed values. Note that the step index of this extended world is decreased by one, because we take a reduction step.
Two pairs are related if their components are pointwise related. Two sealed values are related at a type name $\alpha$ if, after unsealing, the resulting values are in the relation corresponding to $\alpha$ in the current world, $\getRel(\alpha)$.

Finally, two values are related at type $\?$ if they are related at the least-precise type with the same top-level constructor as the second component of the evidence, $\unknowify{\pi_2(\ev[i])}$. The function $\unknowifys$ extracts the top-level constructor of an evidence type, \eg:
$$\unknowify{\cE_1 -> \cE_2}=\? -> \? \qquad \unknowify{\forall X. \cE}=\forall X. \?$$
The intuition is that to be able to relate these unknown values we must take a step towards relating their actual content; evidence necessarily captures at least the top-level constructor (\eg~if a value is a function, the second evidence type is no less precise than $\? -> \?$, \ie~$\unknowify{\cE_1 -> \cE_2}$).

The logical relation is well-founded for two reasons: {\em (i)} in the $\?$ case, $\unknowify{\pi_2(\ev[i])}$ cannot itself be $\?$, as just explained; {\em (ii)} in each other recursive cases, the step index is lowered: for functions and pairs, the relation is between reducible expressions (applications, projections) that either take a step or fail; for type abstractions, the relation is with respect to a world whose indexed is lowered.

The interpretations of stores, type name environments and type environments are straightforward (Figure~\ref{fig:glgrv2}). The logical relation allows us to define logical approximation, whose symmetric extension is logical equivalence. Any well-typed \gsfe term is related to itself at its type:

\begin{restatable}[Fundamental Property]{theorem}{gfunprop}
\label{theorem:gfunprop}
If \;$\EnvSG t \inTermT{\cT}$ then $\logaproxg{t}{t}{\cT}$.
\end{restatable}
As standard, the proof of the fundamental property uses compatibility lemmas for each term constructor and the compositionality lemma\iffullv{ (\S\ref{asec:LRProperties})}.
Almost every compatibility lemma relies on the fact that the ascription of two related values yield related terms.
\begin{restatable}[Ascriptions Preserve Relations]{lemma}{ascriptionlemmapaper}
\label{prop:ascriptonlemmapaper}
If
  $(\W, v_1, v_2) \in \setv{\cT}$,
  $\ijudgment{\ev}{\cT}{\cT'}$, 
  $\W \in \sets$, and $(\W, \tpesub) \in \setd$,
 then
 $(\W, \tpesubSI[1]{\ev} v_1 :: \tpesub(\cT'),\tpesubSI[2]{\ev} v_2 :: \tpesub(\cT')) \in \sett{\cT'}$.
\end{restatable}
Note that type substitution on evidences takes as parameter the corresponding store: $\tpesubSI[i]{\ev}$ is syntactic sugar for $\tpesub(\getStore[i],\ev)$, lifting each substituted type name in the process, \eg~if $\tpesub(X) = \alpha$, $\getStore[1](\alpha) = \Int$, and $\getStore[2](\alpha) = \Bool$, then
$\tpesubSI[1]{\pr{X,X}} = \pr{\richtype{\Int}, \richtype{\Int}}$, and $\tpesubSI[2]{\pr{X,X}} = \pr{\richtype{\Bool}, \richtype{\Bool}}$.

\subsection{Parametricity vs. the \gdgg in \gsf}
\label{sec:gdd-violation}

We now give a different perspective from that presented in \S\ref{sec:dgg-violation}, regarding the violation of the general dynamic gradual guarantee (DGG, stated with respect to $\gprec$). More precisely, we show that the definition of parametricity for \gsf (\S\ref{sec:gsf-param}) is incompatible with the \gdgg. To do so, we again prove that there exists two terms in \gsf, related by precision, whose behavior violates the DGG, but this time we do so with a proof of the intermediate results that is fully-driven by the definition of parametricity, and not by monotonicity of consistent transitivity.
We present the proof sketch of the intermediate results in order to highlight the key properties that imply this incompatibility. This is particularly relevant because these properties also manifest in \lamB, for which the DGG has not been formally explored yet. 

Recall from \S\ref{sec:dgg-violation} that the term that helps us establish the violation of the DGG is  $\idu$, a variant of the polymorphic identity function $\idx$ whose term variable $x$ is given the unknown type. This term always fails when fully applied.

\begin{restatable}{lemma}{idfails}
\label{theorem:id?fails}
For any $\staticgJ[ ]  v : \?$ and $|- {\cT}$, we have $\gsfreds{(\Lambda X. \lambda x:\?.x::X) \;[\cT] \; v}{\error}$.
\end{restatable} 
\begin{proof}
Let $\idu \triangleq \Lambda X. \lambda x:\?.x::X$, $\staticgJ[ ] \idu  \translate v_{\forall} : \forall X. \? -> X$, and $v$ s.t. $\staticgJ[ ]  v \translate v_\? : \?$. 

By the fundamental property (Th.~\ref{theorem:gfunprop}), 
$\logaproxg[]{v_{\forall}}{v_{\forall}}{\forall X. \? -> X}$ so for any $\W_0 \in \sets[\emptyenv]$, $(\W_0, v_{\forall}, v_{\forall}) \in \sett[\emptyset]{\forall X. \? -> X}$. Because $v_{\forall}$ is a value, $(\W_0, v_{\forall}, v_{\forall}) \in \setv[\emptyset]{\forall X. \? -> X}$.
By reduction, $\conf[\cdot]{v_{\forall} \; [\cT_i]} \red^{*}\conf[\store'_i]{\cast{\evp[i]}{v_i} :: \? -> \cT_i}$ for some $ \evp[i]$, $\ev[i]$ and $\ev[i\alpha]$, where $\store'_i = \{ \alpha := \cT_i\}$ and $v_i = \cast{\ev[i]}{(\lambda x: \?. \cast{(\ev[i\alpha]}{x} :: \alpha))} :: \? -> \alpha$.
We can instantiate the definition of $\setv[\emptyset]{\forall X. \? -> X}$
with $\W_0$, $\cT_1 = \cT$ and $\cT_2$ structurally different (and different from $\?$), some $R \in \reln[{\getIdxg[\W][0]}]{\cT_1}{\cT_2}$, $v_1$, $v_2$, $\evp[1]$ and $\evp[2]$, then we have that 
$(\W_1, v_1, v_2) \in \sett[X \mapsto \alpha]{\? -> X}$, where $\W_1 = (\downstep (\W_0 \extworld (\alpha, \cT_1, \cT_2,  R))$. As $v_1$ and $v_2$ are values, $(\W_1, v_1, v_2) \in \setv[X \mapsto \alpha]{\? -> X}$.
Also, by associativity of consistent transitivity, the reduction of $\conf[\store'_i]({\cast{\evp[i]}{v_i} :: \? -> \cT_i}) \; v_{\?}$ is equivalent to that of $\conf[\store'_i]{\cast{\cod(\evp[i])}{(v_i \; (\cast{\dom(\evp[i])}{v_{\?}} :: \?))} :: \cT_i}$.

By the fundamental property (Th.~\ref{theorem:gfunprop}) we know that $\logaproxg[]{v_\?}{v_\?}{\?}$; we can instantiate this definition with $\W_0$, and we have that $(\W_0, v_\?, v_\?) \in \setv[\emptyset]{\?}$. 
By Lemma~\ref{lm:incompatgeneral},
$(\W_1, \cast{\dom(\evp[1])}{v_{\?}} :: \? , \cast{\dom(\evp[2])}{v_{\?}} :: \?) \in \sett[X \mapsto \alpha]{\?}$.
If $\cast{\dom(\evp[1])}{v_{\?}} :: \?$ reduces to $\error$ then the result follows immediately. 
Otherwise, $\conf[\store'_i]{\cast{\dom(\evp[1])}{v_{\?}} :: \?} \red^{*} \conf[\store'_i]{v''_i}$, and $(\W_2, v''_1, v''_2) \in \setv[X \mapsto \alpha]{\?}$, where $\W_2 =  \downstep \W_1$, and some $v''_1$ and $v''_2$.
We can instantiate the definition of $\setv[X \mapsto \alpha]{\? -> X}$ with $\W_2$, $v''_1$ and $v''_2$, obtaining that $(\W_2, v_1 \; v''_1, v_2 \;  v''_2) \in \sett[X \mapsto \alpha]{X}$.
We then proceed by contradiction. Suppose that $\conf[\store'_i]{v_i \; v''_i} \red^{*} \conf[\store''_i]{v'_i}$ (for a big-enough step index).
If $v''_i  = \cast{\ev''_{iv}}{u} :: \?$, then by evaluation $v'_i = \cast{\evp[iv]}{u} :: \alpha$, for some $\evp[iv]$.
But by definition of $\setv[X \mapsto \alpha]{X}$, it must be the case that for some $\W_3 \futureW \W_2$, $(\W_3, \cast{\evp[1v]}{u} :: \cT_1, \cast{\evp[2v]}{u} :: \cT_2) \in R$, which is impossible because $u$ cannot be ascribed to structurally different types $\cT_1$ and $\cT_2$. Therefore $v_1 \; v''_1$ cannot reduce to a value, and hence the term ${v_{\forall}\;[\cT] \; v_\?}$ cannot reduce to a value either. Because $v_{\forall}$ is non-diverging, its application must produce \error.
\end{proof}

The proof above uses a (rather technical) lemma, which crisply captures when Lemma~\ref{theorem:id?fails} holds in a more general setting.
Intuitively, if we consider one step of execution of the application of
$(\Lambda X. \lambda x:\?.t)$ to two different arbitrary types, and the resulting outermost evidences/casts,
then the pair of any $v'$ ascribed to 
the domains of each evidence/cast are related computations at type $\?$.

\begin{restatable}{lemma}{incompatgeneral}
\label{lm:incompatgeneral}
Let $|- (\Lambda X. \lambda x:\?.t) \translate v_\forall : \forall X. \? -> X$ and
$|- v \translate v_\? : \?$.
For any $\cT_1$ and $\cT_2$, such that $\unknowify{\cT_1} \not= \unknowify{\cT_2}$,
if 
$\conf[\cdot]{v_\forall \; [\cT_i]} \red
\conf[\alpha := \cT_i]{\cast{\ev_i}{v_i} :: \? -> \cT_i}$, $\isjudgment{\ev[i]}{\? -> \alpha}{\? -> \cT_i}$\\
then
$\forall \W \in \sets[\emptyenv], \forall R \in \reln[{\getIdx}]{\cT_1}{\cT_2}$,
$\lgrp{\W \extworld (\alpha, \cT_1, \cT_2, R)}{\cast{\invdom(\ev[1])}{v'} :: \?}{\cast{\invdom(\ev[2])}{v'} :: \?} \in  \sett[X \mapsto \alpha]{\?}$
\end{restatable}
\begin{proof}
  Notice that $v_\forall$ has to be of the form $(\cast{\evp}{(\Lambda X. \cast{\evpp}{(\lambda x: \?. t')} :: \? -> X}) :: \forall X. \? -> X)$, where $\evp = \pr{\forall X. \? -> X, \forall X. \? -> X}$ and $\evpp = \pr{\? -> X, \? -> X}$.
  Then $\conf[\cdot]{v_\forall \; [\cT_i]} \red  \cast{\pr{\? -> \evlift{\alpha_i}, \? -> \cE_i}}{t'}$ for some $t'$, where $\evlift{\alpha_i} = \enrich{\alpha}{\alpha \mapsto \cT_i}$ and $\cE_i = \enrich{\cT_i}{\cdot}$.
  We know that $\cdot; \cdot; \cdot|- v_\? : \?$ then as $X \not\in \FTV(v)$, $\cdot; X; \cdot|- v_\? : \?$, therefore by the fundamental property (Thm~\ref{theorem:gfunprop}),
  $\logaproxg[\cdot; X; \cdot]{v_\?}{v_\?}{\?}$, therefore as $\W \in \sets[\emptyenv]$, we can pick
  $\W' = \W \extworld (\alpha, \cT_1, \cT_2, R) \in \sets[\emptyenv]$, and $(\W', X \mapsto \alpha) \in \setd[X]$ and thus conclude that $\lgrp{\W'}{v_\?}{v_\?} \in  \sett[X \mapsto \alpha]{\?}$.
  Now notice that $\invdom(\ev[i]) = \pr{\?, \?}$, but $\ev \trans{} \pr{\?, \?} = \ev$ for any evidence $\ev$, therefore
  $\conf[\alpha := \cT_i]{\cast{\invdom(\ev[i])}{v_\?} :: \?} \red \conf[\alpha := \cT_i]{v_\?}$, then we have to prove that
  $\lgrp{\downstep \W'}{v_\?}{v_\?} \in  \sett[X \mapsto \alpha]{\?}$ which follows directly from the weakening lemma.
\end{proof}

As a consequence of Lemma~\ref{theorem:id?fails}, the dynamic gradual guarantee is violated in \gsf.
\begin{restatable}{corollary}{incompatibility}
\label{theorem:incompat}
There exist $|- t_1 : \cT$ and $t_2 \sqsupseteq t_1$ such that $\gsfreds{t_1}{v}$ and $\gsfreds{t_2}{\error}$.
\end{restatable}
\begin{proof} 
Let $\mathit{id_X} \triangleq \Lambda X. \lambda x: X. x :: X$, and $\mathit{id_\?} \triangleq \Lambda X. \lambda x: \?. x :: X$. By definition of precision, we have $id_X \gprec id_\?$. Let $\staticgJ[ ]{v} \inTermT{\cT}$ and $\staticgJ[ ]{v'} \inTermT{\?}$, such that $v \gprec v'$.
Pose $t_1 \triangleq \mathit{id_X}\; [\cT] \; v$ and $t_2 \triangleq  \mathit{id_\?}\; [\cT] \; v'$. By definition of precision, we have $t_1\gprec t_2$.
By evaluation, $\gsfreds{t_1}{v}$.
But by Lemma~\ref{theorem:id?fails}, $\gsfreds{t_2}{\error}$.
\end{proof}

Interestingly, Lemma~\ref{theorem:id?fails} holds irrespective of the actual choices for representing evidence in \gsfe.
The key reason is the logical interpretation of $\forall X.\cT$. Therefore the incompatibility described here does not apply only to \gsf but to other gradual languages that use similar logical relations such as \lamB: in fact, we have been able to prove that Lemmas~\ref{theorem:id?fails} and~\ref{lm:incompatgeneral} also hold in \lamB (by using casts instead of evidence), so we conjecture that the reason of the incompatibility is the same as in \gsf.

\subsection{Gradual Free Theorems in \gsf}
\label{sec:free-theorems}

The parametricity logical relation (\S\ref{sec:gsf-param}) allows us to define notions of logical approximation ($\preceq$) and equivalence ($\approx$) that are sound with respect to contextual approximation ($\casy$) and equivalence ($\cesy$), and hence can be used to derive free theorems about well-typed \gsf terms~\cite{wadler:fpca89,ahmedAl:icfp2017}. The definitions of contextual approximation and equivalence, and the soundness of the logical relation, are fairly standard\iffullv{ and left to \S\ref{asec:ContextualEquivalence}}.
As shown by \citet{ahmedAl:icfp2017}, in a gradual setting, the free theorems that hold for \sysF are weaker, as they have to be understood ``modulo errors and divergence''. 
\citet{ahmedAl:icfp2017} prove two such free theorems in \lamB. 
However, these free theorems only concern {\em fully static} type signatures. This leaves unanswered the question of what {\em imprecise} free theorems are enabled by gradual parametricity. To the best of our knowledge, this topic has not been formally developed in the literature so far, despite several claims about expected theorems, exposed hereafter. 

\citet{igarashiAl:icfp2017} report that the \sysF polymorphic identity function, if allowed to be cast to $\forall X. \? -> X$, would always trigger a runtime error when applied, suggesting that functions of type $\forall X. \? -> X$ are always failing. Consequently, \sysFg rejects such a cast by adjusting the precision relation (\S\ref{sec:static-issues}). But the corresponding free theorem is not proven.
Also, \citet{ahmedAl:popl2011} declare that parametricity dictates that any value of type $\forall X. X -> \?$ is either constant or always failing or diverging (p.7). This gradual free theorem is not proven either. In fact, in both an older system \citep{ahmedAl:stop2009} and its newest version \citep{ahmedAl:icfp2017}, as well as in \sysFg, casting the identity function to $\forall X. X -> \?$ yields a function 
that returns {\em without errors}, though the returned value is still sealed, and as such unusable (\S\ref{sec:dynamic-issues}). 
The parametricity relation in \gsf does not impose such behavior: it only imposes uniformity of behavior, including failure, of polymorphic terms, which leaves some freedom regarding when to fail. In particular, we show next that ascribing the \sysF identity function to $\forall X. \? -> X$ yields a function that behaves exactly as the identity function (and hence never fails).

We start by defining a small lemma, which states that if we ascribe any value to a less precise type (according to $\gprec$), then the term steps to a less strictly precise value (according to $\sprec$):
\begin{restatable}{lemma}{ascribingharmeless}
\label{prop:ascribingharmeless}
Consider $v$ and $\cT'$ such that $|- v : \cT$, $|- v \translate v''  : \cT$,  and $\cT \gprec \cT'$. Then 
$\gsfreds{v :: \cT'}{v'}$ and $v'' \sprec v'$.
\end{restatable}
\begin{proof}
The proof of this lemma is straightforward using an auxiliary lemma
that states that consistent transitivity between an evidence $\ev$ and the initial evidence between the outermost type of the $\ev$ consistent judgment and a less precise type, has no effect on $\ev$. Formally,
consider $\isjudgment{\ev}{\cT_1}{\cT_2}$, and $\evp$ initial evidence of  $\cT_2 \sim \cT_3$, for $\cT_3$ such that $\cT_2 \gprec \cT_3$, then $\ev \trans{=} \evp = \ev$.
\end{proof}

The \sdgg-related Lemmas~\ref{prop:dggLemma} and~\ref{prop:ascribingharmeless} help us prove that in \gsf types $\forall X. \? -> X$ and $\forall X. X -> \?$ are inhabited by non-constant, non-failing, parametricity-preserving terms. In particular, both types admit the ascribed System F identity function, among many others (for instance, \mbox{$\Lambda X.\lambda x:X.\lambda f: X\!->\!X. f\; x$} of type $\forall X. X\!->\!(X\!->\!X)\!->\!X$ can also be ascribed to $\forall X. X\!->\!\?$). 

We formalize this using the following corollary:
\begin{restatable}{corollary}{corfreetheorem}
\label{prop:corfreetheorem}
Let $t$ and $v$ be static terms such that $|- t : \forall X.T$, $|- v : T'$, and $\gsfreds{t[T']\;v}{v'}$.
\begin{enumerate}
  \item If $\forall X.T \gprec \forall X. X -> \?$ then $\gsfreds{(t :: \forall X. X -> \?) [T']\;v}{v''}$, and $v' \sprec v''$.
  \item If $\forall X.T \gprec \forall X. \? -> X$ then $\gsfreds{(t :: \forall X. \? -> X) [T']\;v}{v''}$, and $v' \sprec v''$.
\end{enumerate}
\end{restatable} 

Such results, which illustrate the motto that ``extrinsic imprecision is harmless'', constitute a valuable compositionality guarantee when embedding fully-static (\sysF) terms in a gradual world.

\myparagraph{Cheap Theorems} The intuition of $\forall X. \? -> X$ denoting always-failing functions is not entirely misguided: in \gsf, this result does hold {\em for a subset} of the terms of that type. This leads us to observe that we can derive ``cheap theorems'' with gradual parametricity: obtained not by looking only at the type, but by also considering the head constructors of a term. For instance:
\begin{restatable}{theorem}{freetheoremone}
\label{theorem:freepropgeneric2}
Let $v \triangleq \Lambda X.\lambda x:\?.t$ for some $t$, such that $\staticgJ[ ] v \inTermT{\forall X. \? -> X}$. Then for any $\staticgJ[ ]  v' \inTermT{\cT}$, we either have $\gsfreds{v \; [\cT] \; v'}{\error}$ or $v \; [\cT] \; v' \divergesy$.
\end{restatable} 

This result is proven by exploiting the gradual parametricity result (Theorem~\ref{theorem:gfunprop}). Note that what makes it a ``free'' theorem is that it holds independently of the body $t$, therefore   
{\em without having to analyze the whole term}. Not as good as a free theorem, but cheap.

\section{Embedding Dynamic Sealing in \gsf}
\label{sec:embedding}
A gradual language is expected to cover a spectrum between two typing disciplines, such as simple static typing and dynamic typing. The static end of the spectrum is characterized by the conservative extension results~\cite{siekAl:snapl2015}, which we have established for \gsf with respect to \sysF (Proposition~\ref{prop:static-eq} and Proposition~\ref{prop:dgequivs}).
The dynamic end of the spectrum is typically characterized by an {\em embedding} from the considered dynamic language to the gradual language~\cite{siekAl:snapl2015}. For instance, in the case of GTLC~\cite{siekTaha:sfp2006}, the dynamic language is an untyped lambda calculus with primitives. 

In this section, we study the ``dynamic end'' of \gsf. Unsurprisingly, \gsf can embed an untyped lambda calculus with primitives. More interestingly, we highlight the expressive power of the underlying type name generation mechanism of \gsf by proving that \gsf can faithfully embed an untyped lambda calculus with {\em dynamic sealing}, \lambdaSeal.  
This language, also known as the cryptographic lambda calculus, was first studied in a typed version by \citet{pierceSumii:2000}, and then untyped \cite{sumiiPierce:popl2004}. One of their objectives was to study whether dynamic sealing could be used in order 
to dynamically impose parametricity via a compiler from \sysF to \lambdaSeal. 
Recently, \citet{devrieseAl:popl2018} prove that such a compiler is not fully abstract, 
\ie~compiled \sysF equivalent terms are not contextually equivalent in \lambdaSeal. Nevertheless, the dynamic sealing mechanism of \lambdaSeal to protect abstract data, and its relation to gradual parametricity, is an interesting question.

In the following, we first present the embedding of a standard dynamically-typed language, called \dynLang, in \gsf. Then, we introduce \lambdaSeal, which extends \dynLang with dynamic sealing constructs. We define an embedding of \lambdaSeal terms into \gsf, and prove that this embedding is semantic preserving. As far as we know, this is the first result regarding the ``dynamic end'' of a gradually-parametric language, establishing a crisp connection with dynamic sealing.

\subsection{Embedding a Dynamically-Typed Language in \gsf}
\label{sec:embed-untyped}

The essence of embedding a dynamically-typed language in a gradual language is to ascribe every introduction form with the unknown type~\cite{siekTaha:sfp2006,siekAl:snapl2015}. For instance the expression $(1\;2)$ from a dynamically-typed language can be embedded as $(1::\?)\;(2::\?)$. Observe that not adding the ascriptions would yield an ill-typed term, as per the conservative extension result with respect to the static typing discipline. 
Let us call \dynLang the dynamically-typed lambda calculus with pairs and primitives. The embedding of \dynLang terms into \gsf is defined as:

\begin{footnotesize}
\begin{center}
\begin{minipage}{0.2\textwidth}
\begin{align*}
  \compileLSA{\const} &= \const :: \? \\
  \compileLSA{\lambda x.t} &= (\lambda x:\?.\compileLSA{t}) :: \?  \\
  \compileLSA{\pair{t_1}{t_2}} &= \pair{\compileLSA{{t_1}}}{\compileLSA{{t_2}}} :: \? \\
  \compileLSA{\op{\vectorOp{t}}} &= \letT{\vectorOp{x}: \?}{{\compileLSA{{\vectorOp{t}}}}}{\op{\vectorOp{x}} :: \?}
  \end{align*}
\end{minipage}\quad\quad
\begin{minipage}{0.2\textwidth}
\begin{align*}
  \compileLSA{x} &= x \\
  \compileLSA{t_1\;t_2} &= \letT{x: \?}{\compileLSA{t_1}}{\letT{y: \?}{\compileLSA{t_2}}{x \; y}}\\
  \compileLSA{\proj{1}[t]} &= \proj{1}[\compileLSA{t}]\\
  \compileLSA{\proj{2}[t]} &= \proj{2}[\compileLSA{t}]
  \end{align*}
\end{minipage}
\end{center}
\end{footnotesize}

The only novelty here with respect to prior work is that the embedding produces application terms in A-normal form in order to ensure that embedded terms behave as expected. For example, the term $\lste{(1 \; \Omega)}$, with $\lste{\Omega = (\lambda x. x\; x) \; (\lambda x. x\; x)}$, diverges in the dynamically-typed language. But if we would embed an application $\compileLSA{t_1\;t_2}$ simply as $ \compileLSA{t_1} \; \compileLSA{t_2}$, the embedded term would fail in \gsf instead of diverging, because evidence combination would detect the underlying type error before reducing the application. Note that this precaution is unnecessary for pairs, because there are no typing constraints between both components.

As we will see, the correctness of embedding \dynLang follows from the correctness of embedding of \lambdaSeal, discussed next.

\subsection{The Cryptographic Lambda Calculus \lambdaSeal}
\label{sec:cryptographic-lambda-calculus}

\begin{figure}[t]
  \begin{small}
  \begin{displaymath}
    \begin{array}{rcll}
    \multicolumn{4}{c}{  
      \lste{x \in \Var, \seal \in \Seal, \heap  \subset \Seal}
      }\\
     \lste{ t } & ::= & \lste{\const} | \lste{\lambda x.t }| \lste{\pair{t}{t}}| \lste{x} | \lste{t\;t} | \lste{\proj{i}[t]} | \lste{\op{\vectorOp{t}}} |  \lste{\sealC{x}{t}} | \lste{\sealedC{t}{t}} | \lste{\unsealC{x}{t}{t}{t}}| \lste{\seal} & \text{(terms)}\\
      \lste{v} & ::= & \lste{\const} | \lste{\lambda x.t}  | \lste{\pair{v}{v}} | \lste{\sealedC{v}{\seal}} | \lste{\seal} & \text{(values)} 
    \end{array}   
  \end{displaymath}
 \end{small}
\begin{small}
\begin{flushleft}
\framebox{\lste{$\storeevalLS{t} \longrightarrow \storeevalLS t$}}
~\textbf{Notion of reduction}
\end{flushleft}
\begin{mathpar}
\lste{\storeevalLS{(\lambda x. t) \;v} \longrightarrow \storeevalLS {t[v/x]}} \and
\lste{\storeevalLS{\proj{i}[\pair{v_1}{v_2}]} \longrightarrow \storeevalLS{v_i}} \and
  \and
  \lste{\storeevalLS{\op{\vectorOp{v}}} \longrightarrow \storeevalLS{\redop{\vectorOp{v}}}} \and
\lste{\storeevalLS{\sealC{x}{t}} \longrightarrow \storeevalLS[\heap, \seal]{t[{\seal}/x]}}\quad \text{where }\lste{\seal \not\in \heap}\and
 \lste{\storeevalLS{\unsealC{x}{\seal}{\sealedC{v}{\seal'}}{t}} \nred}
  \begin{cases}
  \lste{\storeevalLS{t[v/x]} \qquad \;\;\;\;\;\;\;  \seal \equiv \seal'} \\
  \lste{\unsealError \qquad \;\seal \not\equiv \seal'}
  \end{cases} 
\end{mathpar}
\begin{flushleft}
\framebox{\lste{$ \storeevalLS t \longmapsto \storeevalLS t$}}
~\textbf{Evaluation frames and reduction}
\end{flushleft}
\begin{displaymath}
\begin{array}{rcll}
\lste{f} & ::= & \lste{[] \ t} | \lste{ v \ [] } | \lste{\pair{[]}{t}} | \lste{\pair{v}{[]}} | \lste{\proj{i}[[]]} | \lste{\sealedC{{[]}}{t}} | \lste{\sealedC{v}{{[]}}}| \lste{\op{\vectorOp{v}, [] , \vectorOp{t}}} & \text{(term frames)} \\
 &  & | \lste{\unsealC{x}{{[]}}{t}{t}} | \lste{\unsealC{x}{{v}}{[]}{t}} &  
\end{array}
\end{displaymath}
\begin{mathpar}
\inference{\lste{\storeevalLS{t} \nred \storeevalLS[\heap']{t'}}}{\lste{\storeevalLS{t} \red \storeevalLS[\heap']{t'}}}
\and
\inference{\lste{\storeevalLS{ t } \red \storeevalLS[\heap']{ t'}}}{\lste{ \storeevalLS {f[t]} \red \storeevalLS[\heap']{f[t']}}}\\
\inference{\lste{\storeevalLS{t} \nred \unsealError}}{\lste{\storeevalLS{t} \red \unsealError}}\and 
\inference{\lste{\storeevalLS{t} \red \unsealError}}{\lste{\storeevalLS{f[t]} \red \unsealError}}
\end{mathpar}
\end{small}
 \caption{\lambdaSeal: Untyped Lambda Calculus with Sealing}  
  \label{fig:lambdaseal}
  \label{fig:lambdaseal-syntax-statics}  
  \label{fig:lambdaseal-dyn}
\end{figure}

The cryptographic lambda calculus \lambdaSeal is an extension of \dynLang with primitives for protecting abstract data by sealing~\cite{sumiiPierce:popl2004}. 
Figure~\ref{fig:lambdaseal} presents the syntax and dynamic semantics of the \lambdaSeal language we consider here, which is a simplified variant of that of~\cite{sumiiPierce:popl2004}.
In addition to standard terms, which correspond to \dynLang, the \lambdaSeal language introduces four new syntactic constructs dedicated to sealing. First, the term $\lste{\sealC{x}{t}}$ generates a fresh key to seal and unseal values, bound to $\lste{x}$ in the body $\lste{t}$. Seals, denoted by the metavariable $\lste{\seal}$, are values tracked in the set of allocated seals $\lste{\heap}$.
 The sealing construct $\lste{\sealedC{t_1}{t_2}}$ evaluates $\lste{t_1}$ to a value $\lste{v}$ and $\lste{t_2}$ to a seal $\lste{\seal}$, and seals $\lste{v}$ with $\lste{\seal}$. Term $\lste{\unsealC{x}{t_1}{t_2}{t}}$ is for unsealing. At runtime, $\lste{t_1}$ should evaluate to a seal $\lste{\seal}$ and $\lste{t_2}$ to a sealed value $\lste{\sealedC{v}{\seal'}}$. If $\lste{\seal = \seal'}$, unsealing succeeds and $\lste{t}$ is evaluated with $\lste{x}$ bound to $\lste{v}$. Otherwise, unsealing fails, producing a runtime sealing error $\lste{\unsealError}$.\footnote{The original term for unsealing in \lambdaSeal has the syntax $\lste{\unsealC{x}{t_1}{t_2}{t} \;\mathsf{else} \; t_3}$; if the unsealing fails, reduction recovers from error evaluating $\lste{t_3}$. To be able to encode such a construct, we would need to extend \gsf with error handling.}

To illustrate, consider the following term:
$$\lste{\sealC{x}{\sealC{y}{\lambda b.\unsealC{n}{x}{\sealedC{1}{(\mathsf{if} \; b \; \mathsf{then} \; x \; \mathsf{else} \; y)}}{n + 1}}}}$$

This term first generates two fresh seals $\lste{x}$ and $\lste{y}$, and then defines a function that receives a boolean $\lste{b}$ and attempts to unseal a sealed value. The value $\lste{1}$ is sealed using either $\lste{x}$ or $\lste{y}$, depending on $\lste{b}$, and unsealed with $\lste{x}$. If the function is applied to $\lste{\true}$, unsealing succeeds because the seals coincide, and the function returns $\lste{2}$. Otherwise, unsealing fails, and an $\lste{\unsealError}$ is raised.

Overall, we can distinguish three kinds of runtime errors in \lambdaSeal: in addition to unsealing errors, $\lste{\unsealError}$, there are two kinds of runtime {\em type} errors (omitted in Figure~\ref{fig:lambdaseal}), hereafter called \typeError and \sealError. The former corresponds to standard runtime type errors such as applying a non-function, and can happen in \dynLang. The latter is specific to \lambdaSeal, and corresponds to expressions that do not produce seals when expected, such as $\lste{\sealedC{1}{2}}$.

\subsection{Embedding \lambdaSeal in \gsf}
\label{sec:embedding-lambdaSeal-in-gsf}
We now present a semantic-preserving embedding of \lambdaSeal terms in \gsf. The embedding relies on a general seal/unseal generator, expressed as a \gsf term. This term, called $\su$ hereafter, is a polymorphic pair of two functions, of type $\forall X.\pairtype{(X -> \?)}{(\? -> X)}$, instantiated at the unknown type, and ascribed to the unknown type:
$$\su \equiv (\Lambda X.\pair{(\lambda x:X. x :: \?)}{(\lambda x:\?. x :: X)}) \; [\?] :: \?$$

When evaluated, the type application generates a fresh type name, simulating the seal generation of \lambdaSeal's term $\lste{\sealC{x}{t}}$. Then the first component of the pair represents a sealing function, while the second component represents an unsealing function, which can only successfully be applied to values sealed with the first component.
We write $\sus$ to denote a particular value resulting from the evaluation of the term $\su$, where the type name $\sigma$ is generated and stored in $\store$.
Crucially, a value that passed through $\fst[\sus]$ is sealed, and can only be observed after passing through the unsealing function $\snd[\sus]$. Trying to unseal it with a different function results in a runtime error.

\begin{figure}[t]
\begin{footnotesize}
\begin{center}
\begin{minipage}{0.2\textwidth}
\begin{align*}
  \compileLSA{\const} &= \const :: \? \\
  \compileLSA{\lambda x.t} &= (\lambda x : \?.\compileLSA{t}) :: \?  \\
  \compileLSA{\pair{t_1}{t_2}} &= \pair{\compileLSA{{t_1}}}{\compileLSA{{t_2}}} :: \? \\
  \compileLSA{\op{\vectorOp{t}}} &= \letT{\vectorOp{x}: \?}{{\compileLSA{{\vectorOp{t}}}}}{\op{\vectorOp{x}} :: \?}
  \end{align*}
\end{minipage}\quad\quad
\begin{minipage}{0.2\textwidth}
\begin{align*}
  \compileLSA{x} &= x \\
  \compileLSA{t_1\;t_2} &= \letT{x: \?}{\compileLSA{t_1}}{\letT{y: \?}{\compileLSA{t_2}}{x \; y}}\\
  \compileLSA{\proj{1}[t]} &= \proj{1}[\compileLSA{t}]\\
  \compileLSA{\proj{2}[t]} &= \proj{2}[\compileLSA{t}]
  \end{align*}
\end{minipage}\\
\begin{align*}
\compileLSA{\sealC{x}{t}} &= \letC{x: \?}{\su }{\compileLSA{{t}}}\\
  \compileLSA{\sealedC{t_1}{t_2}} &= \letT{x: \?}{\compileLSA{t_1}}{\letT{y: \?}{\compileLSA{t_2}}{\proj{1}[{{y}}] \; {{x}}}}\\
  \compileLSA{\unsealC{z}{t_1}{t_2}{t_3}} &= \letT{x: \?}{\compileLSA{t_1}}{\letT{y: \?}{\compileLSA{t_2}}{\letC{z: \?}{\proj{2}[x]\; y}{\compileLSA{t_3}}}} 
   \end{align*}
$$\text{where } \su \equiv ((\Lambda X.\pair{(\lambda x:X. x :: \?)}{(\lambda x:\?. x :: X)}) \; [\?]) :: \?$$
\end{center}
\end{footnotesize}

 \caption{Embedding \lambdaSeal in \gsf}
  \label{fig:compilelambdasealToGsf}
\end{figure}

\myparagraph{Embedding Translation}
Figure~\ref{fig:compilelambdasealToGsf} defines the embedding from \lambdaSeal to \gsf. The cases unrelated to sealing are as presented in \S\ref{sec:embed-untyped}.
The crux of the embedding is in the use of the term $\su$. A seal generation term $\lste{\sealC{x}{t}}$ is embedded into \gsf by let-binding the variable $x$ to the term $\su$, whose value $\sus$ will be substituted in the translation of $\lste{t}$. Recall that the first component of the pair $\sus$ is used for sealing, and the second one for unsealing.  
Therefore, the sealing operation $\lste{\sealedC{t_1}{t_2}}$ is embedded by let-binding the translations of $\lste{t_1}$ and $\lste{t_2}$ to fresh variables $x$ and $y$, and applying the first component of $y$ (the sealing function) to $x$ (the value to be sealed). Likewise, an unsealing $\lste{\unsealC{z}{t_1}{t_2}{t_3}}$ is embedded by binding the translation of $\lste{t_1}$ and $\lste{t_2}$ to fresh variables $x$ and $y$, then unsealing $y$ using the second component of $x$ (the unsealing function), and binding the result to $z$, for use in the translation of the term $\lste{t_3}$. The use of A-normal forms in the embedding of sealing and unsealing is required because both are eventually interpreted as function applications, so the precaution discussed in \S\ref{sec:embed-untyped} applies.
Finally, note that because seals $\seal$ cannot appear in source text, so the translation need not consider them.

\myparagraph{Illustration}
As example, the embedding of the \lambdaSeal term $\lste{\sealC{x}{\sealC{y}{\unsealC{n}{x}{\sealedC{1}{x}}{n + 1}}}}$ is the following \gsf term: 

\begin{small}
\begin{mathpar}
\begin{array}{ll}
&\letC{x:\?}{\su}{}\\
&\letC{y:\?}{\su}{}\\
&\letC{u:\?}{x}{}\\
&\letC{z:\?}{(\letC{n_1:\?}{1}{\letC{s:\?}{x}{\fst[s] \; n_1)}}}{}\\
&\letC{n:\?}{\snd[u] \; z}{n + 1}
\end{array} 
\end{mathpar}
\end{small}

The following reduction trace shows the most critical steps of the program above. We define $\sue$ as the translation of $\su$ to \gsfe, and $\suse$ is the value of $\sue$, where a fresh seal $\sigma$ is generated.  Note that we omit some trivial evidences and type annotations for readability. This program generates two fresh type names, reducing the first $\su$ to $\suse$ and the second one to $\susep$. Then, after a few substitution steps, the first component of $\suse$ is applied to $1$, sealing the value, and then applies the second component of $\suse$, unsealing the sealed value. The whole program reduces to $2$.

\begin{small}
\begin{mathpar}
\begin{array}{rcrlrr}
&\emptyenv & \triangleright& \letC{x}{\sue}{\letC{y}{\sue}{\cdots}}{} & \text{\footnotesize initial program}\\
\red^{*}&\sigma := \?, \sigma' := \? & \triangleright& 
{\letC{u}{\suse}{\cdots}}{\; \letC{s}{\suse}{\cdots}} & \text{\footnotesize $\sigma$ and $\sigma'$ are generated} \\
\red^{*}&\sigma := \?, \sigma' := \?& \triangleright& \letC{z}{\fst[\suse] (\ev[\Int]1::\?)}{\cdots}& \text{\footnotesize substitution steps}\\
\red^{*}&\sigma := \?, \sigma' := \?& \triangleright& {\letC{n}{\snd[\suse] (\cast{\pr{\Int,\richtype[\sigma]{\Int}}}{1::\?)}}{n + 1}}& \text{\footnotesize argument is sealed by $\sigma$} \\
\red^{*}&\sigma := \?, \sigma' := \?& \triangleright& {\letC{n}{\cast{\ev[\Int]}{1 :: \?}}{n + 1}}& \text{\footnotesize unsealing eliminates $\sigma$}\\
\red^{*}&\sigma := \?, \sigma' := \?& \triangleright& {\cast{\ev[\Int]}{2 :: \?}}&\\
\end{array} 
\end{mathpar}
\end{small}

If we slightly modify the previous \lambdaSeal program by $\lste{\sealC{x}{\sealC{y}{\unsealC{n}{y}{\sealedC{1}{x}}{n + 1}}}}$, then unsealing fails with $\lste{\unsealError}$ because it uses a different seal to unseal than the one used to seal. The embedding of this \lambdaSeal term in \gsf is very similar to the previous one; the only difference is that, now, $u$ is bound to $y$. The following reduction trace illustrates where the embedding of the \lambdaSeal term fails. Note that the resulting value of {\small $\snd[\susep]$} is {\small$\cast{\pr{\? -> \richtype[\sigma']{\?}, \? -> \?}}{(\lambda x:\?.\cast{\pr{\richtype[\sigma']{\?}, \richtype[\sigma']{\?}}}{x :: \sigma'}) :: \? -> \?}$}. Then, the sealed value {\small$\cast{\pr{\Int,\richtype[\sigma]{\Int}}}{1::\?}$} is substituted in the body of the function, failing in the consistent transitivity {\small$\cast{\pr{\Int, \richtype[\sigma]{\Int}}}{}\trans{}\cast{\pr{\richtype[\sigma']{\?}, \richtype[\sigma']{\?}}}{}$}.

\begin{small}
\begin{mathpar}
\begin{array}{rcrlrr}
&\emptyenv & \triangleright& \letC{x}{\sue}{\letC{y}{\sue}{\cdots}}{} & \text{\footnotesize initial program}\\
\red^{*}&\sigma := \?, \sigma' := \? & \triangleright& 
{\letC{u}{\susep}{\cdots}}{\; \letC{s}{\suse}{\cdots}} & \text{\footnotesize $\sigma$ and $\sigma'$ are generated} \\
\red^{*}&\sigma := \?, \sigma' := \?& \triangleright& \letC{z}{\fst[\suse] (\ev[\Int]1::\?)}{\cdots}& \text{\footnotesize substitution steps}\\
\red^{*}&\sigma := \?, \sigma' := \?& \triangleright& {\letC{n}{\snd[\susep] (\cast{\pr{\Int,\richtype[\sigma]{\Int}}}{1::\?)}}{n + 1}}& \text{\footnotesize argument is sealed by $\sigma$} \\
\red^{*}&&& \error&\text{\footnotesize \error \; unsealing by $\sigma'$}
\end{array} 
\end{mathpar}
\end{small}

\subsection{Embedding of \lambdaSeal}
\label{sec:semantics-preservation}
\newsavebox{\lsRGsf}

We now prove that the embedding of \lambdaSeal into \gsf is correct, namely that a \lambdaSeal term and its translation to \gsf behave similarly: either they both terminate to a value, both diverge, or both yield an error. Note that the semantic preservation theorem below only accounts for what we call {\em valid} \lambdaSeal terms, \ie~terms that do not produce runtime type errors related to sealing, \ie~\sealError. We come back to this point at the end of this section.
We write ${\lste{{t}\terminationsy}}$ or $\terminationLS{t}$ if $\storeevalLS[\emptyenv]{t} \; \lste{\red^{*}} \; \storeevalLS[\heap]{v}$, for some $\lste{v}$ and $\heap$. We write ${\lste{\divergeLS{t}}}$ if $\lste{t}$ diverges, and ${\lste{{t}\terminationsy \stError}}$ if $\storeevalLS[\emptyenv]{t} \; \lste{\red^{*}} \; \stError$, where $\stError \triangleq \typeError \text{ or } \unsealError$. As before, we write $\terminationGSF$ if $|- t  \translate t_{\ev} : \?$ and $\conf[\emptyenv]{t_{\ev}} \red^{*} \conf[\store] v$, for some $v$ and $\store$.

\begin{restatable}[Embedding of \lambdaSeal]{theorem}{semanticsPreservationLamdaSealMain}
\label{theorem:semanticsPreservationLamdaSealMain}
Let $\lste{t}$ be a valid closed \lambdaSeal term.
\begin{enumerate}[label=\alph*.]
  \item $|- \compileLSA{t} : \?$
  \item ${\lste{{t}\terminationsy}}$ implies $\gsfreds{\compileLSA{t}}$
  \item ${\lste{\divergeLS{t}}}$ implies $\compileLSA{t} \divergesy$
  \item ${\lste{{t}\terminationsy \stError}}$ implies  $\gsfreds{\compileLSA{t}}{\error}$
  \end{enumerate}
\end{restatable}

\begin{figure}[!p]
\label{def:relationLambdaSealGsf}
\begin{small}
  \begin{mathpar}
    \inference[(TRx)]{ x : \? \in \Gamma 
    }{  \simulationRelT{x}{x}{\?}}
    \and
    \inference[(Rb)]{\ftype(b) = \basetype }
    {\simulationRelT{b}{\cast{\ev[\basetype]}{b :: \?}}{\?}}
    \and
     \inference[(TRu)]{
       \simulationRelT{v}{\cast{\ev[\gT]}{u :: \?}}{\?}
    }{
     \simulationRelT{v}{\cast{\ev[\gT]}{(\cast{\ev[\gT]}{u :: \gT}) :: \?}}{\?}
    }
    \and
    \inference[(Rs)]{\seal \in \heap  & \sigma := \? \in \store 
    }{  \simulationRelT{\sigma}{\suse}{\?}}
    \and
 \inference[(Rp)]{
       \simulationRelT{v_1}{\cast{\ev[\gT_1]}{u_1 :: \?} }{\?}  & 
        \simulationRelT{v_2}{\cast{\ev[\gT_2]}{u_2 :: \?} }{\?}
    }{
     \simulationRelT{\pair{v_1}{v_2}}{\cast{\ev[\pairtype{\gT_1}{\gT_2}]}{\pair{u_1}{u_2} :: \?}}{\?}
    }\and
     \inference[(R$\lambda$)]{
       \simulationRelT{t_1}{t_2}{\?}[\heap; \store; \Gamma, x : \?]
    }{
       \simulationRelT{(\lambda x.t_1)}{\cast{\ev[\? -> \?]}{(\lambda x.t_2) :: \?}}{\?}
    }
\and
 \inference[(TRpt)]{
       \simulationRelT{t_1 }{t'_1}{\?}  & 
        \simulationRelT{t_2}{t'_2}{\?}
    }{
     \simulationRelT{\pair{t_1}{t_2}}{\cast{\ev[\pairtype{\?}{\?}]}{\pair{t'_1}{t'_2} :: \?}}{\?}
    }\and
 \inference[(R$\?$)]{
       \simulationRelT{t}{t'}{\?}
    }{
     \simulationRelT{t}{\cast{\ev[\?]}{t' :: \?}}{\?}
    }
    \and
    \inference[(Rapp)]{
       \simulationRelT{v_1 }{v'_1}{\?}  &
       \simulationRelT{v_2 }{v'_2 }{\?}
    }{
        \simulationRelT{v_1 \; v_2}{(\cast{\ev[\? -> \? ]}{v'_1 :: \? -> \? })\;v'_2}{\?}
        }
    \and
    \inference[(TRappL)]{
       \simulationRelT{t_1 }{t'_1}{\?}  &
       \simulationRelT{t_2 }{t'_2 }{\?}
    }{
        \simulationRelT{t_1 \; t_2}{\letT{x}{t'_1}{\letT{y}{t'_2}{(\cast{\ev[\? -> \? ]}{x :: \? -> \? })\; y}}}{\?} 
    }
    \and
    \inference[(RappR)]{
       \simulationRelT{v_1 }{v'_1}{\?}  &
       \simulationRelT{t_2 }{t'_2 }{\?}
    }{
        \simulationRelT{v_1 \; t_2}{\letT{y}{t'_2}{(\cast{\ev[\? -> \? ]}{v'_1 :: \? -> \? })\; y}}{\?} 
    } 
    \and
\inference[(TRpi)]{
      \simulationRelT{t}{t'}{\?} 
    }{
      \simulationRelT{\proj{i}[t]}{\proj{i}[\cast{\ev[{\pairtype{\?}{\?}}]}{t' :: \pairtype{\?}{\?}}]}{\?}  
    }
    \and   
    \inference[(TRsG)]{
        \simulationRelT{t}{t'}{\?}[{\heap;\store; \Gamma, x: \?}]
    }{
        \simulationRelT{\sealC{x}{t}}{\letC{x}{\sue}{t'}}{\?} 
    }
    \and
       \inference[(Rsed1)]{
      \simulationRelT{v_1}{v'_1}{\?} & \simulationRelT{v_2}{v'_2}{\?}
    }{
      \simulationRelT{\sealedC{v_1}{v_2}}{{\cast{\ev[\? -> \?]}{\proj{1}[\cast{\ev[\pairtype{\?}{\?}]}{{{{v'_2}} :: \pairtype{\?}{\?}}}]:: \? -> \?} \; {{v'_1}}}}{\?}   
    }
    \and
       \inference[(TRsed1L)]{
      \simulationRelT{t_1}{t'_1}{\?} & \simulationRelT{t_2}{t'_2}{\?}
    }{
      \simulationRelT{\sealedC{t_1}{t_2}}{\letT{x}{t'_1}{\letT{y}{t'_2}{(\cast{\ev[\? -> \?]}{\proj{1}[\cast{\ev[\pairtype{\?}{\?}]}{{{{y}} :: \pairtype{\?}{\?}}}]:: \? -> \?}) \; {{x}}}}}{\?}   
    }
    \and
       \inference[(Rsed1R)]{
      \simulationRelT{v_1}{v'_1}{\?} & \simulationRelT{t_2}{t'_2}{\?}
    }{
      \simulationRelT{\sealedC{v_1}{t_2}}{\letT{y}{t'_2}{{(\cast{\ev[\? -> \?]}{\proj{1}[\cast{\ev[\pairtype{\?}{\?}]}{{{{y}} :: \pairtype{\?}{\?}}}]:: \? -> \?}) \; {{v'_1}}}}}{\?}   
    }
\and 
\inference[(Rsed2)]{
      \simulationRelT{v}{\cast{\pr{\gE_1,{\gE_2}}}{u :: \?}}{\?} &\seal \in \heap &\sigma := \? \in \store 
    }{
      \simulationRelT{\sealedC{v}{\seal}}{\cast{\pr{\gE_1,\richtype[\sigma]{\gE_2}}}{u :: \?}}{\?}   
    } 
    \and
       \inference[(Runs)]{
      \simulationRelT{v_1}{v'_1}{\?} & \simulationRelT{v_2}{v'_2}{\?} & \simulationRelT{t_3}{t'_3}{\?}[\heap;\store; \Gamma, z: \?]
    }{
      \simulationRelT{\unsealC{z}{v_1}{v_2}{t_3}}{\letC{z}{{\cast{\ev[\? -> \?]}{\proj{2}[\cast{\ev[\pairtype{\?}{\?}]}{{{{v'_1}} :: \pairtype{\?}{\?}}}]:: \? -> \?} \; {{v'_2}}}}{t'_3}}{\?}   
    }
    \and
       \inference[(TRunsL)]{
      \simulationRelT{t_1}{t'_1}{\?} & \simulationRelT{t_2}{t'_2}{\?} & \simulationRelT{t_3}{t'_3}{\?}[\heap;\store; \Gamma, z: \?]
    }{
      \simulationRelT{\unsealC{z}{t_1}{t_2}{t_3}}{\letT{x}{t'_1}{\letT{y}{t'_2}{{{\letC{z}{{\cast{\ev[\? -> \?]}{\proj{2}[\cast{\ev[\pairtype{\?}{\?}]}{{{{x}} :: \pairtype{\?}{\?}}}]:: \? -> \?} \; {{y}}}}{t'_3}}}}}}{\?}   
    }
    \and
       \inference[(RunsR)]{
      \simulationRelT{v_1}{v'_1}{\?} & \simulationRelT{t_2}{t'_2}{\?} & \simulationRelT{t_3}{t'_3}{\?}[\heap;\store; \Gamma, z: \?]
    }{
      \simulationRelT{\unsealC{z}{v_1}{t_2}{t_3}}{\letT{y}{t'_2}{{\letC{z}{{\cast{\ev[\? -> \?]}{\proj{2}[\cast{\ev[\pairtype{\?}{\?}]}{{{{v'_1}} :: \pairtype{\?}{\?}}}]:: \? -> \?} \; {{y}}}}{t'_3}}}}{\?}   
    }
  \end{mathpar}
  \end{small}
 \caption{Simulation relation between \lambdaSeal and \gsfe terms}
  \label{fig:bisimulationMain}
\end{figure}

To prove Theorem~\ref{theorem:semanticsPreservationLamdaSealMain}, we use a simulation relation $\sylogeqLS$ between \lambdaSeal and \gsfe, defined in Figure~\ref{fig:bisimulationMain}. 
The simulation relation $\simulationRelT{t}{t_{\ev}}{\?}$  uses a set of allocated seals $\heap$ by the reduction of the \lambdaSeal term $\lste{t}$. The \gsfe term  $t_{\ev}$ typechecks in the typing environment $\Gamma$ where all variables have type unknown, and it typechecks and it is evaluated in the store $\store$ with all its type names instantiated, also, to the unknown type. In all the rules of the simulation, we implicitly assume that $\heap$ and $\store$ are synchronized, \ie if $\seal \in \heap$ then $\sigma := \? \in \store$. Rules whose names begin with (TR) relate a \lambdaSeal term and its translation in \gsfe, \ie embedding first the \lambdaSeal term into \gsf, and then translating the resulting \gsf term to a \gsfe term. For instances, Rule (TRu) relates the \lambdaSeal value $\lste{1}$ with the  $\gsfe$ term $\cast{\ev[\Int]}{(\ev[\Int] 1 :: \Int) :: \?}$. Note that Rule (TRu) uses metavariable $\gT$ to denote the possible types of \gsfe raw values ($u$), obtained by the embedding: either a base type $\basetype$, an unknown function type $\? -> \?$, or a pair of raw values $\pairtype{\gT}{\gT}$. Rule (TRsG) relates the seal generation term $\sealC{x}{t}$ with the \gsfe term that let-binds the variable $x$ to the term $\sue$ to be substituted in $t'$; it requires that the bodies of the seal generation and let-binding be related.
The remaining rules, whose names begin with (R), help us keep terms related as they reduce. One of the most important rules is (Rsed2), which relates a \lambdaSeal sealed value with a \gsf value that has sealing evidence. Rule (Rsed1) relates a sealed value $\sealedC{v_1}{v_2}$ with a \gsfe term that takes the first component of $v'_2$ (expected to be a $\suse$ value related to the seal $\lste{v_2}$), and applies it to $v'_1$ related to $\lste{v_1}$. Dually, Rule (Runs) relates a term for unsealing with a \gsfe term that takes the second component of $v'_1$ (expected to be a $\suse$ value related to the seal $\lste{v_1}$), and applies it to $v'_2$ related to $\lste{v_2}$ (expected to be a sealed value). Also, the rule requires the bodies 
$\lste{t_3}$ and $t_3$ to be related.
  
We first establish a number of useful lemmas. First, all \gsfe terms that are in the relation have type unknown, simulating the fact that they are related to untyped \lambdaSeal terms.
\begin{restatable}{lemma}{typingGsfeTermsByUnknownMain}
\label{lemma:typingGsfeTermsByUnknownMain}
If $\simulationRelT{t}{t_{\ev}}{\?}$ then $\store; \Gamma |- t_{\ev} : \?$.
\end{restatable}
Also, the relation $\sylogeqLS$ guarantees that if we have a \lambdaSeal value related to a \gsfe term, then 
the latter reduces to a related value. 

\begin{restatable}{lemma}{valueOnTheLeftMain}
\label{lemma:valueOnTheLeftMain}
If $\simulationRel{v}{t_{\ev}}{\?}$ , then there exists $v_{\ev}$ s.t. $\conf[\store]{t_{\ev}}  \red^{*} \conf[\store]{v_{\ev}}$, and $\simulationRel{v}{v_{\ev}}{\?}$.
\end{restatable}
For example, we know by Rule (Rsed1L) that 
{\small $$\simulationRel{\sealedC{1}{\seal}}{\letT{x}{\cast{\ev[\Int]}{1 :: \?}}{\letT{y}{{\suse}}{ (\cast{\ev[\? -> \?]}{\proj{1}[\cast{\ev[\pairtype{\?}{\?}]}{{y :: \pairtype{\?}{\?}}}]:: \? -> \?}) \; x}}}{\?}[\seal; \sigma:= \?]$$}
Thus, we know that the \gsfe term reduces to a value, in this case, $\cast{\pr{\Int, \richtype[\sigma]{\Int}}}{1 :: \?}$.

Lemma~\ref{lemma:substitutionpreservesrelationMain} establishes 
substituting related values in related terms yields related terms.
\begin{restatable}{lemma}{substitutionpreservesrelationMain}
\label{lemma:substitutionpreservesrelationMain}
If $\simulationRelT{t}{t_{\ev}}{\?}[\heap;\store; \Gamma, x:\?]$ and $\simulationRelT{v}{v_{\ev}}{\?}$, then \mbox{$\simulationRelT{t[v/ x]}{t_{\ev}[v_{\ev}/ x]}{\?}$}.
\end{restatable}

Lemma~\ref{lemma:simulatesReduction} shows that the relation $\sylogeqLS$ simulates both the notions of reduction \lste{$\nred$} and $\nred$, and the reduction relations \lste{$\red$} and $\red$, including error cases. 
Note that a single step of reduction in \lambdaSeal can be simulated by several reduction steps in \gsfe, hence the use of $\red^{*}$ in the conclusions of the lemma cases.
For example, we have $\simulationRel{\fst[\pair{1}{2}]}{\fst[\cast{\ev[\pairtype{\?}{\?}]}{(\cast{\ev[\pairtype{\Int}{\Int}]}{\pair{1}{2} :: \?) :: \pairtype{\?}{\?}}}]}{\?}$, and the \gsfe term needs to reduce inside the frame $\fst[{[]}]$ before eliminating the projection like the \lambdaSeal term. 

\begin{restatable}{lemma}{simulatesReduction}
\label{lemma:simulatesReduction} 
Suppose that $\lste{t}$ is a term of \lambdaSeal, $t_{\ev}$ is a term from \gsfe and $\simulationRel{t}{t_{\ev}}{\?}$.  
\begin{enumerate}[label=\alph*.]
\item \label{case:caseA} If ${\lste{\storeevalLS{t} \nred \storeevalLS[\heapp]{t'}}}$, then there exists $t'_{\ev}$ s.t. $\conf[\store]{t_{\ev}} \red^{*} \conf[\store']{t'_{\ev}}$ and  $\simulationRel{t'}{t'_{\ev}}{\?}[\heapp;\store']$
\item \label{case:caseB}  If ${\lste{\storeevalLS{t} \nred \stError}}$, then $\conf[\store]{t_{\ev}} \red^{*} \error$
\item \label{case:caseC}  If ${\lste{\storeevalLS{t} \red \storeevalLS[\heapp]{t'}}}$, then there exists $t_{\ev}'$ s.t. $\conf[\store]{t_{\ev}} \red^{*} \conf[\store']{t_{\ev}'}$ and  $\simulationRel{t'}{t_{\ev}'}{\?}[\heapp;\store']$
\item \label{case:caseD}  If ${\lste{\storeevalLS{t} \red \stError}}$, then $\conf[\store]{t_{\ev}} \red^{*} \error$
\end{enumerate}
\end{restatable}
\begin{proof}
The proof is by induction on $\simulationRel{t}{t_{\ev}}{\?}$ and by analysis of the different cases.
\begin{case}[\ref{case:caseA}] Most of the cases use Lemma~\ref{lemma:valueOnTheLeftMain}, Lemma~\ref{lemma:substitutionpreservesrelationMain} and the consistent transitivity relation.
\end{case}
\begin{case}[\ref{case:caseB}] Most of the cases use Lemma~\ref{lemma:valueOnTheLeftMain} and the consistent transitivity relation.
\end{case}
\begin{case}[\ref{case:caseC}] The proof follows by cases analysis on ${\lste{\storeevalLS{t} \red \storeevalLS[\heapp]{t'}}}$and  from Case~(\ref{case:caseA}).
\end{case}
\begin{case}[\ref{case:caseD}] The proof follows by cases analysis on ${\lste{\storeevalLS{t} \red \stError}}$ and  from Case~(\ref{case:caseB}).
\end{case}
\end{proof}

The main property of the relation $\sylogeqLS$ is that related terms behave similarly:
\begin{restatable}[]{lemma}{TerminationSimulatedTermLambdaSealMain}
\label{lemma:TerminationSimulatedTermLambdaSealMain}
If $\simulationRel{t}{t_{\ev}}{\?}[]$ then 
\begin{itemize}
  \item ${\lste{\terminationLS{t}}}$ implies  $\conf[\emptyenv]{t_{\ev}} \red^{*} \conf{v_{\ev}}$, where $\simulationRel{v}{v_{\ev}}{\?}$.
  \item ${\lste{\divergeLS{t}}}$ implies $t_{\ev}$ diverges.
  \item ${\lste{\terminationLSTError{t}}}$ implies  $\conf[\emptyenv]{t_{\ev}} \red^{*}  {\error}$.
  \end{itemize}
\end{restatable}
\begin{proof}
The proof is by case analysis on the reduction of $t$.
\begin{itemize}
\item Suppose ${\lste{\terminationLS{t}}}$. Then  $\conf[\emptyenv]{t_{\ev}} \red^{*} \conf{v_{\ev}}$ and $\simulationRel{v}{v_{\ev}}{\?}$ by Lemmas \S\ref{lemma:valueOnTheLeftMain} and \S\ref{lemma:simulatesReduction}(\ref{case:caseC}).
\item Suppose ${\lste{\divergeLS{t}}}$. Then $t_{\ev}$ diverges by Lemma \S\ref{lemma:simulatesReduction}(\ref{case:caseC}).
\item Suppose ${\lste{\terminationLSTError{t}}}$, then  $\conf[\emptyenv]{t_{\ev}} \red^{*}  {\error}$ by Lemma \S\ref{lemma:simulatesReduction}(\ref{case:caseC} and \ref{case:caseD}).
\end{itemize}
\end{proof}

Finally, a \lambdaSeal term and its embedding into \gsfe are related.
\begin{restatable}{lemma}{semanticsPreservationLamdaSealRelationMain}
\label{lemma:semanticsPreservationLamdaSealRelationMain}
If $|- \compileLSA{t} \translate t_{\ev} : \?$, then $\simulationRelT{t}{t_{\ev}}{\?}[]$.
\end{restatable}

Semantics preservation (Theorem~\ref{theorem:semanticsPreservationLamdaSealMain}) follows from 
Lemma~\ref{lemma:TerminationSimulatedTermLambdaSealMain} and
Lemma~\ref{lemma:semanticsPreservationLamdaSealRelationMain}.

\myparagraph{Leaking the Encoding} As mentioned earlier, the semantic preservation result does not account for \lambdaSeal terms that can raise runtime seal type errors, \sealError. The reason is that, without further caution, the encoding of seals as pairs of functions could be abused. For instance, the term $\lste{\unsealC{y}{\pair{\lambda x.x}{\lambda x.x}}{1}{y}}$ raises a \sealError in \lambdaSeal, because the expression that is supposed to produce a seal produces a pair of functions. Nevertheless, the embedding of this term in \gsf reduces to $1$. To properly deal with such cases---and therefore obtain a semantic preservation statement with equivalences instead of implications---would require introducing a primitive way of distinguishing ``proper seals'' produced by the translation from standard pairs of functions (\eg~an implementation could used a privately generated token). Of course, in a statically-typed version of \lambdaSeal~\cite{pierceSumii:2000}, this problem is sidestepped because a \sealError can never occur at runtime.

\myparagraph{Embedding of the Dynamically-Typed Language}
Finally, a direct consequence of the semantics preservation theorem is that the embedding of \dynLang is also correct; in fact the embedding result holds as stated by~\cite{siekAl:snapl2015}(Theorem 2), with equivalences instead of implications:

\begin{restatable}[Embedding of \dynLang]{corollary}{equivalenceSemantics}
\label{corollary:equivalenceSemantics}
Let $\lste{t}$ be a closed \dynLang term.
\begin{enumerate}[label=\alph*.]
  \item $|- \compileLSA{t} : \?$
  \item ${\lste{{t}\terminationsy}}$ if and only if $\gsfreds{\compileLSA{t}}$
  \item ${\lste{\divergeLS{t}}}$ if and only if $\compileLSA{t} \divergesy$
  \end{enumerate}
\end{restatable}

This result follows from Theorem~\ref{theorem:semanticsPreservationLamdaSealMain} combined with the fact that a \sealError simply cannot occur in \dynLang, which has no sealing-related terms.

\section{Beyond Explicit Polymorphism: Dynamic Implicit Polymorphism}
\label{sec:dynamic-implicit-polymorphism}

Being able to faithfully embed untyped programs in a gradual language, as presented in \S\ref{sec:embedding}, does not necessarily mean that the interoperability between typed and untyped programs is fully supported. In particular, because \sysF and \gsf have term-level constructs related to type abstraction and application, there is a kind of ``impedance mismatch'' between both worlds.
Specifically, the strict reliance of \gsf on {\em explicit polymorphism} gets in the way.
Consider the following example: 
\begin{lstlisting}[numbers=none]
  let id:(*$\forall. X. X -> X$*) = (*$\Lambda$*)X.(*$\lambda$*)x:X.x
  let f:? = (*$\lceil\lambda$*)y.y 1(*$\rceil$*)
  f id 
\end{lstlisting}
When the body of \lstinline|f| tries to apply \lstinline|id| to \lstinline|1|, a runtime error is raised: $\forall X. X -> X$ cannot be used as an $\Int -> \Int$ function.
Similarly, the following example is also rejected at runtime:
\begin{lstlisting}[numbers=none]
  let g:? = (*$\lambda$*)x:(*$(\forall. X. X -> X)$*).x [Int] 1
  let h:? = (*$\lceil\lambda$*)x.x(*$\rceil$*)
  g h
\end{lstlisting}
This time, the runtime error is raised when \lstinline|g| is applied to \lstinline|h|, because 
$\? -> \?$ (the ``underlying type'' of \lstinline|h|) is not consistent with the polymorphic function type $\forall X. X -> X$.

Of course, we can fix both examples fairly easily by inserting a type instantiation to $\?$ in the former (rewriting the last expression to \lstinline|f (id [?])|), and a type abstraction in the latter (rewriting the last expression to \lstinline|g (|$\Lambda$\lstinline|_.h)|).
In essence, these fixes reflect the subtyping relation induced by implicit polymorphism~\cite{mitchell:ic1988,oderskyLaufer:popl1996} (\S\ref{sec:back:param}). However, due to gradual typing, we cannot statically and modularly decide when to insert such type instantiations and abstractions. In the examples, both \lstinline|f| and \lstinline|h| have type $\?$, so their use without such adaptations is perfectly valid typing-wise, and modularly we cannot inspect their definitions to anticipate which adaptations can be necessary.

\subsection{Implicit Polymorphism, Dynamically}
\label{sec:implicit-polymorphism-dynamically}
To avoid the interoperability issues of \gsf presented above, a gradual polymorphic language could bake in implicit polymorphism in its gradual typing rules. This is the approach taken by \lamB~\cite{ahmedAl:icfp2017}, which features two type compatibility rules to support this kind of implicit polymorphism:
\begin{mathpar}
  \inference[(Comp-AllR)]{\Sigma; \Delta, X |- T_1 <: T_2 & X \not\in T_1}
  {\Sigma; \Delta |- T_1 <: \forall X. T_2}
  \and
  \inference[(Comp-AllL)]{\Sigma; \Delta |- T_1[\? / X] <: T_2}
  {\Sigma; \Delta |- \forall X. T_1 <: T_2}
\end{mathpar}
As first identified by \citet{xieAl:esop2018}, the problem with these rules is that they break the conservative extension of \sysF: types $\forall X. X -> X$ and $\forall X. \Int -> \Bool$ are deemed compatible in \lamB. In \gsf, concentrating on explicit polymorphism ensures the conservative extension result, but limits interoperability. Here, we propose a novel technique to solve this conundrum: to support implicit polymorphism {\em dynamically}, as part of the runtime semantics of \gsf.

Specifically, we extend the reduction rules of \gsf by adding fallbacks in some cases where consistent transitivity is not defined. Intuitively, if consistent transitivity fails because a subterm $t$ of type $\forall X. \cT_1$ is not consistent with some type $\cT_2$, instead of raising an error, we can instantiate $t$ with the unknown type, $t\;[\?]$, and keep reducing the program. This mimics Rule (Comp-AllL) from \lamB. Similarly, if consistent transitivity fails because a subterm $t$ of type $\cT_1$ is not consistent with $\forall X. \cT_2$, we can insert a dummy type abstraction $\forall \_. t$, and keep reducing the program. This effectively mimics Rule (Comp-AllR) dynamically.

\subsection{Dynamic Implicit Polymorphism in \gsf}
\label{sec:dynamic-implicit-polymorphism-in-gsf}
Technically, we realize dynamic implicit polymorphism in \gsf by calling the partial function $\impPolyName$ after every failed consistent transitivity. The definition of $\impPolyName$ and its use in the modified runtime semantics are presented in Figure~\ref{fig:implicit-poly}. 

The $\impPolyName$ function takes as argument a value $v$ and an evidence $\ev$---such that the consistent transitivity between the underling evidence of $v$ and $\ev$ is undefined---and returns an adjusted term if a potential implicit polymorphism interaction is detected.

The first case of the definition of $\impPolyName$ inserts a type abstraction when a polymorphic type is expected but the value is something else.
The second case inserts a type instantiation to $\?$ when a value other than a type abstraction is expected.
The last two cases deal with higher-order scenarios. Specifically, the third case applies when both evidences are pairs of functions: implicit polymorphism can be plausible in both argument and return positions of evidence function types, so we delay plausible adjustments to until the function is applied. Similarly, the fourth case applies when both evidences are type abstractions; we delay plausible adjustments to when the type abstraction is instantiated.

Figure~\ref{fig:implicit-poly}\footnote{
$\cscheme$ (consistently) extracts the schema of a gradual type, \ie~$\cscheme(\forall X.\cT) = \cT$, $\cscheme(\?) = \?$, undefined o/w.} shows the modified reduction rules (Rasc) and (Rapp), which now fallback on $\impPolyName$ whenever consistent transitivity is not defined. If $\impPolyName$ is defined, the runtime adaptation of dynamic implicit polymorphism happens, otherwise a runtime error is raised. Note that this adaptation, which potentially converts a value to a term, cannot introduce divergence because (1) when inserting a type instantiation, the size of the underlying evidence of value $v$ in $\impPoly{v}{\ev}{}$ will get smaller after the actual instantiation, and (2) when inserting a type abstraction around a value $v$, subsequent steps of reductions can introduce a bounded amount of type abstraction around $v$, bounded by the size of $\ev$ in $\impPoly{v}{\ev}{}$.

\begin{figure}[t]
\begin{small}
\begin{multline*}
\impPoly{v}{\ev[2]}{\cT_2} = \\
  \qquad\begin{cases}
    \cast{\ev[\forall X. \cT_1]}{(\Lambda X. \cast{\ev[1]}{u} :: \?)} :: \cT_1 & v = \cast{\ev[1]}{u} :: \cT_1 \land \ev[1] \neq \pr{\forall X. \_, \forall X. \_} \land\ev[2] = \pr{\forall X. \_, \forall X. \_} \\
    \cast{\ev[\cT_1]}{((\cast{\ev[1]}{(\Lambda X. t)} :: \forall X. \?)\; [\?])} :: \cT_1 & v = \cast{\ev[1]}{(\Lambda X. t)} :: \cT_1 \land \ev[1] = \pr{\forall X. \_, \forall X. \_} \land\ev[2] \neq \pr{\forall X. \_, \forall X. \_}\\
    \cast{->\ev[\cT_1]}{(\lambda x: \?. (\cast{\ev[1]}{u} :: \? -> \?) ) \; (\cast{\ev[\?]}{x} :: \?) } :: \cT_1 & v = \cast{\ev[1]}{u} :: \cT_1 \land \ev[i] = \pr{\_ -> \_, \_ -> \_}\\ 
    \cast{\forall X.\ev[\cT_1]}{(\Lambda X. \cast{\cscheme(\ev[1])}{t} :: \?)} :: \cT_1 & v = \cast{\ev[1]}{(\Lambda X. t)} :: \cT_1 \land \ev[i] = \pr{\forall X. \_, \forall X. \_}
  \end{cases}
\end{multline*}
where $X$ is fresh\\ 
Notation: $\forall X. \pr{\cE_1, \cE_2} \triangleq \pr{\forall X.\cscheme(\cE_1), \forall X.\cscheme(\cE_2)}$ and $->\pr{\cE_1, \cE_2} \triangleq \pr{\dom(\cE_1) -> \cod(\cE_1), \dom(\cE_2) -> \cod(\cE_2)}$ 
\end{small}

\begin{small}
\begin{displaymath}
  \text{($R$asc}) \qquad \conf{\cast{\ev[2]}(\cast{\ev[1]}{u} :: \cT_1) :: \cT_2} \nred
  \begin{cases}
  \conf{\cast{(\ev[1] \trans{=} \ev[2])}{u} :: \cT_2 }\\
  \conf{\cast{\ev[2]}{t} :: \cT_2} \qquad 
    \begin{block}
    \text{if not defined, and }\\
    t = \impPoly{\cast{\ev[1]}{u} :: \cT_1}{\ev[2]}{\cT_2}
    \end{block}\\
  \error \qquad \text{otherwise}
  \end{cases}
\end{displaymath}

\begin{multline*}
  \text{($R$app)} \qquad   \conf{(\cast{\ev[1]}{
              (\lambda x:\cT_{11}.t) :: \cT_{1} -> \cT_{2}})\;
          (\cast{\ev[2]}{u} :: \cT_1)}
          \nred\\
          \begin{block}
          \begin{cases}
          \conf{\cast{\invcod(\ev[1])}{
              (  t[
              \cast{(\ev[2] \trans{=} \invdom(\ev[1]))}
               {u} :: 
              \cT_{11})/x])} :: \cT_{2}} 
              \\
          \conf{(\cast{\ev[1]}{
              (\lambda x:\cT_{11}.t) :: \cT_{1} -> \cT_{2}})\;
          t'} \qquad 
          \begin{block}
            \text{if not defined, and }\\
            t' = \impPoly{\cast{\ev[2]}{u} :: \cT_1}{\invdom(\ev[1])}{\cT_1}
          \end{block}\\
          \error \qquad \text{otherwise}
          \end{cases}
          \end{block}
\end{multline*}
\end{small}

\caption{Dynamic implicit polymorphism in \gsf}
\label{fig:implicit-poly} 
\end{figure}

\subsection{Dynamic Implicit Polymorphism in Action}
\label{sec:dynamic-implicit-polymorphism-in-action}

Let us revisit the first example presented in this section. Suppose 
$\text{\lstinline|id|} = \cast{\ev[\forall X. X -> X]}{(\Lambda X. \lambda x. x)} :: \forall X. X->X$ and 
$\text{\lstinline|f|} = \cast{\ev[\?->\?]}{(\lambda y: \?. (\cast{\ev[\?->\?]}{y} :: \? -> \?) \; \cast{\ev[\Int]}{1} :: \?)} :: \?$.
This example failed right before applying \lstinline|id| to \lstinline|1|. With the new semantics, it successfully reduces to \lstinline|1| as sketched below:

\begin{small}
\begin{mathpar}
\begin{array}{rllr}
\conf[\emptyenv]{\text{\lstinline|f|}~\text{\lstinline|id|}}
& \red^{*} &\conf[\emptyenv]{\dots
  (\cast{\ev[\? -> \?]}{
    (\cast{\ev[\forall X. X -> X]}{(\Lambda X. \lambda x. x)} :: \forall X. X->X)
  } :: \? -> \?) \; \cast{\ev[\Int]}{1} :: \?\dots
}& \text{\footnotesize imminent failure}\\
\text{\footnotesize($R$asc)} & \red & \conf[\emptyenv]{\dots
  (\cast{\ev[\? -> \?]}{
    ((\cast{\ev[\forall X. X -> X]}{(\Lambda X. \lambda x. x)} :: \forall X. \Gbox{\?}) \Gbox{[\?]})
  } :: \? -> \?) \; \cast{\ev[\Int]}{1} :: \?\dots
} & \text{\footnotesize $\impPolyName$ is used}\\
\text{\footnotesize($R$inst)} & \red & \conf[\alpha \mapsto  \?]{\dots
  (\cast{\ev[\? -> \?]}{
    (\pr{\evlift{\alpha} -> \evlift{\alpha}, \? -> \?}(\cast{\ev[\evlift{\alpha} -> \evlift{\alpha}]}{(\lambda x. x)} :: \?) :: \?)
  } :: \? -> \?) \; \cast{\ev[\Int]}{1} :: \?\dots
} & \\
& \red^{*} & \conf[\alpha \mapsto  \?]{\ev[\Int]1::\?} &
\end{array}
\end{mathpar}
\end{small}

In the ($R$asc) step, $\ev[\forall X. X->X] \trans{=} \ev[\? -> \?]$ is undefined, but because $\impPoly{\text{\lstinline|id|}}{\ev[\? -> \?]}{\? -> \?}$  is defined, the program is adjusted by instantiating \lstinline|id| to $\?$.

Let us now revisit the second example above. Suppose

\begin{small}
\begin{mathpar}
\begin{array}{rcl}
\text{\lstinline|g|} &=& \cast{\ev[(\forall X. X -> X) -> \Int]}{(\lambda x. (\cast{\ev[\Int -> \Int]}{(x [\Int])}::\Int->\Int) (\cast{\ev[\Int]}{1}::\Int))} :: \? \\ 
\text{\lstinline|h|} &=& \cast{\ev[\?->\?]}{(\lambda x. x)} :: \?
\end{array}
\end{mathpar}
\end{small}

The example failed before applying \lstinline|g| to \lstinline|h|; with the adjusted semantics the program now runs successfully as outlined below:

\begin{small}
\begin{mathpar}
\begin{array}{rll}
\conf[\emptyenv]{\text{\lstinline|g|}~\text{\lstinline|h|}} & \red & \conf[\emptyenv]{
  (\cast{\ev[(\forall X. X -> X) -> \Int]}{(\lambda x. \dots )} :: \? -> \?)\; 
  \cast{\ev[\?->\?]}{(\lambda x. x)} :: \? 
} \\
\text{\footnotesize($R$app)}& \red & \conf[\emptyenv]{
  (\cast{\ev[(\forall X. X -> X) -> \Int]}{(\lambda x. \dots )} :: \? -> \?)\;
  \cast{\Gbox{\ev[\forall X. X -> X]}}{(\Gbox{\Lambda X.} \cast{\ev[\?->\?]}{(\lambda x. x)} :: \?)} :: \?
} \\
& \red^{*} & \conf[\alpha \mapsto  \?]{\ev[\Int]1::\?} \\
\end{array}
\end{mathpar}
\end{small}

\noindent In the ($R$app) step, $\ev[\? -> \?] \trans{=} \ev[\forall X. X -> X]$ is undefined, but as $\impPoly{\text{\lstinline|h|}}{\ev[\forall X. X -> X]}{\forall X. X -> X}$  is defined, the program is fixed by inserting a dummy type abstraction around \lstinline|h|.

Let us now illustrate dynamic implicit polymorphism with a higher-order scenario.
Consider the the following program:
\begin{lstlisting}[numbers=none]
let f:?  = (*$\Lambda$*)Y.(*$\lambda$*)y.(*$\lambda$*)x.x in
let g:(*$\forall Y. Y -> (\forall X. X -> X)$*) = f in
(g [Int] 1) [Int] 2
\end{lstlisting}
Suppose

\begin{small}
\begin{mathpar}
\begin{array}{rcl}
\text{\lstinline|f|} &=& \cast{\ev[\forall Y. \? -> (\? -> \?)]}{(\Lambda Y. \lambda y. \dots }) :: \? \\ 
\text{\lstinline|g|} &=& \cast{\ev[\forall Y. Y -> (\forall X. X -> X)]}{\text{\lstinline|f|}} :: \?
\end{array}
\end{mathpar}
\end{small}

Here, function \lstinline|f| of underlying type $\forall Y. \? -> (\? -> \?)$ is used as a function of type $\forall Y. Y -> (\forall X. X -> X)$. This example fails in plain \gsf, because both types are not compatible.
With the adjusted semantics the program runs successfully as outlined next.

\begin{small}
\begin{mathpar}
\begin{array}{rll}
\conf[\emptyenv]{\dots\text{\lstinline|g|}\;[\Int]\;\dots} & = & \conf[\emptyenv]{
  \dots (\cast{\ev[\forall Y. Y -> (\forall X. X -> X)]}{(\cast{\ev[\forall Y. \? -> (\? -> \?)]}{(\Lambda Y. \lambda y. \dots }) :: \?)} :: \_) [\Int] \dots
} \\
\text{\footnotesize($R$asc)}& \red & \conf[\emptyenv]{
  \dots (\cast{\ev[\forall Y. Y -> (\forall X. X -> X)]}{(\cast{\Gbox{\ev[\forall Y.\?]}}{(\Lambda Y. \Gbox{\ev[\? -> (\? -> \?)]}(\lambda y. \dots) :: \? }) \Gbox{:: \?})} :: \_) [\Int] \dots
} \\
& \red^{*} & \conf[\alpha \mapsto \Int]{
  \dots (\cast{\ev[\alpha -> (\forall X. X -> X)]}{(\cast{\ev[\? -> (\? -> \?)]}{(\lambda y. \dots }) :: \_)} :: \_) \dots
} \\
\text{\footnotesize($R$asc)}& \red & \conf[\alpha \mapsto \Int]{
  \dots (\cast{\ev[\alpha -> (\forall X. X -> X)]}{
    ((\Gbox{\lambda z.} (\cast{\ev[\? -> \?]}{(\lambda y. \dots)} :: \_) \Gbox{(\cast{\ev[\?]}{z} :: \_}))
  } :: \_) \dots
} \\
& \red^{*} & \conf[\alpha \mapsto \Int]{
  \dots (\cast{\ev[\forall X. X -> X]}{(\cast{\ev[\? -> \?]}{(\lambda x. \dots }) :: \_)} :: \_) [\Int] \dots 
}\\
\text{\footnotesize($R$asc)} & \red & \conf[\alpha \mapsto \Int]{
  \dots (\cast{\ev[\forall X. X -> X]}{(\cast{\Gbox{\ev[\forall X. \? -> \?]}}{(\Gbox{\Lambda X.} (\lambda x. \dots })) :: \_)} :: \_) [\Int] \dots 
}\\
& \red^{*} & \conf[\alpha \mapsto \Int, \beta \mapsto \Int]{
  \cast{\ev[\Int]}{2} :: \?
}
\end{array}
\end{mathpar}
\end{small}

In the first ($R$asc) step, $\ev[\forall Y. \? -> (\? -> \?)] \trans{=} \ev[\forall Y. Y -> (\forall X. X -> X)]$ is undefined, and as $\impPolyName$ cannot introduce instantiations or type abstractions (they should be made inside the body of \lstinline|f|), it adjusts \lstinline|f| by pushing the scheme of its evidence inside of the type abstraction.
In the second ($R$asc) step, $\ev[\? -> (\? -> \?)] \trans{=} \ev[\alpha -> (\forall X. X -> X)]$ is undefined. Once again, $\impPolyName$ cannot introduce instantiations or type abstractions at that moment, so it creates a function proxy to delay plausible adjustments.
Finally, in the third ($R$asc) step, $\impPolyName$ is called as presented in the second example above, making the program reduce successfully to $2$.

\myparagraph{On Laziness}
One of the consequences of using proxies to delay adjustments for dynamic implicit polymorphism is that some programs that use to fail, may now fail later (or not fail at all).
For instance, consider the term $(\lambda x:\Int. x) :: \? :: \Bool -> \Bool$. This term reduces to a runtime error in plain \gsf, but with the adjusted semantics, it reduces to the function value
$f = \ev[\Bool -> \Bool](\lambda y: \?. (\cast{\ev[\Int -> \Int]}{(\lambda x:\Int. x)} :: \?) (\cast{\ev[\?]}{y} :: \?) :: \Bool -> \Bool$. Sure enough, if $f$ is ever applied, it will always raise an exception; \eg~$f\;\ttt$ fails because $\ev[\Int] \trans{=} \ev[\Bool]$ is not defined. 

We could try to recover the more eager semantics of \gsf by adding a consistency check before applying $\impPolyName$, to determine whether the evidence types involved in the consistent transitivity relation are deemed consistent by using rules such as (Comp-AllR) or (Comp-AllL), and raise an error otherwise.
For instance, program $(\lambda x:\Int. x) :: \? :: \Bool -> \Bool$ would now fail, because the undefinedness of consistent transitivity has nothing to do with implicit polymorphism. While this approach would recover some eagerness, it would still be lazier than plain \gsf. Indeed, the term $(\Lambda X.\lambda x:X. x) :: \? :: \Bool -> \Int$ would still reduce without failure (to an always-failing function).

\myparagraph{Parametricity with Dynamic Implicit Polymorphism}
We conjecture that parametricity still holds in the adjusted semantics with dynamic implicit polymorphism.
Informally, note that in $\impPoly{v}{\ev[2]}{} = t$, the underlying evidence of $v$ either gets pushed inside a type abstraction or function, or it is used in a type instantiation.  This means that the way type names interact with other types during consistent transitivity may only be delayed, and not changed.
If a program raises an error because a non-parametric behavior is detected (e.g. $\pr{\Int, \richtype{\Int}} \trans{=} \pr{\Int, \Int}$), then $\impPolyName$ will never be able to recover or delay that program.

More formally, a first observation is that $\impPolyName$ adapts well-typed values into well-typed terms.
\begin{restatable}{lemma}{dipwelltype}
\label{prop:dipwelltype}
If $\staticgJ[\store, \cdot, \cdot]{v}  \inTermT{\cT}$, and 
$\impPoly{v}{\ev}{} = t$ for some $\ev$, then
$\staticgJ[\store, \cdot, \cdot]{t}  \inTermT{\cT}$
\end{restatable}
This means that the resulting term can be safely plugged in the original evaluation context.

More importantly, given two related values, the application of $\impPolyName$ is either defined for both, or undefined for both.
\begin{restatable}{lemma}{dipsync}
\label{prop:dipsync}
If $\logaproxg[\store; \Delta]{v_1}{v_2}{\cT}$, $\ev |- \ceqrulessimpl{\cT}{\cT'}$, then 
$\impPoly{v_1}{\ev}{}$ is defined if and only if $\impPoly{v_2}{\ev}{}$ is defined.
\end{restatable}
This is important to maintain the failure sensitivity of the logical relation.
Then, given two related values, the application of $\impPolyName$ (if defined) yield related terms.
\begin{restatable}[\textbf{Compatibility-dip}]{proposition}{compgdip}
\label{prop:compgdip}
If $\logaproxg[\store; \Delta]{v_1}{v_2}{\cT}$, $\ev |- \ceqrulessimpl{\cT}{\cT'}$, $\impPoly{v_1}{\ev}{}$ is defined, and $\swellGamma$  then $\logaproxg{\impPoly{v_1}{\ev}{}}{\impPoly{v_2}{\ev}{}}{\cT}$.
\end{restatable}
This last property is key to prove the fundamental property for the adjusted semantics. Intuitively, consider the executions of two related programs, in which a consistent transitivity error is detected in both executions. Then either $\impPolyName$ is undefined for both programs and the result holds, or $\impPolyName$ is defined for both and the resulting terms are related. Finally, as both terms reduce to related values, they can be plugged into contexts that yield related terms (compatibility propositions).

\section{Gradual Existential Types in \gsf}
\label{sec:existentials}

Existential types are the foundation of data abstraction and information hiding: concrete representations of abstract data types are elements of existential types~\citep{mitchellPlotkin:toplas1888,pierce:tapl}. It is well known that existential types can be encoded in terms of universal types~\citep{pierce:tapl}. However, several polymorphic languages~\citep{ahmed:esop2006,ahmedAl:popl2009,neisAl:icfp2009} include both universal and existential types primitively, instead of relying on the encoding. The reason is that proving certain properties, such as representation independence results, is much simpler with direct support for existential types. 

Although some efforts have already been developed to protect data abstraction in a dynamically-typed language~\citep{abadiAl:jfp1995, rossberg:ppdp2003,sumiiPierce:popl2004}, prior work on gradual parametric polymorphism leaves the treatment of existential types as future work~\citep{ahmedAl:icfp2017, toroAl:popl2019}. In this section, we present an extension of \gsf with existential types, dubbed \gsfex. We first briefly review existential types (\S\ref{sec:existentials-intro}) and why a direct treatment is preferable to an encoding (\S\ref{sec:gsfex-poe}). We then informally introduce gradual existential types in action (\S\ref{sec:gsfex-action}) before formally developing \gsfex (\S\ref{sec:gsfex-formal}). Finally, we discuss the metatheory of \gsfex (\S\ref{sec:gsfex-properties}).

\subsection{Existential Types in a Nutshell}
\label{sec:existentials-intro}

An {\em abstract data type} (ADT for short) guarantees that a client can neither guess nor depend on its implementation~\citep{reynolds:83,mitchellPlotkin:toplas1888}.  Formally, an ADT consists of a type name $A$, a concrete representation type $T$, implementations of some operations for creating, querying and manipulating values of type $T$, and an abstraction boundary enclosing the representation and operations~\cite{pierce:tapl}. Thus, an ADT provides a public name to a type but hides its representation. The {\em representation independence} property for an ADT establishes that we can change its representation without affecting clients. This property is a particularly useful application of relational parametricity~\citep{reynolds:83}; we can show that two different implementations of an ADT are contextually equivalent so long as there exists a relation between their concrete type representations that is preserved by their operations. 

Data abstraction is formalized by extending \sysF with existential types, of the form $\tyExists[X][T]$. Elements of an existential type are usually called packages, written $\packExists[T'][t][X][T]$, where $T'$ is the hidden representation type and the term component $t$ has type $T[T'/X]$. The existential elimination construct $\unpackExists$ allows the components of the package to be accessed by a client, keeping the actual representation type hidden.
Packages with different hidden representation types can inhabit the same existential type. Thus, we can implement an ADT in different ways, creating different existential packages. 

For instance, consider a semaphore ADT with three operations: $\mathit{bit}$ to create a semaphore, $\mathit{flip}$ to produce a semaphore in the inverted state, and $\mathit{read}$ to consult the state of the semaphore, as a $\Bool$. We can encode such an ADT as an existential type with a triple:
$$
\semaf \equiv \exists X. \pairtype{X}{\pairtype{(X -> X)}{(X -> \Bool)}}
$$
Alternatively, for readability, we can use a hypothetical record syntax:
$$
\semaf \equiv \tyExists[X][\{\mathit{bit:}\;X, \mathit{flip:}\;X -> X, \mathit{read:}\;X -> \Bool\}]
$$
Below are two equivalent implementations of this $\semaf$ ADT:

\begin{small}
\begin{flalign*}
\begin{array}{ll}
s_1 \equiv \packExistsC[\Bool][v_1][\semaf] & where \;  v_1 \equiv {\{{\mathit{bit} = \true,\; }{{\mathit{flip} = (\lambda x :\Bool. \neg \; x),\; }{\mathit{read} = (\lambda x:\Bool. x)\}}}}\\
s_2 \equiv \packExistsC[\Int][v_2][\semaf] & where \;  v_2 \equiv \{\mathit{bit} = 1,\; \mathit{flip} = (\lambda x:\Int. 1 - x),\; \mathit{read} = (\lambda x:\Int. 0 < x) \}
\end{array}
\end{flalign*}
\end{small}

In the first implementation, the concrete representation type is $\Bool$, and in the second it is $\Int$.  The representation and operations of the $\semaf$ ADT are abstract to a client, in the sense that the representation of $\mathit{bit}$ is hidden, and it can only be manipulated and queried by the operations $\mathit{flip}$ and $\mathit{read}$ . For instance, if we have the expression $\unpackExists[X][x][s][t]$, where $s$ is an implementation of $\semaf$, we can do $(x.\mathit{read}\; (x.\mathit{flip} \; x.\mathit{bit}))$ in the expression $t$, but ${(x.\mathit{read} \; (x.\mathit{flip} \; \true))}$ or ${x.\mathit{bit} == \true}$ are invalid programs that do not typecheck. Note that these programs would not get stuck with $s_1$, but they would crash with $s_2$.

\subsection{Existential Types: Primitive or Encoded?}
\label{sec:gsfex-poe}
Existential types are closely connected with universal types, and in fact they can simply be {\em encoded} in terms of universal types, using the following encoding~\cite{pfpl}:

\begin{small}
\begin{flalign*}
\begin{array}{rlll}
\exists X. T & \equiv & \forall Y. (\forall X. T -> Y) -> Y &\\
\packExists[T'][t][X][T]& \equiv & \Lambda Y. \lambda f: (\forall X. T -> Y). f \;[T']\;t &\\
\unpackExists[X][x][t_1][t_2] & \equiv & t_1 \; [T_2] \; (\Lambda X. \lambda x: T. t_2) \; where \; |- t_1: \exists X. T \; and \; X; x: T |- t_2 : T_2 &  
\end{array}
\end{flalign*}
\end{small}

The intuition behind this encoding is that an existential type is viewed as a universal type taking the overall result type $Y$, followed by a polymorphic function representing the client with result type $Y$, and yielding a value of type $Y$ as result. A package is a polymorphic function taking the client as argument, and unpacking corresponds to applying this polymorphic function.

Therefore, to study gradual existential types in \gsf, one could simply adopt this encoding. However, if we want to reason about interesting properties such as representation independence and free theorems, it is preferable to give meaning to existential types directly.

The benefit of a direct treatment of existential types can already be appreciated in the fully-static setting, with the simple examples of packages $s_1$ and $s_2$ above. Suppose we want to show that $s_1$ and $s_2$ are contextually equivalent, \ie~indistinguishable by any context. To prove this equivalence, it is sufficient to show that the packages are related according to a parametricity logical relation that is sound with respect to contextual equivalence~\citep{reynolds:83}. Using the direct interpretation of existential types, such a proof is considerably easier and more intuitive than using their universal encodings.

The additional complexity of reasoning about existential types via their universal encoding hardly scales to more involved examples. For instance, \citet{ahmedAl:popl2009} prove challenging cases of equivalences in the presence of abstract data types and mutable references, where the encoding would have been a liability; hence their choice of supporting existential types directly. Considering that the \gsf logical relation also involves a number of technicalities (evidence, worlds, etc.), providing direct support for existentials is all the more appealing.

\subsection{Gradual Existential Types in \gsfex}
\label{sec:gsfex-action}
In this section, we show some illustrative examples of gradual existential types in action, highlighting their benefits and expected properties when type imprecision is involved. In particular, we want to dynamically preserve the information hiding property presumed for abstract data types.

\myparagraph{Typed-Untyped Interoperability}
Gradual existential types allow programmers to embed an untyped implementation of a library as a static ADT, by picking the unknown type as the hidden representation type. For instance, if $v_3$ is an untyped record, then $s_3$ below is a gradually well-typed implementation of the $\semaf$ ADT. The translation $\lceil \cdot \rceil$ embeds untyped terms in the gradual language, basically by introducing $\?$ on all binders and constants~\cite{siekTaha:sfp2006}.

\begin{small}
\begin{flalign}
\begin{array}{ll}
&\text{let} \; v_3 = \{\mathit{bit} = 1, \mathit{flip} = (\lambda x. 1 - x),\mathit{read} = (\lambda x. 0 < x)\} \; \text{in} \\
&\text{let} \; s_3 = {\packExistsC[\?][\lceil v_3 \rceil][\semaf]} \; \text{in} \; C[s_3]\\
&\text{where } C \equiv \text{let} \; s = [] \; \text{in} \;
\unpackExistsIn[X][x][{s}] \; (x.\mathit{read} \; (x.\mathit{flip} \; x.\mathit{bit}))
\end{array}
\end{flalign}
\end{small}

The package $s_3$ is essentially a version of the package $s_2$ where types have been erased (replaced with the unknown type). As illustrated later (\S\ref{sec:existentialequivalences}), one can prove in \gsfex that $s_3$ is contextually equivalent to $s_2$ (and hence to $s_1$ as well), using a direct interpretation of gradual existential types. The static client or context $C$, given a package implementation of the $\semaf$ ADT, changes the state of the semaphore and then reads the state. The whole example runs without error, producing $\false$ as the final result.

Of course, we could have associated a package implementation that does not respect the ADT signature. For instance, we define $v'_3$ as a variant of $v_3$, where $\mathit{flip}$ has type $\? -> \Bool$. We obtain the package $s'_3$, which is still gradually well-typed. However, using the package with client $\clienteEx$ results in a runtime type error. The runtime error happens when the $\neg$ operator is applied to $x.\mathit{bit}$, because $\neg$ expects a $\Bool$ argument, but dynamically $\mathit{bit}$ is an $\Int$.

\begin{small}
\begin{flalign}
\begin{array}{l}
\text{let} \; v'_3 = \{\mathit{bit} = 1, \mathit{flip} = (\lambda x. \neg \; x),\mathit{read} = (\lambda x. 0 < x)\} \; \text{in} \\
\text{let} \; s'_3 = {\packExistsC[\?][\lceil v_3' \rceil][\semaf]} \; \text{in} \; C[s'_3]
\end{array}
\end{flalign}
\end{small}

The dual case of typed/untyped interoperability is that of a static package being used in dynamic code. The following example defines the untyped function $g$, which take as arguments the function $f$ and an expression $x$ to be applied to $f$. The function $g$ is applied to the typed components of the package $s_2$, reducing the whole program without error to $\true$.

\begin{small}
\begin{flalign}
\begin{array}{l}
\text{let} \; g = (\lambda f. \lambda x. f \; x) \; \text{in}  \;  \unpackExistsIn[X][x][s_2] \; 
((g \; x.\mathit{read}) \; x.\mathit{bit})
\end{array}
\end{flalign}
\end{small}

Taking the same example, but changing $x.\mathit{bit}$ to the expression 
$(1 :: \?)$ yields a runtime error, because the function $x.\mathit{read}$ is expecting a sealed value, but instead it receives an unsealed $\Int$.

\begin{small}
\begin{flalign}
\begin{array}{l}
 \text{let} \; g = (\lambda f. \lambda x. f \; x) \; \text{in}  \; \unpackExistsIn[X][x][s_2] \; 
((g \; x.\mathit{read}) \; (1 :: \?))
\end{array}
\end{flalign}
\end{small}

\myparagraph{Optimistic Type Checking} The following example shows how the optimistic gradual type checker accepts programs that run without errors, which would be rejected with a static type checker.

\begin{small}
\begin{flalign*}
\begin{array}{l}
\unpackExistsIn[X][x][s_2]\;\\
\text{let} \; f = \lambda z. \text{if}(z) \; \text{then} \; (x. \mathit{flip} :: \?) \; \text{else} \: ((\lambda x: \Int. 1 - x) :: \?) \; \text{in}\\
\text{let} \; v'_2 = \{\mathit{bit} = x.\mathit{bit},\; \mathit{flip} = f \; \true,\; \mathit{read} = x.\mathit{read} \} \; \text{in}\\
\text{let} \; s'_2 =  {\packExistsC[X][v'_2][\semaf]} \; \text{in}\\
\unpackExistsIn[Y][y][s'_2] \; C[s'_2]
\end{array}
\end{flalign*}
\end{small}

The package $s'_2$ is essentially the same as $s_2$---in fact they are equivalent. The function $f$ receives a $\Bool$ argument to decide whether to return the (hidden) $\mathit{flip}$ function from package $s_2$, or a literal (not hidden) function. This program is gradually well-typed because of the ascriptions to the unknown type in the  branches of the conditional. In contrast, a static type system would reject this program (without the $\?$ascriptions in the conditional branches) because the then branch would have type $X -> X$, while the else branch  would have type $\Int -> \Int$. 
The gradual program runs properly, yielding $\false$ as a result.  

Note that if the definition of $\mathit{flip}$ in $v'_2$ would be $f\;\false$, then a runtime error would be raised. 
The error would be produced during the evaluation of the definition of $s'_2$ because $v'_2$ ought to have type $\pairtype{X}{\pairtype{(X -> X)}{(X -> \Bool)}}$, but instead it would have type $\pairtype{X}{\pairtype{(\Int -> \Int)}{(X -> \Bool)}}$. 
Also, keeping $f\;\true$ for $\mathit{flip}$ in $v'_2$, but changing the representation type of the package $s'_2$ to $\?$ would generate a runtime error in the application of $(y.\mathit{flip})$ to $(y.\mathit{bit})$.

\myparagraph{Intrinsic vs. Extrinsic Imprecision}
Another point to take into account is the nature of the imprecision of a term of existential type. As discussed previously regarding universal types (\S\ref{sec:wdgg-gsfe}), the imprecision for existential types can be either intrinsic or extrinsic. The following program is fully static except for the imprecise ascription of $s_2$ to the type {\small$\semafvar \equiv \tyExists[X][\pairtype{X}{\pairtype{(X -> \?)}{(X -> \Bool)}}]$}. Observe that {\small$ \semaf \gprec \semafvar $}.

\begin{small}
\begin{flalign}
\begin{array}{l}
\unpackExistsIn[X][x][s_2 :: \semafvar] \; (x.\mathit{read} \; (x.\mathit{flip} \; (x.\mathit{flip} \; x.\mathit{bit})))
\end{array}
\end{flalign}
\end{small}

Here we are in presence of an ascribed imprecision (\ie~extrinsic), preserving the \gsf property that if we ascribe a static closed term to a less precise type, its behavior is preserved.
Thus, this program runs without error, and evaluates to $\true$. 

Conversely, in the following example, the imprecision is intrinsic due to the imprecise signature {\small$\semafvar$} of the package.

\begin{small}
\begin{flalign}
\begin{array}{l}
\unpackExistsIn[X][x][{\packExistsC[\Int][v_2][\semafvar]}]\; (x.\mathit{flip} \; (x.\mathit{flip} \; x.\mathit{bit})) + 10
\end{array}
\end{flalign}
\end{small}

This program is accepted statically. The function $x.\mathit{flip}$ has type $X -> \?$, which specifies that it has to be applied to a sealed value and could return another sealed value, or in this case, an $\Int$ value. The application $(x.\mathit{flip} \; x.\mathit{bit})$ has type $\?$, and it is used as the argument of $x.\mathit{flip}$ again, optimistically treated as an abstract type. Then, the result of the second application of $x.\mathit{flip}$ is added to $10$, being optimistic again with the result of the function $x.\mathit{flip}$, but this time at type $\Int$. 

In this example, the imprecise signature of the package raises some challenges. We want to enforce that the interaction of gradual components of the package with static components does not reveal hidden information. Taking this into account, the program fails at runtime because of the attempt to use $x.\mathit{flip}$ with both types $X -> X$ and $X -> \Int$. Note that if we allow both behaviors of the function $x.\mathit{flip}$, returning $11$, then we would be revealing that the hidden representation type is $\Int$. Thus, we admit at runtime the first application of $x.\mathit{flip}$, accepted only with the type $X -> \Int$, but it fails in the second application because it receives an $\Int$ instead of a sealed value.

\subsection{Semantics of \gsfex}
\label{sec:gsfex-formal}

In this section, we formally present the design and semantics of \gsfex, an extension of \gsf with existential types that exhibits the behaviors illustrated above. First, we introduce the static language \SPFLex, which is the starting point to apply AGT. Actually, we only apply AGT to the new features in \SPFLex since the others have already been gradualized. Then, we focus on \gsfex, the static and dynamic semantics derived by AGT. Finally, we show the principal properties that \gsfex fulfills.

\begin{figure}[t]
  \begin{small}
  \begin{displaymath}
    \begin{array}{rcll}
    \multicolumn{4}{c}{  
      x \in \Var, X \in \VarType, \alpha \in \TypeName \quad
      \sstore  \in \TypeName \finto \Type,
      \Delta \subset \VarType,
      \Gamma \in \Var \finto \Type
      }\\
      T & ::= & \cdots | \tyExists[X][T] & \text{(types)}\\
      t & ::= & \cdots | \packExists[T][t][X][T] | \unpackExists[X][x][t][t] & \text{(terms)}\\
      v & ::= & \cdots | \packExists[T][v][X][T] & \text{(values)}
    \end{array}   
  \end{displaymath}
 \end{small}
 \begin{small}
 \begin{flushleft}
  \framebox{$\EnvSS t : T$}~\textbf{Well-typed terms}
  \end{flushleft}
  \begin{mathpar}
    \inference[(Tpack)]{\EnvSS[\sstore][\Delta][\Gamma]t : T_1 & \eqrules{T_1}{T[T'/ X]} & \Sigma; \Delta |- T'}{\EnvSS \packExists[T'][t][X][T] : \tyExists[X][T]} \and
    \inference[(Tunpack)]{ \EnvSS t_1 : T_1 & \EnvSS[\sstore][\Delta, X][\Gamma, x : \schemeEx({T_1})]t_2 : T_2 & \Sigma; \Delta |- T_2
    }{\EnvSS \unpackExists : T_2}
  \end{mathpar}
\end{small}
\begin{small}
\begin{flushleft}
\framebox{$\eqrules{T}{T}$}~\textbf{Type equality}
\end{flushleft}
  \begin{mathpar} 
           \inference{\eqrules{T_1}{T_2}[\sstore][\Delta,X]}{\eqrules{\exists X. T_1}{\exists X. T_2}}     \end{mathpar}
\end{small}

\begin{small}
\begin{flushleft}
\framebox{$ \storeeval{t} \longrightarrow \storeeval t$}
~\textbf{Notion of reduction}
\end{flushleft}
\begin{mathpar}
    \storeeval (\unpackExists[X][x][{\packExists[T'][v][X][T]}][t] \longrightarrow  \storeeval[\sstore, \alpha := T'] t[\alpha/X][v/x] \quad\text{where }\alpha \not\in \dom(\sstore) \\ \and
\end{mathpar}
\begin{flushleft}
\framebox{$ \storeeval t \longmapsto \storeeval t$}
~\textbf{Evaluation frames and reduction}
\end{flushleft}
\begin{displaymath}
\begin{array}{rcll}
f & ::= & \cdots| \unpackExists[X][x][{[]}][t] & \text{(term frames)}
\end{array}
\end{displaymath}
\end{small}
 \caption{{\SPFLex: Syntax, Static and Dynamic Semantics (extends Figure~\ref{fig:spfl})}}  
  \label{fig:spflET}
  \label{fig:spfl-syntax-staticsET}  
  \label{fig:spfl-dynET}
\end{figure}

\myparagraph{The Static Language \SPFLex}
We derive \gsfex by applying AGT to \SPFL extended with existential types, called \SPFLex (Figure~\ref{fig:spflET}).
We extend \SPFL statics with the rules (Tpack) and (Tunpack) for a package and its elimination form, which are standard. We augment the definition of type equality to deal with existential types, and use the $\schemeEx$ function to extract the schema of an existential type: 
$\schemeEx({\exists X. {T}}) = T$ and is undefined otherwise.

The dynamic semantics of the unpack constructor is very similar to the type application; also a fresh type name $\alpha$ is generated and bound to the representation type $T'$ in the global type name store $\sstore$. Then, we substitute $\alpha$ (instead of the representation type) and the term component, for the variables $X$ and $x$ in the body of the unpack. Like \SPFL, \SPFLex is also type safe, and all well-typed terms are parametric.

\myparagraph{\gsfex: Statics}
We derive the statics of \gsfex following AGT. As in Section~\ref{sec:spfl-lang}, we first define the syntax of gradual typing, and we give them meaning through the concretization function. Then, we lift the static semantics of the static language to gradual settings using the corresponding abstraction function, which forms a Galois connection. Being consistent with the above, we extend the syntactic category of gradual types $\cT \in \GType$ with existential types:
$$ \cT  ::= \basetype | \cT -> \cT| \forall X. \cT| \pairtype{\cT}{\cT}| X  |\alpha| \? |\Gbox{\exists X. \cT}$$

As usual, the unknown type represents any type, including existential types. We naturally extend the concretization function $\cs$ and abstraction function $\as$ to existential types, preserving the Galois connection established earlier (Proposition~\ref{Galoisconnection}):

\begin{small}
\begin{displaymath}
    \conc{\exists X. \cT} = \{\exists X. T| T \in \conc{\cT}[\store;\Delta,X]\}
    \qquad\qquad
    \abst{\set{\overline{\exists X.T_i}}} = \exists X. \abst{{\set{\overline{T_i}}}}[\store;\Delta, X]
\end{displaymath}
\end{small}

We define in Figure~\ref{fig:gsf-staticsETexistential} the inductive definition of type precision, which is equivalent to Definition~\ref{def:precision} (Proposition ~\ref{PrecisioninductivelyExistential}). As a result, $\exists X.\?$ denotes any existential type,  is more precise than the unknown type and less precise than $\exists X. X -> X$. 
The is straightforward (Figure~\ref{fig:spfl-dynET}), and remains sound and optimal 

With the meaning of gradual types, the \gsfex static semantics follow as usual with AGT. In this case, we need to define the gradual counterpart of the type equality predicate, whose lifting is type consistency. Following Definition~\ref{def:consistency}, we can find in Figure~\ref{fig:gsf-staticsETexistential} an equivalent inductive characterization of type consistency (Proposition~\ref{prop:ConsistencyinductivelyExistential}). Then, we lift functions using abstraction, concretization and Definition~\ref{def:fun-lift}. Our only new function in \SPFLex is $\schemeEx$, whose lifting $\cschemeEx : \GType \rightharpoonup \GType$ is as expected: 

\begin{small}
    \begin{displaymath}
    {\cschemeEx(\exists X. {\cT}) = \cT} \qquad 
    {\cschemeEx(\?) = \?} \qquad
    {\cschemeEx(\cT)\undefinedow}
\end{displaymath}
\end{small}

The gradual typing rules of \gsfex (Figure~\ref{fig:gsf-staticsETexistential}) extend those of \gsf. The new rules are obtained by replacing type predicates and functions with their corresponding consistent liftings in the static typing rules. Observe that Rule (Gpack) uses type consistency instead of type equality so that the implementation term can be of a type that is distinct from, but consistent with the package type (after substituting for the representation type). For example:

\begin{small}
  \begin{displaymath} \packExistsC[\Bool][v_1][\semafvartres] \qquad \text{where }\semafvartres \equiv \tyExists[X][\pairtype{X}{\pairtype{(X -> X)}{(X -> \?)}}]
\end{displaymath}
\end{small}

Here, the type of $v_1$ is $\pairtype{\Bool}{\pairtype{(\Bool -> \Bool)}{(\Bool -> \Bool)}}$, which is more precise than $\semafvartres[\Bool/X]$.

Rule (Gunpack) uses the consistent existential schema function $\cschemeEx$, which allows a term of unknown type to be optimistically treated as a package, and therefore unpacked.
\begin{figure}[t]
\begin{small}
  \begin{displaymath}
    \begin{array}{rcll}
      \multicolumn{4}{c}{  
      x \in \Var, X \in \VarType, \alpha \in \TypeName \quad
      \store  \in \TypeName \finto \GType,
      \Delta \subset \VarType,
      \Gamma \in \Var \finto \GType
      }\\
      \cT & ::= & \cdots | \tyExists[X][\cT] & \text{(gradual types)}\\
      t & ::= &  \cdots | \packExists[\cT][t][X][\cT] | \unpackExists[X][x][t][t] & \text{(gradual terms)}\\
    \end{array}   
  \end{displaymath}
 \end{small}
  \begin{small}
   \begin{flushleft}
  \framebox{$\EnvSG t : \cT$}~\textbf{Well-typed terms}
  \end{flushleft}
  \begin{mathpar}
    \inference[(Gpack)]{\EnvSG[\store][\Delta][\Gamma]t : \cT_1 & \ceqrules{\cT_1}{\cT[\cT'/ X]} & \gtwf{\cT'} }{\EnvSG \packExists[\cT'][t][X][\cT]: \tyExists}
     \and
    \inference[(Gunpack)]{ \EnvSG t_1 : \cT_1 & \EnvSG[\store][\Delta, X][\Gamma, x : \cschemeEx({\cT_1})]t_2 : \cT_2 & \gtwf{\cT_2}
    }{\EnvSG \unpackExists : \cT_2}
  \end{mathpar}
\end{small}
\begin{small}
\begin{flushleft}
\framebox{$\ceqrules{\cT}{\cT}$}~\textbf{Type consistency}
\end{flushleft}
    \begin{mathpar}
           \inference{\ceqrules{\cT_1}{\cT_2}[{\store; \Delta, X}]}{\ceqrules{\exists X. \cT_1}{\exists X. \cT_2}}
    \end{mathpar}
  \end{small}
  \begin{small}
\begin{flushleft}
\framebox{$\tprules{\cT}{\cT}$}~\textbf{Type precision}
\end{flushleft}
\begin{mathpar}
           \inference{\tprules{\cT_1}{\cT_2}[\store;\Delta, X]}{\tprules{\exists X. \cT_1}{\exists X. \cT_2}}
\end{mathpar}
\end{small}
 \caption{\gsfex: Syntax and Static Semantics (extends Figure~\ref{fig:gsf-statics})}
  \label{fig:gsf-staticsETexistential}
\end{figure}

\myparagraph{\gsfex: Dynamics}
We now turn to the dynamic semantics of \gsfex. As we did before, we give the dynamic semantics of \gsfex in terms of a more informative variant called \gsfeex.  In \gsfeex, all values are ascribed, and ascriptions carry evidence.

Figure~\ref{fig:gsfeexistential} presents the syntax, static and dynamics semantics of \gsfeex; essentially those of \gsfe naturally extended with existential types. 
The reduction rule ($R$unpack) specifies the reduction of an unpack expression: we substitute 
a fresh type name $\alpha$ for $X$ in the body of the unpack, as well as a (carefully ascribed) package implementation for $x$. In particular, this rule combines the evidence from the actual implementation term $\cast{\ev_1}{u :: \cT_1}$ with the evidence of the package, substituting the representation type on the left $\cT'$ and the fresh type name $\alpha$ on the right for the type variable $X$. Note that the evidence $\ev$ justifies that the static type of the package declared by the keyword ``as'' is consistent with $\exists X. \cT$. Thus $\evidenceExists{\ev}{\evlift{\cT'}}{\evlift{\alpha}}$ justifies that the static type after the substitution by $\cT'$ is consistent with $\cT[\alpha/X]$. Consequently, the resulting evidence of $\cast{\ev_1 \trans{=} \evidenceExists{\ev}{\evlift{\cT'}}{\evlift{\alpha}}}$ justifies that the type of the implementation term is consistent with $\cT [\alpha/X]$. Failure to justify this judgment produces an error, specifying that the implementation term is not appropriate. This evidence plays a key role in making the implementation term abstract, \ie ensuring information hiding.

\begin{figure}[t]
\begin{small}
  \begin{displaymath}
  \begin{array}{rcll}
    t &::=& \cdots| \cast{\ev}{(\packExists[\cT][t][X][\cT]) :: \cT}| \unpackExists[X][x][t][t] & (\text{terms})\\
    u & ::= & \cdots | \packExistst[\cT'][v][X][\cT] & (\text{raw values})\\
    \end{array}   
  \end{displaymath}
   \end{small}
  \begin{small}
   \begin{flushleft}
  \framebox{$\EnvSG \tu : \cT$}~\textbf{Well-typed terms}
  \end{flushleft}
  \begin{mathpar}
      \inference[(Epack)]{\staticgJ{t} : \cT_1[\cT'/ X] & \gtwf{\cT'} & \ijudgment{\ev}[\store; \Delta]{\exists X.\cT_1}{\cT}}{\staticgJ{{{\cast{\ev}{(\packExists[\cT'][t][X][\cT_1]) :: \cT}}}}  : \cT}
     \and
    \inference[(Eunpack)]{\EnvSG t_1 : \tyExists[X][\cT_1] & \EnvSG[\store][\Delta, X][\Gamma, x : \cT_1]t_2 : \cT_2 & \gtwf{\cT_2}
    }{\EnvSG \unpackExists : \cT_2}
  \end{mathpar}
\end{small} 
\begin{small}  
\begin{flushleft}
\framebox{$\conf{t} \nred \conf{t} \text{ or } \error$}
~\textbf{Notion of reduction}
\end{flushleft}
\begin{equation*}
\text{($R$unpack)} \;\unpackExists[X][x][{\cast{\ev}{\packExistsr[\cT'][\cast{\ev_1}{u :: \cT_1}][X][\cT] :: \tyExists[X][\cT]}}][t] \nred 
\begin{block}
          \begin{cases}
          \conf[\store']{\substTermPaper{X}{\evlift{\alpha}}{t}[(\cast{(\ev[1] \trans{=} \evidenceExists{\ev}{\evlift{\cT'}}{\evlift{\alpha}})}{{u :: \cT[\alpha/X]})/x]}}
            \\
            \text{where } \store' \triangleq \store, \alpha := {\cT'} \text{ for some } \alpha \notin \dom(\store) \\
             \text{and } \evlift{\alpha} = \evliftname_{\store'}(\alpha)\\
          \error \qquad \text{if not defined}
          \end{cases}
          \end{block}
\end{equation*}
\end{small} 
\begin{small}  
\begin{flushleft}
\framebox{$\conf{t} \red \conf{t} \text{ or } \error$}
~\textbf{Evaluation frames and reduction}
\end{flushleft}
\begin{equation*}
      \begin{array}{rcl}
      f & ::= & \cdots| \cast{\ev}{\packExistsr[\cT][{[]}][X][\cT] :: \cT}
      \end{array}
  \end{equation*}
\end{small}
 \caption{{\gsfeex: Syntax, Static and Dynamic Semantics (extends Figure~\ref{fig:gsfe})}}
  \label{fig:gsfeexistential}
\end{figure}

To support the dynamic semantics for existential types, we need to extend the representation of evidence types $\cE$ in \gsfeex, adding $\exists X. \cE$ for existential evidence types.
Additionally, we extend the definitions of consistent transitivity naturally:
consistent transitivity between evidences with existential types simply relies on the underlying schemes:

{\small $$\inference[(ex)]{
        \pr{\cE_{1}, \cE_{2}} \trans{} \pr{\cE_{3}, \cE_{4}} = \pr{\cE'_{1}, \cE'_{2}}
      }
      {\pr{\exists X. \cE_{1}, \exists X. \cE_{2}} \trans{=} \pr{\exists X. \cE_{3}, \exists X. \cE_{4}} = 
      \pr{\exists X. \cE'_{1}, \exists X. \cE'_{2}}}$$}

\myparagraph{Illustration}  We now return to the gradual semaphore implementation $s^{*}_3$, which is the translation of the term $s_3$ from \gsfex to \gsfeex. Remember that all base values in \gsfeex are ascribed to their base types, but for simplicity below, we omit trivial evidences. The following reduction trace illustrates all the important aspects of reduction in \gsfeex:

\begin{small}
\begin{mathpar}
\begin{array}{p{2.9em}p{0.8em}lr}
&&\unpackExists[X][x][\cast{\ev[\semaf]}{\packExistsr[\?][v^{*}_3][\semaf] :: \semaf}][(x.\mathit{read} \; (x.\mathit{flip} \;x.\mathit{bit}))] & \text{\footnotesize initial evidence}\\
  \footnotesize($R$unpack)&$\red^{*}$&
  (\pr{\?\barr\Bool, \richtype{\?}\barr\Bool}
  (\lambda x. 0 < x)::\alpha\barr\Bool) 
& \\

\footnotesize($R$proj$i$)& &((\pr{\?\barr\Int, \richtype{\?}\barr\richtype{\Int}}
  (\lambda x. 1 - x)::\alpha\barr\alpha) \; (\pr{\Int, \richtype{\Int}} 1::\alpha))& \text{\footnotesize consistent transitivity}\\

  \footnotesize($R$app)&$\red$&
  (\pr{\?\barr\Bool, \richtype{\?}\barr\Bool}
  (\lambda x. 0 < x)::\alpha\barr\Bool) (\pr{\Int, \richtype{\Int}} (1 - 1) ::\alpha)
  & \text{\footnotesize unsealing eliminates $\alpha$}\\

\footnotesize($R$op,$R$asc)&$\red^{*}$&
  (\pr{\?\barr\Bool, \richtype{\?}\barr\Bool}
  (\lambda x. 0 < x)::\alpha\barr\Bool) (\pr{\Int, \richtype{\Int}} 0 ::\alpha)
  &\text{\footnotesize the return is sealed}\\
\footnotesize($R$app)&$\red^{}$&
  \ev[\Bool](0 <  0)::\Bool & \text{\footnotesize unsealing eliminates $\alpha$}  \\

\footnotesize($R$op,$R$asc)&$\red^{*}$& 
  \ev[\Bool]\false::\Bool
  & \text{\footnotesize}
\end{array} 
\end{mathpar}
\end{small}

In this example, the initial evidence of the package is fully static. We omit some steps in the reduction, but is crucially to show in the rule ($R$unpack) how the evidence $\evidenceExists{{\ev[\semaf]}}{{\?}}{\richtype{\?}}$ is calculated:

\begin{small}
\begin{mathpar}
\begin{array}{l}
\evidenceExists{{\ev[\semaf]}}{{\?}}{\richtype{\?}} \equiv \pr{\pairtype{\?}{\pairtype{(\? -> \?)}{(\? -> \Bool)}}, \pairtype{\richtype{\?}}{\pairtype{(\richtype{\?} -> \richtype{\?})}{(\richtype{\?} -> \Bool)}}}
\end{array} 
\end{mathpar}
\end{small}

After some application of the rule ($R$proj$i$), the term component is protected by the type name $\alpha$. The application step ($R$app) then gives rise to unsealing evidence to interact with the implementation and sealing evidence to protect the implementation.

\subsection{Properties of \gsfex}
\label{sec:gsfex-properties}

In this section, we summarize the main properties, statics and dynamics, concerning \gsf. We cover the refined criteria for gradual typing and parametricity.

\myparagraph{Static Properties} 
We can show that the \gsfex meet the same static properties as \gsf.
\begin{restatable}[\gsfex: Precision, inductively]{proposition}{PrecisioninductivelyExistential}
\label{PrecisioninductivelyExistential} The inductive definition of type precision given in Figure~\ref{fig:gsf-staticsETexistential} is equivalent to Definition~\ref{def:precision}.
\end{restatable}
\begin{restatable}[\gsfex: Consistency, inductively]{proposition}{ConsistencyinductivelyExistential}
\label{prop:ConsistencyinductivelyExistential} The inductive definition of type consistency given in Figure~\ref{fig:gsf-staticsETexistential} is equivalent to Definition~\ref{def:consistency}.
\end{restatable}
The type system of \gsfex is equivalent to the \SPFLex type system on fully-static terms (Proposition~\ref{prop:static-eqExistential}), where $|-_S$ denote the typing judgment of \SPFLex.

\begin{restatable}[\gsfex: Static equivalence for static terms]{proposition}{StaticequivalenceforstatictermsExistential}
\label{prop:static-eqExistential}
Let $t$ be a static term and $\cT$ a static type ($\cT = T$). We have 
$|-_S t : T$ if and only if $|- t : T$
\end{restatable}
The static semantics of \gsfex satisfy the static gradual guarantee (Proposition~\ref{prop:StaticgradualguaranteeExistential}), where type precision (Def.~\ref{def:precision}) extends naturally to {\em term} precision $\iffullv{ (\S\ref{asec:gsfWithExistential})}$.
\begin{restatable}[\gsfex: Static gradual guarantee]{proposition}{StaticgradualguaranteeExistential}
\label{prop:StaticgradualguaranteeExistential}
Let $t$ and $t'$ be closed \gsfex terms such that $t \gprec t'$ and $|- t : \cT$.
Then $|- t' : \cT'$ and $\cT \gprec \cT'$.
\end{restatable}

\myparagraph{Dynamic Gradual Guarantees}
Not surprisingly, \gsfex does not satisfy the dynamic gradual guarantee (\S\ref{sec:gsf-dgg}) with respect to precision $\gprec$ for existential types.
Let us return to the semaphore implementation $s_1$. Note that
$s_1 \gprec s_4$, where $s_4 = \packExists[\Bool][v_1][X][\pairtype{X}{\pairtype{X -> X}{\? -> \Bool}}]$. If we use these terms in the same context as follows, we will obtain that
$$\unpackExists[X][x][s_1][(x.\mathit{read} \; (x.\mathit{flip} \; x.\mathit{bit}))] \gprec \unpackExists[X][x][s_4][(x.\mathit{read} \; (x.\mathit{flip} \;  x.\mathit{bit}))]$$
However, the term on the left reduces to $2$, while the (less precise) term on the right produces a runtime error because of the attempt to apply the function $\mathit{read}$ (in this case of type $\? -> \Bool$) to a sealed value. 
On the other hand, with a simple extension of strict precision to existential types, \gsfex does satisfy the weaker dynamic gradual guarantee \sdgg (Theorem~\ref{theorem:wdggMGSF}).

\myparagraph{Parametricity}
We establish parametricity for \gsfex by proving parametricity for \gsfeex. We extend the step-indexed logical relation for \gsfe (Figure~\ref{fig:lgr}), adding the interpretation of existential types. Usually, two packages are related if their term components are related under some conditions~\cite{ahmed:esop2006,neisAl:icfp2009}. But in gradual settings, the definition of $\setv{\exists X. \cT}$ is more complex. We start with the classical interpretation of the existential types adapted to our previous logical relation:

\begin{small}
      \begin{flalign*}
      \begin{array}{@{}>{\displaystyle}l@{}>{\displaystyle{}}l@{}}
  \lgrvm{\W}{\cast{\ev[1]}{\packExistsr[\cT_1][v_1][X][\cT'_1] :: \exists X.\tpesub(\cT)}}{\cast{\ev[2]}{\packExistsr[\cT_2][v_2][X][\cT'_2] :: \exists X.\tpesub(\cT)}}{\exists X.\cT}{
  \\& \qquad \quad
          \forall \W' \futureW \W, \alpha. \exists R \in \reln[{\getIdxp}]{\cT_1}{\cT_2}.
          (\W' \extworld (\alpha, \cT_1, \cT_2,  R),
          v_1, v_2) \in  \setv[\tpesub \lcorchete X \mapsto \alpha \rcorchete]{\cT}
        }
           \end{array}
          \end{flalign*}
  \end{small}

Let us focus on some simple and not very interesting programs, but useful to explain our existential types interpretation. For example, if we relate these two package, $\cast{\ev[\exists X.X]}{\packExistsr[\Int][\cast{\ev[\Int]}{1 :: \Int}][X][\cT] :: \exists X.X}$ and $\cast{\ev[\exists X.X]}{\packExistsr[\Bool][\cast{\ev[\Bool]}{\true :: \Bool}][X][\cT] :: \exists X.X}$, under the above definition, then we would have to prove that their component terms, $\cast{\ev[\Int]}{1 :: \Int}$ and $\cast{\ev[\Bool]}{\true :: \Bool}$, are related in $\setv[{\tpesub[X -> \alpha]}]{X}$, which is not true. Keep in mind that for two terms to be related in our logical relationship they must have the same type, and they are related in a type variable if they are related in the type variable substituted by its associated type name. 

Taking the above into account, we change the logical interpretation of existential types slightly.

\begin{small}
      \begin{flalign*}
      \begin{array}{@{}>{\displaystyle}l@{}>{\displaystyle{}}l@{}}
  \lgrvm{\W}{\cast{\ev[1]}{\packExistsr[\cT_1][v_1][X][\cT'_1] :: \exists X.\tpesub(\cT)}}{\cast{\ev[2]}{\packExistsr[\cT_2][v_2][X][\cT'_2] :: \exists X.\tpesub(\cT)}}{\exists X.\cT}{
  \\& \qquad \quad
          \forall \W' \futureW \W, \alpha. \exists R \in \reln[{\getIdxp}]{\cT_1}{\cT_2}. (\W' \extworld (\alpha, \cT_1, \cT_2,  R), \\
          & \qquad \quad
          {\cast{\Gbox{\evidenceExists{\ev[1]}{\evlift{\cT_1}}{\evlift{\alpha}}}}{v_1 :: \Gbox{\tpesub(\cT)[\alpha/X]}}},  {\Gbox{\cast{\evidenceExists{\ev[2] }{\evlift{\cT_2}}{\evlift{\alpha}}}}{v_2 :: \Gbox{\tpesub(\cT)[\alpha/X]}}}) \in  \Gbox{\sett[\tpesub \lcorchete X \mapsto \alpha \rcorchete]{\cT}}
        }
           \end{array}
          \end{flalign*}
  \end{small}
  
First, we establish that two packages are related if their term components ascribed to the existential type body, substituting the fresh type name $\alpha$ by $X$, are related. Second, since we ascribed term components to other types, we need evidence justifying this.
More specifically, we need two evidences that justify $\tpesub(\cT)[\cT_1/X]$ is consistent with  $\tpesub(\cT)[\alpha/X]$ and $\tpesub(\cT)[\cT_2/X]$ is consistent with  $\tpesub(\cT)[\alpha/X]$, respectively. In this sense, we use evidences $\evidenceExists{\ev[1]}{\evlift{\cT_1}}{\evlift{\alpha}}$ and $\evidenceExists{\ev[2]}{\evlift{\cT_2}}{\evlift{\alpha}}$; they are just $\ev[1]$ and $\ev[2]$, substituting representation types in the left and the fresh type name $\alpha$ in the right, by $X$. Note that the combination of these evidences with the internal evidences of the package (term component evidences) through transitivity can fail.

This interpretation of existential types is pretty complete but is not enough. Now, suppose that we have the packages $\cast{\ev[\exists X.\?]}{\packExistsr[\Int][\cast{\ev[\Int]}{1 :: \Int}][X][\cT] :: \exists X.\?}$ and $\cast{\ev[\exists X.\?]}{\packExistsr[\Bool][\cast{\ev[\Int]}{1 :: \Int}][X][\cT] :: \exists X.\?}$. These two packages are very similar; the only difference consist in their representation type. They are related under the above interpretation of existential types, due the fact that we can relate $\cast{\ev[\?]}{(\cast{\ev[\Int]}{1 :: \Int)}}$ and $\cast{\ev[\?]}{(\cast{\ev[\Int]}{1 :: \Int)}}$ under the unknown type. But we do not want to relate these packages. First, it is easy to show that the encodings to universal types of these two packages are not related because they do not behave in the same way. Second, we could use the packages in the same context (\eg if we ascribe them by the type $\exists X. X$) with different behaviors, losing the property that says if two packages are related then they are contextually equivalents. Therefore, we need to be more strict in the definition of when to packages are related. 

Finally, we define the interpretation of existential types as follows:

  \begin{small}
      \begin{flalign*}
      \begin{array}{@{}>{\displaystyle}l@{}>{\displaystyle{}}l@{}}
  \lgrvm{\W}{\cast{\ev[1]}{\packExistsr[\cT_1][v_1][X][\cT'_1] :: \exists X.\tpesub(\cT)}}{\cast{\ev[2]}{\packExistsr[\cT_2][v_2][X][\cT'_2] :: \exists X.\tpesub(\cT)}}{\exists X.\cT}{
  \\& \qquad \quad
          \forall \W' \futureW \W, \alpha. \exists R \in \reln[{\getIdxp}]{\cT_1}{\cT_2}.\Gbox{\forall \ijudgment{\ev}[\getinitialStore; \initialDelta]{\exists X.\cT}{\exists X.\cT}.} \\
          & \qquad \quad
          (\Gbox{(\ev[1] \trans{} \tpesubSI[1]{\ev}) \land (\ev[2] \trans{} \tpesubSI[2]{\ev})})  =>
          (\W' \extworld (\alpha, \cT_1, \cT_2,  R),\\
          & \qquad \quad
          {\cast{\evidenceExists{(\ev[1] \Gbox{\trans{} \tpesubSI[1]{\ev}})}{\evlift{\cT_1}}{\evlift{\alpha}}}{v_1 :: \tpesub(\cT)[\alpha/X]}},  {\cast{\evidenceExists{(\ev[2] \Gbox{\trans{} \tpesubSI[2]{\ev}})}{\evlift{\cT_2}}{\evlift{\alpha}}}{v_2 :: \tpesub(\cT)[\alpha/X]}}) \in  \sett[\tpesub \lcorchete X \mapsto \alpha \rcorchete]{\cT}
        }
           \end{array}   
          \end{flalign*}
  \end{small}

 The representation type of a package in gradual settings act as a pending substitution, which has to make sense for all possible (more precise) existential types. In a static world, we do not have to deal with this problem, because evidence never gains precision, and the initial type checking ensures that the program never fails. For this reason, we extend the interpretation of existential types quantifying in all evidence that justifies $\exists X. \cT$ is consistent with $\exists X. \cT$. The universal quantification in all possible evidence that justifies the above judgment ensures that the representation type will behave correctly for any existential type, more precise than $\exists X. \cT$. We arrive at the same reasoning of the interpretation of existential types, studying the encoding of existential into universal types.

\myparagraph{Representation Independence and Gradual Free Theorems}
\label{sec:existentialequivalences}
We prove the soundness of the logical relation extended with existential types with respect to contextual equivalence.
\begin{restatable}{proposition}{ConetaxtualeqExistential}
\label{prop:ConetaxtualeqExistential}
 If $\logeqg{t_1}{t_2}{\cT}$, then $\contequiv{t_1}{t_2}{\cT}$.
\end{restatable}
With this result, we can return to the semaphore example and show the representation independence for the two different implementations $s_1$ and $s_3$. Let us recall the definition of these packages, where the former uses $\Bool$ as representation type, while the latter uses the unknown type:

\begin{small}
\begin{flalign*}
\begin{array}{ll}
s_1 \equiv \packExistsC[\Bool][v_1][\semaf] & where \;  v_1 \equiv {\{{\mathit{bit} = \true,\; }{{\mathit{flip} = (\lambda x :\Bool. \neg \; x),\; }{\mathit{read} = (\lambda x:\Bool. x)\}}}}\\
 s_3 \equiv {\packExistsC[\?][\lceil v_3 \rceil][\semaf]} & where \; v_3 \equiv \{\mathit{bit} = 1, \mathit{flip} = (\lambda x. 1 - x),\mathit{read} = (\lambda x. 0 < x)\} 
\end{array}
\end{flalign*}
\end{small}

To prove that these two packages are contextually equivalent (Proposition~\ref{prop:equivalencesemimprecise}), it suffices by Proposition~\ref{prop:ConetaxtualeqExistential} to show that each logically approximates the other. 
(Note that to proceed with the proof below, we deal with the tuple-based representation of $\semaf$, since \gsf has no records.)
We prove only one direction, namely ${s_1}\preceq{s_3}: {\semaf}$; the other is proven analogously. Therefore, we are required to show that ${s^{*}_1}\preceq{s^{*}_3}:{\semaf}$, where  $s^{*}_1$ and $s^{*}_3$ are the translation of $s_1$ and $s_3$ from \gsfex to \gsfeex, respectively.
\begin{restatable}[]{proposition}{equivalencesemimprecise}
\label{prop:equivalencesemimprecise}
$\contequivempty{s_1}{s_3}{\semaf}$
\end{restatable}
To prove ${s^{*}_1}\preceq{s^{*}_3}:{\semaf}$, we are required to show that for all $\W$, $(\W, s^{*}_1, s^{*}_3) \in \sett[\emptyset]{\semaf}$. Therefore, we have to prove that $|-{s^{*}_i}\inTermT{\semaf}$ (but this is already proven) and $(\W, s^{*}_1, s^{*}_3) \in \setv[\emptyset]{\semaf}$ (since $s^{*}_i$ are already values). Expanding the definition of $\setv[\emptyset]{\semaf}$, we need to show that $\forall \W' \futureW \W$ and $\alpha$, $\exists R \in \reln[{\getIdxp}]{\Bool}{\?}$, such that $\forall \ijudgment{\ev}[\cdot;\cdot]{\semaf}{\semaf}$:
$$(W'', {\cast{\evidenceExists{(\ev[\semaf] {\trans{} {\ev}})}{\Bool}{\evlift{\alpha}}}{v^{*}_1 :: \cT[\alpha/X]}},  {\cast{\evidenceExists{(\ev[\semaf] {\trans{}{\ev}})}{\?}{\evlift{\alpha}}}{v^{*}_3 :: \cT[\alpha/X]}}) \in  \sett[\lcorchete X \mapsto \alpha \rcorchete]{\cT}$$
where $W'' = \W' \extworld (\alpha, \Bool, \?,  R)$, $\cT = \cschemeEx(\semaf) = \pairtype{X}{\pairtype{(X -> X)}{(X -> \Bool)}}$, $s^{*}_1 = \cast{\ev[\semaf]}{\packExistsr[\Bool][v^{*}_1][][] :: \semaf}$ and $s^{*}_2 = \cast{\ev[\semaf]}{\packExistsr[\?][v^{*}_2][][] :: \semaf}$. Since $\ev[\semaf]$ is an static evidence, it can not gain precision and so $(\ev[\semaf] {\trans{} {\ev}}) = \ev[\semaf]$. Therefore, now we are required to show
$$(\downstep W'', v_1', v_3') \in \setv[\lcorchete X \mapsto \alpha \rcorchete]{\cT}$$
where

\begin{small}
\begin{flalign*}
\begin{array}{ll}
v'_1 = & \cast{\pr{\pairtype{\Bool}{\pairtype{(\Bool -> \Bool)}{(\Bool -> \Bool)}}, \pairtype{\richtype{\Bool}}{{(\pairtype{\richtype{\Bool} -> \richtype{\Bool})}{(\richtype{\Bool} -> {\Bool})}}}}}{}\\
&\pair{\true}{\pair{(\lambda x:\Bool. \neg \; x)}{(\lambda x:\Bool. x)}} :: \pairtype{\alpha}{\pairtype{\alpha -> \alpha}{\alpha -> \Bool}}\\
v'_3 = & \cast{\pr{\pairtype{\Int}{\pairtype{(\? -> \Int)}{(\? -> \Bool)}}, \pairtype{\richtype{\Int}}{\pairtype{(\richtype{\?} -> \richtype{\Int})}{(\richtype{\?} -> {\Bool})}}}}{}\\
&\pair{1}{\pair{(\lambda x. 1 - x)}{(\lambda x. 0 < x)}} :: \pairtype{\alpha}{\pairtype{(\alpha -> \alpha)}{(\alpha -> \Bool)}}
\end{array}
\end{flalign*}
\end{small}

Taking $R = \{(W^{*}, \cast{\ev[\Bool]}{\true :: \Bool} ,\cast{\ev[\Int]}{1 :: \?}), (W^{*}, \cast{\ev[\Bool]}{\false :: \Bool} ,\cast{\ev[\Int]}{0 :: \?}) | W^{*} \futureW \W'\}$, it is easy to show that

\begin{small}
\begin{flalign*}
\begin{array}{@{}>{\displaystyle}l@{}>{\displaystyle{}}l@{}}
-\;(\downstep^{2}W'',& \; \cast{\pr{\Bool, \richtype{\Bool}}}{\true :: \alpha} ,\cast{\pr{\Int, \richtype{\Int}}}{1 :: \alpha}) \in \setv[\lcorchete X \mapsto \alpha \rcorchete]{X}\\
-\;(\downstep^{2}W'',& \; \cast{\pr{\Bool -> \Bool,\richtype{\Bool} -> \richtype{\Bool}}}{(\lambda x:\Bool. \neg \; x):: \alpha -> \alpha}, \\
& \;  \cast{\pr{\? -> \Int,\richtype{\?} -> \richtype{\Int}}}{(\lambda x. 1 - x) :: \alpha -> \alpha}) \in \setv[\lcorchete X \mapsto \alpha \rcorchete]{X -> X}\\
-\;  (\downstep^{2}W'', & \;  \cast{\pr{\Bool -> \Bool,\richtype{\Bool} -> {\Bool}}}{(\lambda x:\Bool. x):: \alpha -> \Bool},\\ 
& \;  \cast{\pr{\? -> \Bool,\richtype{\?} -> {\Bool}}}{(\lambda x. 0 < x) :: \alpha -> \Bool}) \in \setv[\lcorchete X \mapsto \alpha \rcorchete]{X -> \Bool}
\end{array}
\end{flalign*}
\end{small}

Note that $\downstep^{2}W'' \futureW \W'$. Thus, the result follows immediately.

\section{Related Work}
\label{sec:related}

We have already discussed at length related work on gradual parametricity, especially the most recent developments~\citep{ahmedAl:icfp2017,igarashiAl:icfp2017,xieAl:esop2018,newAl:popl2020}, highlighting the different design choices, properties and limitations of each. Hopefully our discussions adequately reflect the many subtleties involved in designing a gradual parametric language.

The relation between parametric polymorphism in general and dynamic typing much predates the work on gradual typing. \citet{abadiAl:toplas1991} first note that without further precaution, type abstraction might be violated. Subsequent work explored different approaches to protect parametricity, especially runtime-type generation (RTG) \citep{leroyMauny:fpca1991,abadiAl:jfp1995,rossberg:ppdp2003}. \citet{sumiiPierce:popl2004} and \citet{guhaAl:dls2007} use dynamic sealing, originally proposed by \citet{morris:cacm1972}, in order to dynamically enforce type abstraction. \citet{matthewsAhmed:esop2008} also use RTG in order to protect polymorphic functions in an integration of Scheme and ML. This line of work eventually led to the polymorphic blame calculus~\cite{ahmedAl:popl2011} and its most recent version with the proof of parametricity by \citet{ahmedAl:icfp2017}. We adapt their logical relation to the evidence-based semantics of \gsf.

\citet{houAl:jfp2016} prove the correctness of compiling polymorphism to dynamic typing with embeddings and partial projections; the compilation setting however differs significantly from gradual typing. \citet{newAhmed:icfp2018} use embedding-projection pairs to formulate a semantic account of the dynamic gradual guarantee, coined graduality, in a language with explicit casts. Inspired by the work of \citet{neisAl:icfp2009} on parametricity in a non-parametric language, they extended their approach to gradual parametricity, yielding the \polyG language design with explicit sealing~\cite{newAl:popl2020}, discussed at length in this article.

\citet{devrieseAl:popl2018} disprove a conjecture by \citet{pierceSumii:2000} according to which 
the compilation of \sysF to a language with dynamic sealing primitives is fully abstract, \ie~preserves contextual equivalences. 
They show that, for similar reasons, the embedding of \sysF in a polymorphic blame calculus like \lamB is not fully abstract; their observation also applies to \gsf.
Full abstraction might be too strong a criteria for gradual typing: already in the simply-typed setting, embedding typed terms in gradual contexts is not fully abstract, because gradual types admit non-terminating terms. Exploring full abstraction as yet another possible criterion for a gradual language is an interesting perspective, at least to precisely characterize what such a language ensures to programmers.
In our previous work~\cite{toroAl:popl2019} we show that \gsf satisfies an {\em imprecise termination} property, which is a weaker yet useful result that sheds light on gradual free theorems about imprecise type signatures. In this revised article, we have extended the result by precisely characterizing the space of imprecision evolutions that are harmless.

This work is generally related to gradualization of advanced typing disciplines, in particular to gradual information-flow security typing~\citep{disneyFlanagan:stop2011,fennellThiemann:csf2013,fennellThiemann:ecoop2016,garciaTanter:arxiv2015,toroAl:toplas2018}. In these systems, one aims at preserving {\em noninterference}, \ie~that private values dot not affect public outputs. Both parametricity and noninterference are 2-safety properties, expressed as a relation of two program executions. While \citet{garciaTanter:arxiv2015} show that one can derive a pure security language with AGT that satisfies both noninterference and the dynamic gradual guarantee, \citet{toroAl:toplas2018} find that in presence of mutable references, one can have either the dynamic gradual guarantee, or noninterference, but not both. Also similarly to this work, AGT for security typing needs a more precise abstraction for evidence types (based on security {\em label intervals}) in order to enforce noninterference. Together, these results suggest that 
type-based approaches to gradual typing are in tension with semantically-rich typing disciplines. Solutions might come from restricting the considered {\em syntax}, as in \polyG in the context of parametricity, or the {\em range} of graduality, as recently established by \cite{azevedo:lics2020} in the context of noninterference, where the dynamic end of the spectrum is not fully untyped security-wise.

\section{Conclusion}
\gsf is a gradual parametric language that bridges between \sysF and an untyped language with dynamic sealing primitives. The spectrum between both extremes is fairly continuous, even if not perfectly: because of the implicit type-driven resolution of sealing in its runtime semantics, which appears necessary in order to respect the syntax of \sysF, some evolution scenarios towards imprecision can trigger failure. We precisely characterize the weaker continuity that \gsf supports, along with all its other properties. We also study several extensions towards a more full-fledged and practical language, in particular with a novel dynamic support for implicit polymorphism, and existential types for gradual data abstraction. The design of \gsf is largely driven by the Abstracting Gradual Typing (AGT) methodology. We find that AGT greatly streamlines the static semantics of \gsf, but does not yield a language that respects parametricity by default; non-trivial exploration was necessary to uncover how to strengthen the structure and treatment of runtime evidence in order to recover a notion of gradual parametricity. In turn, this strengthening broke the dynamic gradual guarantee in specific scenarios of loss of precision. 

This work focuses on the semantics and meta-theoretical properties of \gsf, without explicitly taking into account efficiency considerations such as pay-as-you-go~\citep{siekTaha:sfp2006,igarashiAl:icfp2017}, space efficiency~\citep{hermanAl:hosc10,siekAl:popl10}, cast elimination~\citep{rastogi:popl2012}, etc. Optimizing the dynamic semantics of \gsf is left for future work. Likewise, blame tracking has not been considered. Tracking blame in order to report more informative error messages is valuable, but most important is to properly {\em identify} error cases. As extensively discussed, gradual parametricity is subtle, and there are many scenarios when the decision of failing or not is open to debate and various considerations. This work contributes to this discussion by proposing several practical principles, with which related languages with blame tracking do not concur. We expect on-going work by colleagues on incorporating blame into AGT to be directly applicable to \gsf, because there does not seem to be any parametricity-specific challenges related to blame.

More importantly, this work recognizes two main trends in the design of gradual parametric languages: those based on \sysF, like \lamB, \csa, \sysFg and \gsf, and those that depart from that syntax, like \polyG. We believe that \gsf goes beyond prior work in the \sysF trend. Also, we have argued that while \polyG enjoys a stronger metatheory than languages from the other trend, several limitations regarding modularity and abstraction caused by its use of explicit sealing are not benign. The question remains open of whether there is a third way, embracing both \sysF and a fully satisfying metatheory.

\begin{acks}
We thank Amal Ahmed, Dominique Devriese, Kenji Maillard, Max New, Gabriel Scherer, the attendees of various oral presentations of this work, and the anonymous reviewers, for useful feedback and suggestions that improved both the presentation and our understanding of this work.
\end{acks}

\bibliography{_Bib/strings,_Bib/pleiad,_Bib/bib,_Bib/common}

\iffullv{\input{appendix}}
\end{document}

%% file: definitions.tex
\usepackage{array}
\usepackage{enumitem}
\usepackage{amsmath}

\makeatletter
\newcommand\tableofcontentsA{%
    \section*{\contentsname
        \@mkboth{%
           \MakeUppercase\contentsname}{\MakeUppercase\contentsname}}%
    \@starttoc{toca}%
    }
\makeatother

\ExplSyntaxOn
\DeclareExpandableDocumentCommand{\IfNoValueOrEmptyTF}{mmm}
 {
  \IfNoValueTF{#1}{#2}
   {
    \tl_if_empty:nTF {#1} {#2} {#3}
   }
 }
\ExplSyntaxOff

\declaretheorem[style=remark,numbered=no]{case}

 \DeclareDocumentCommand{\impR}{O{\Omega} m m m m O{\sprec}}{#1 |- #2 #6 #3: #4 \sprec #5}
 \DeclareDocumentCommand{\impRu}{O{\Omega} m m m m O{\seprecv}}{#1 |- #2 #6 #3: #4 \sprec #5}
  \DeclareDocumentCommand{\impRev}{O{\Omega} m m m m}{#1 |- #2 \sprec #3: #4 \sprec #5}
\DeclareDocumentCommand{\similar}{m m}{#1 \triangleright #2}
\newcommand{\seprecv}{\sprec_{\hspace{-0.1em}v\hspace{0.1em}}}

\DeclareDocumentCommand{\impRel}{m m m m O{\W} O{\tpesub}}
{{({#5}, {#1}, {#2}) \in \mathcal{V}_{#6}\llbracket#3 \gprec #4\rrbracket}}

\DeclareDocumentCommand{\inImpDef}{m m O{\tpesub}}{\text{Imp}_{#3} [#1 \gprec #2]}

\DeclareDocumentCommand{\setvIR}{m m O{\tpesub}}{{\mathcal{V}_{#3}\llbracket#1 \gprec #2\rrbracket}}

\newcommand{\expandTypeSy}{\mathit{ex}}
\DeclareDocumentCommand{\expandType}{m O{}}{\expandTypeSy({#1}, {#2})}

\DeclareDocumentCommand{\setnIP}{m m O{\tpesub}}{\mathcal{V}_{#3}\llbracket {#1} \sim {#2} \rrbracket}

\DeclareDocumentCommand{\setcIP}{m m O{\tpesub}}{\mathcal{C}_{#3} \llbracket{#1} \sim {#2}\rrbracket}

\DeclareDocumentCommand{\inImpIS}{m m O{\tpesub} }{\text{Imp}_{#3} \llbracket{#1} \sim {#2}\rrbracket}

\DeclareDocumentCommand{\setnSP}{ m m O{\tpesub}}{{\mathcal{V}_{#3}\llbracket#1 \gprec #2\rrbracket}}
\newcommand{\wo}{W}

\DeclareDocumentCommand{\setcSP}{ m m O{\tpesub}}{{\mathcal{C}_{#3}\llbracket#1 \gprec #2\rrbracket}}

\DeclareDocumentCommand{\inStaticSP}{m m  O{\tpesub} }{\text{Imp}_{#3} [#1 \gprec #2]}

\DeclareDocumentCommand{\staticGSP}{ m m m O{\wo.\sstore} O{\tpesub} }{ \uvalStatic(#1) \land \proj{2}[\ev] = \enrich{#5(#2)}{#4}}

\DeclareDocumentCommand{\lrwggSP}{O{\sstore; \Delta; \Gamma} m m m}{#1 \Vdash #2 : #3 \gprec  #4}

\DeclareDocumentCommand{\setggSP}{ m O{\tpesub}}{{\mathcal{G}_{#2}\llbracket#1\rrbracket}}

\DeclareDocumentCommand{\atomvalueSP}{m O{\rho}}{\text{Atom}^{=}_{#2}[#1]}
\DeclareDocumentCommand{\atomtermSP}{m O{\rho}}{\text{Atom}_{#2}[#1]}
\DeclareDocumentCommand{\atomSP}{m O{n}}{\text{Atom}_{#2}[#1]}
\DeclareDocumentCommand{\atomvSP}{m O{n}}{\text{Atom}^\text{val}_{#2}[#1]}

\newcommand{\normyName}{\mathit{norm}}
\DeclareDocumentCommand{\normy}{m O{}}{\normyName({#1}, {#2})}

\newcommand{\dynLang}{\textnormal{${\lste{\lambda_{\footnotesize{\textsf{dyn}}}}}$}\xspace}
\newcommand{\lambdaSeal}{\textnormal{${\lste{\lambda_{\footnotesize{\textsf{seal}}}}}$}\xspace}
\DeclareDocumentCommand{\sealedC}{m m}{\lste{\{#1\}_{#2}}}
\DeclareDocumentCommand{\unsealC}{m m m m}{\mathsf{let} \; \{#1\}_{#2} = {#3} \; \mathsf{in} \; {#4}}

\DeclareDocumentCommand{\letC}{m m m}{\mathsf{let} \; {#1} = {#2} \; \mathsf{in} \; {#3}}
\DeclareDocumentCommand{\letT}{m m m}{\mathsf{let} \; {#1} = {#2} \; \mathsf{in} \; {#3}}
\DeclareDocumentCommand{\sealC}{m m}{\lste{\nu {#1}.{#2}}}
\newcommand{\heap}{\lste{\mu}}
\newcommand{\heapp}{\lste{\mu'}}
\newcommand{\seal}{\lste{\sigma}}
\newcommand{\Seal}{\oblset{Seal}}

\DeclareDocumentCommand{\confLS}{O{\heap} m }{#1 \triangleright #2}

\newcommand{\su}{\mathit{su}}
\newcommand{\sus}{\su^{\sigma}}

\newcommand{\sue}{\mathit{su}_{\ev}}

\newcommand{\susep}{\su^{\sigma'}_{\ev}}
\newcommand{\suse}{\su^\sigma_{\ev}}
\newcommand{\sylogeqLS}{\approx}
\newcommand{\lsealcolor}{propcolor}

\DeclareDocumentCommand{\storeevalLS}{O{\heap} m}{\lste{{#2}  \;\|\; {#1}}}

\DeclareDocumentCommand{\lste}{m}{{\color{\lsealcolor}{#1}}}

\DeclareDocumentCommand{\simulationRelT}{m m m O{\lste{\heap};\store; \Gamma}}{ {#4} |- {\color{\lsealcolor}{#1}}  \approx {#2} : {#3}}
\DeclareDocumentCommand{\simulationRel}{m m m O{\lste{\heap};\store}}{ {#4} |- {\color{\lsealcolor}{#1}}  \approx {#2} : {#3}}

\newcommand{\gT}{D}

\newcommand{\gE}{E}
\newcommand{\stError}{\lste{\textup{\textbf{error}}}\xspace}
\newcommand{\typeError}{\lste{\textup{\textbf{type\_error}}}\xspace}
\newcommand{\sealError}{\lste{\textup{\textbf{seal\_type\_error}}}\xspace}
\newcommand{\unsealError}{\lste{\textup{\textbf{unseal\_error}}}\xspace}

\DeclareDocumentCommand{\compileLSA}{m}{\lceil {\color{\lsealcolor}{#1}} \rceil}

\DeclareDocumentCommand{\terminationGSF}{O{t}}{{#1}\terminationsy}  
\DeclareDocumentCommand{\terminationLS}{m O{\heap} O{v}}{\lste{{#1} \terminationsy {#3}  \;\|\; {#2}}}  
\DeclareDocumentCommand{\terminationLSError}{m O{\heap} O{v}}{\lste{{#1 \terminationsy   \; {\unsealError}}}}
\DeclareDocumentCommand{\terminationLSTypeError}{m O{\heap} O{v}}{\lste{{#1 \terminationsy   \; {\typeError}}}}\DeclareDocumentCommand{\terminationLSTError}{m O{\heap} O{v}}{\lste{{#1 \terminationsy   \; {\stError}}}}              
\DeclareDocumentCommand{\divergeLS}{m }{\lste{{#1}\divergesy}}

\definecolor{mygray}{HTML}{800000}

\DeclareDocumentCommand{\initialDelta}{O{\dom(\tpesub)}}{#1}
\DeclareDocumentCommand{\evidenceExists}{m m m}{{#1}[{#2}, {#3}]}
\DeclareDocumentCommand{\getinitialStore}{O{\W}}{#1\!.\store}

\DeclareDocumentCommand{\projectunliftliftEv}{m O{} }{{\IfNoValueTF{#2}{{\pi^{*}_{#1}}}{{\pi^{*}_{#1}({#2})}}}}

\DeclareDocumentCommand{\projectCompEv}{m O{} }{{\IfNoValueTF{#2}{{\pi^{2}_{#1}}}{{\pi^{2}_{#1}({#2})}}}}

\DeclareDocumentCommand{\tpgsfeVSim}{m m m m m O{\sstore} O{\Delta} O{\Gamma}}{#6; #7; #8 |- #1 : #3 \joatsyExttVSim[#5] #2 : #4}

\DeclareDocumentCommand{\gsfTypingImp}{m m m O{\sstore} O{\Delta} O{\Gamma}}{#4; #5; #6 |- #1 : #2 \gprec #3  }

\newcommand{\joatsyExttVSim}[1][T]{\approx_{#1}}

\DeclareDocumentCommand{\swellGammaTPSim}{O{\store} O{\Delta}  O{\Gamma}}{#1; #2 |-   #3}

\DeclareDocumentCommand{\typingctxImp}{O{T} O{T'} O{\ctx}  O{\sstore; \Delta; \Gamma} O{\sstore'; \Delta'; \Gamma'} O{\cT} O{\cT'}}{\vdash #3: (#4 \vdash #1 \gprec #6)\tcsy (#5 \vdash #2  \gprec #7)}

\DeclareDocumentCommand{\typingctxImpP}{O{T} O{T''} O{\ctx} O{\cT} O{\cT''} O{\sstore; \Delta; \Gamma} O{\sstore'; \Delta'; \Gamma'}}{\vdash #3: (#6 \vdash #1 \gprec #4)\tcsy (#7 \vdash #2  \gprec #5)}

\DeclareDocumentCommand{\contequivImp}{m m m m O{T}  O{\sstore; \Delta; \Gamma}}{{#6} \vdash {#1} : {#3} \cesyimp{#5} {#2} : {#4}}

\DeclareDocumentCommand{\contequivImpempty}{m m m m O{T}}{{#1} : {#3} \cesyimp{#5} {#2} : {#4}}

\DeclareDocumentCommand{\contaproxImp}{m m m m  O{\store; \Delta; \Gamma}}{{#5} \vdash #1 \casy #2 : {#3} \gprec {#5}}

\newcommand{\cesyimp}[1]{\approx^{\mathit{ctx}}_{#1}}

\newcommand{\iSWstore}{\sstore}

\DeclareDocumentCommand{\igetSStore}{O{} O{}}{{\W_{#2}\!.\iSWstore_{#1}}}
\DeclareDocumentCommand{\igetSStorep}{O{} O{}}{{\W'_{#2}\!.\iSWstore_{#1}}}
\DeclareDocumentCommand{\igetSStorepp}{O{} O{}}{{\W''_{#2}\!.\iSWstore_{#1}}}
\DeclareDocumentCommand{\igetSStoreppp}{O{} O{}}{{\W'''_{#2}\!.\iSWstore_{#1}}}

\DeclareDocumentCommand{\iigetStorew}{O{} O{} }{{\W_{#1}\!.\iSWstore_{#2}}}
\DeclareDocumentCommand{\iigetStorewp}{O{} O{} }{{\W'_{#1}\!.\iSWstore_{#2}}}
\DeclareDocumentCommand{\iigetStorewpp}{O{} O{} }{{\W''_{#1}\!.\iSWstore_{#2}}}
\DeclareDocumentCommand{\iigetStorewppp}{O{} O{} }{{\W'''_{#1}\!.\iSWstore_{#2}}}
\DeclareDocumentCommand{\iigetStorewpppp}{O{} O{} }{{\W''''_{#1}\!.\iSWstore_{#2}}}

\DeclareDocumentCommand{\isetv}{m O{\cT} O{\rho}}{\mathcal{N}_{#3} \llbracket #1 \gprec  #2\rrbracket}
\DeclareDocumentCommand{\iisetv}{O{\rho} m }{\mathcal{N}_{#1} \llbracket #2 \gprec  #2\rrbracket}

\DeclareDocumentCommand{\isett}{m O{\cT} O{\rho}}{\mathcal{C}_{#3} \llbracket #1 \gprec  #2\rrbracket}

\DeclareDocumentCommand{\iisett}{O{\rho} m }{\mathcal{C}_{#1} \llbracket #2 \gprec  #2\rrbracket}

\newcommand{\isetg}[2][\rho]{\mathcal{G}_{#1}\llbracket#2\rrbracket}

\DeclareDocumentCommand{\iatomvalue}{m m O{\rho}}{\text{Imp}_{#3}[#1 \gprec #2]}
\DeclareDocumentCommand{\iatomterm}{m m O{\rho}} {\text{Imp}_{#3}^{\mathit{t}}[#2]}
\DeclareDocumentCommand{\iatom}{m m O{n}}{\text{Imp}_{#3}[#1, #2]}
\DeclareDocumentCommand{\iatomv}{m m O{n}}{\text{Imp}^\text{val}_{#3}[#1, #2]}
\DeclareDocumentCommand{\ilgrp}{m m m }{(#1, #2, #3)}
\DeclareDocumentCommand{\ilgrpp}{m m m m m O{\rho}} {(#1, #2, #3) \in \iatomvalue{#4}{#5}[#6]}
\DeclareDocumentCommand{\ilgrpt}{m m m m m O{\rho}} {(#1, #2, #3) \in \iatomterm{#4}{#5}[#6]}
\DeclareDocumentCommand{\ilgrvm}{m m m m m O{\rho} O{\cT}}{\isetv{#4}[#7][#6] &=\quad \{\ilgrpp{#1}{#2}{#3}{#4}{#7}[#6] ~|~ #5 \}}
\DeclareDocumentCommand{\ilgrvmalpha}{m m m m m O{\rho} O{\cT}}{\isetv{#4}[#7][#6] &=\quad \{(#1, #2, #3) \in \iatomvalue{#4}{#7}[#6] ~|~ #5 \}}
\DeclareDocumentCommand{\ilgrtm}{m m m m m O{\rho} O{\cT}}{\isett{#4}[#7][#6] &=\quad \{\ilgrpt{#1}{#2}{#3}{#4}{#7}[#6] ~|~ #5 \}}

\DeclareDocumentCommand{\ilgrt}{m m m m O{\rho} }{\ilgrp{#1}{#2}{#3} \in \isett{#4}[#5]}
\DeclareDocumentCommand{\ilgrv}{m m m m O{\rho} }{\ilgrp{#1}{#2}{#3} \in \isetv{#4}[#5]}
\DeclareDocumentCommand{\ilgrg}{m m m m O{\rho} }{\ilgrp{#1}{#2}{#3} \in \isetg[#5]{#4}}

\DeclareDocumentCommand{\precEvGG}{O{\ev[1]} O{\ev[2]} O{\Delta} O{\store} O{\store'}}{#4 \vdash {#1}\ \joatsyExttV[#3] \ #5 \vdash {#2}}

\DeclareDocumentCommand{\precEvGGExt }{m m O{\Delta} O{\store} O{\store'}}{#4 \vdash{#1}\ \joatsyExttV[#3] \ #5 \vdash {#2}}

\newcommand{\joatsy}{{\mathrel{\ooalign{\raise0.2ex\hbox{$\sqsubset$}\cr\hidewidth\raise-0.9ex\hbox{\scalebox{0.9}{$\sim$}}\hidewidth\cr}}}}

\DeclareDocumentCommand{\joatsyExtt}{O{\Delta} O{\Gamma}}{\mathrel{%
  \ooalign{\raise0.2ex\hbox{$\sqsubset$}\cr\hidewidth\raise-0.8ex\hbox{\scalebox{0.9}{$\sim$}}\hidewidth\cr}}_{#1; #2}}

\DeclareDocumentCommand{\joatsyExttV}{O{\Delta}}{\mathrel{%
  \ooalign{\raise0.2ex\hbox{$\sqsubset$}\cr\hidewidth\raise-0.8ex\hbox{\scalebox{0.9}{$\sim$}}\hidewidth\cr}}_{#1}}

\DeclareDocumentCommand{\syctaprocjack}{O{\Delta} O{\Gamma}}{\mathrel{%
  \ooalign{\raise0.2ex\hbox{$\prec$}\cr\hidewidth\raise-0.8ex\hbox{\scalebox{0.9}{$\sim$}}\hidewidth\cr}}_{#1; #2}}

 \newcommand\joatsysource{\mathrel{%
  \ooalign{\raise0.2ex\hbox{$\prec$}\cr\hidewidth\raise-0.8ex\hbox{\scalebox{0.9}{$\sim$}}\hidewidth\cr}}_\mathit{s}}

\DeclareDocumentCommand{\swellGammaTP}{O{\store}O{\store'}  O{\Delta}  O{\Gamma} O{\Gamma'}}{#1; #2; #3 |-   #4; #5}  

\DeclareDocumentCommand{\tpgsfe}{m m m m O{\store} O{\store'} O{\Delta} O{\Gamma} }{#5 |- #1 : #3 \joatsyExtt[#7][#8] #6 |- #2 : #4}

\DeclareDocumentCommand{\tpgsfeV}{m m m m O{\store} O{\store'} O{\Delta} O{\Gamma} O{\Gamma'}}{#5; #8 |- #1 : #3 \joatsyExttV[#7] #6; #9 |- #2 : #4}

\DeclareDocumentCommand{\ctaprocjack}{m m m m O{\store} O{\store'} O{\Delta} O{\Gamma}}{#5 |- #1 : #3 \syctaprocjack[#7][#8] #6 |- #2 : #4}

\DeclareDocumentCommand{\tpgsfeG}{m m O{\store} O{\store'}}{#3 |- #1 \joatsy_{\Delta} #4 |- #2}

\DeclareDocumentCommand{\tpgsfeE}{m m O{\store} O{\store'}}{ #1 \joatsy #4}

\newcommand{\clienteEx}{C}

\newcommand{\semafvar}{\mathit{Sem}_1}

\newcommand{\semafvartres}{\mathit{Sem}_3}
\newcommand{\semaf}{\mathit{Sem}}
\DeclareDocumentCommand{\evinn}{O{\ev}  O{}}{#1[#2\mathit{in}]}
\DeclareDocumentCommand{\evoutE}{O{\ev}  O{}}{#1[#2\mathit{out}]}
\DeclareDocumentCommand{\evoutEE}{O{\ev}  O{}  O{}}{#1_{#2\mathit{out}}^{#3}}

\DeclareDocumentCommand{\cEUO}{O{}}{\cE_{#1*}}

\DeclareDocumentCommand{\tyExists}{O{X} O{\cT}}{\exists #1. #2}

 \DeclareDocumentCommand{\unpackExists}{O{X} O{x} O{t_1} O{t_2}}{\textsf{unpack} \langle #1, #2 \rangle = #3 \ \textsf{in}  \ #4}

\DeclareDocumentCommand{\unpackExistsIn}{O{X} O{x} O{t_1}}{\textsf{unpack} \langle #1, #2 \rangle = #3 \ \textsf{in}}

\DeclareDocumentCommand{\packExists}{O{\cT'} O{v} O{X} O{\cT} }{\textsf{pack}\langle#1, #2\rangle \ \textsf{as} \ \exists #3. #4}
\DeclareDocumentCommand{\packExistsG}{O{\cT'} O{v} O{\exists X.\cT}}{\textsf{pack}\langle#1, #2\rangle \ \text{as} \ #3}

\DeclareDocumentCommand{\packExistsDy}{O{\cT'} O{v} O{X} O{\cT} }{\textsf{pack}\langle #2\rangle \ \textsf{as} \ #4}

\DeclareDocumentCommand{\packExistsC}{O{\cT'} O{v}  O{\cT}}{\textsf{pack}\langle#1, #2\rangle \ \textsf{as} \ #3}

\DeclareDocumentCommand{\packExistst}{O{\cT'} O{v}  O{X} O{\cT}}{\textsf{pack}\langle#1, #2\rangle \ \textsf{as} \ \exists #3. #4}

\DeclareDocumentCommand{\packExistsr}{O{\cT'} O{v}  O{X} O{\cT}}{\textsf{pack}\langle#1, #2\rangle}

\DeclareDocumentCommand{\evExistst}{O{\ev}  O{\evlift{\cT'}} O{\evlift{\alpha}} }{#1^{#2}_{#3}}

\DeclareDocumentCommand{\evEsubst}{O{\ev} O{X}  O{\hat{\cT'}} O{\hat{\alpha}} }{\mathit{subst}(#1, #2, #3, #4)}

\DeclareDocumentCommand{\precTR}{O{\store; \Delta; \Gamma} m m m m}{#1 |- #2 : #3 \sprec #4: #5}

\newcommand{\gsfredsy}{\Downarrow}

\DeclareDocumentCommand{\gsfred}{m m O{\store} O{\store'}}{\conf[#3]{#1} \gsfredsy \conf[#4]{#2}}
\newcommand{\gsfreds}[2]{#1 \gsfredsy #2}

\DeclareDocumentCommand{\lrwgg}{O{\sstore; \Delta; \Gamma} m m m}{#1 \Vdash #2 : #3 \gprec  #4}
\DeclareDocumentCommand{\StaticgJ}{ O{\sstore; \Delta; \Gamma} m }{#1 |- #2}

\DeclareDocumentCommand{\atomvalue}{m O{\rho}}{\text{Atom}^{=}_{#2}[#1]}
\DeclareDocumentCommand{\atomterm}{m O{\rho}}{\text{Atom}_{#2}[#1]}
\DeclareDocumentCommand{\atom}{m m O{n}}{\text{Atom}_{#3}[#1, #2]}
\DeclareDocumentCommand{\atomv}{m m O{n}}{\text{Atom}^\text{val}_{#3}[#1, #2]}

\renewcommand{\consistent}[1]{#1^{\sharp}}

\newcommand{\equiSymbol}{\equiv}

\DeclareDocumentCommand{\equivType}{O{\rho} m m }{#2 \equiSymbol #3}

\newcommand{\richtype}[2][\alpha]{{#1}^{#2}}
\newcommand{\const}{b}
\newcommand{\ftype}{\mathit{ty}}
\newcommand{\basetype}{B}

\newcommand{\vectorOp}[1]{\overline{#1}}
\newcommand{\aop}{\mathit{op}}
\newcommand{\op}[1]{\aop(#1)}

\newcommand{\pairsy}{\times}
\newcommand{\pairtype}[2]{#1 \pairsy #2}
\newcommand{\pair}[2]{\langle #1, #2 \rangle}

\DeclareDocumentCommand{\fst}{ O{} }{ {\IfNoValueTF{#1}{{\pi_1}}{{\pi_1({#1})}}}}
\DeclareDocumentCommand{\snd}{ O{} }{ {\IfNoValueTF{#1}{{\pi_2}}{{\pi_2({#1})}}}}
\DeclareDocumentCommand{\proj}{ m O{} }{ {\IfNoValueTF{#2}{{\pi_{#1}}}{{\pi_{#1}({#2})}}}}
\DeclareDocumentCommand{\pconf}{ m O{} }{ \conf[{\IfNoValueTF{#2}{{\store}}{{\store[#2]}}}]{#1}}

\definecolor{evcolor}{HTML}{333399}
\newcommand{\evc}[1]{{\color{evcolor} #1}}

\newcommand{\evblack}[1]{\bgroup \renewcommand{\evc}[1]{##1} #1 \egroup}

\DeclareDocumentCommand{\evproj}{ m O{} }{
\evc{{\IfNoValueTF{#2}{{p_{#1}}}{{p_{#1}({#2})}}}}}

\DeclareDocumentCommand{\tproj}{ m O{} }{\proj{#1}[#2]}

\newcommand{\TypeName}{\oblset{TypeName}}

\newcommand{\VarType}{\oblset{TypeVar}}

\newcommand{\EType}{\oblset{EType}}

\DeclareDocumentCommand{\eqrules}{m m O{\sstore} O{\Delta}}{#3; #4 |- #1 = #2}

\DeclareDocumentCommand{\EnvSS}{O{\sstore} O{\Delta} O{\Gamma}}{#1; #2; #3 |- }
\DeclareDocumentCommand{\EnvSG}{O{\store} O{\Delta} O{\Gamma}}{#1; #2; #3 |- }
\DeclareDocumentCommand{\EnvG}{O{\store}}{#1 |-  }

\DeclareDocumentCommand{\swellGamma}{O{\store} O{\Delta}  O{\Gamma}}{#1; #2 |- #3}
\DeclareDocumentCommand{\sswellGamma}{O{\sstore} O{\Delta}  O{\Gamma}}{#1; #2 |- #3}
\DeclareDocumentCommand{\twf}{m O{\sstore} O{\Delta}}{#2; #3 |- #1}
\DeclareDocumentCommand{\gtwf}{m O{\store} O{\Delta}}{#2; #3 |- #1}
\newcommand{\twfsimpl}[2]{#1 |- #2}

\DeclareDocumentCommand{\storeeval}{O{\sstore}}{{#1}  \triangleright}

\DeclareDocumentCommand{\abst}{O{} m O{\store;\Delta}}{\as^{#1}(#2)}
\DeclareDocumentCommand{\conc}{O{} m O{\store;\Delta}}{\cs^{#1}(#2)}

\newcommand{\as}{A}
\newcommand{\cs}{C}

\renewcommand{\cT}{G}

\newcommand{\ceqrules}[2]{\store |- #1 \sim  #2}

\DeclareDocumentCommand{\deqrules}{O{\rls} m m}{#2 = #3}

\renewcommand{\TermT}[1]{#1}
\newcommand{\itu}[1]{#1}
\newcommand{\inTermT}[1]{: \TermT{#1}}
\newcommand{\initu}[1]{: \itu{#1}}
\newcommand{\itermi}{t}
\newcommand{\iterm}[1]{\itermi}
\newcommand{\itermp}[1]{\itermi'}

\DeclareDocumentCommand{\itermg}{m O{} O{}}{\itermi_{#2}^{#3}}

\DeclareDocumentCommand{\ceqrulessimpl}{O{\store; \Delta} m m}{#1 |- #2 \sim  #3}

\renewcommand{\cast}[2]{\evcast{\evc{#1}}{#2}}
\DeclareDocumentCommand{\newev}{m O{\key[1]} O{\lock[2]}}{{\braket{#1}}}
\DeclareDocumentCommand{\newevs}{m O{\key} O{\lock}}{{\braket{#1}}}
\renewcommand{\ev}[1][]{\evc{\varepsilon_{#1}}}

\renewcommand{\trans}[1]{\circ}

\newcommand{\evp}[1][]{\evc{\varepsilon'_{#1}}}
\newcommand{\evpp}[1][]{\evc{\varepsilon''_{#1}}}

\DeclareDocumentCommand{\evsub}{m O{}}{\evc{\ev_{#1}^{#2}}}
\newcommand{\evinst}[2]{\evc{{#1}[{#2}]}}

\newcommand{\evlift}[1]{\evc{\hat{#1}}}
\newcommand{\evliftname}{\textit{lift}}
\newcommand{\evunliftname}{\textit{unlift}}
\newcommand{\evunlift}{\evunliftname_{\store}}
\newcommand{\unlift}[1]{\evunliftname({#1})}
\newcommand{\evliftop}[2]{\evliftname_{#2}({#1})}

\newcommand{\EElemType}{\oblset{EType}}

\DeclareDocumentCommand{\evg}{O{} O{}}{\evc{\varepsilon_{#1}^{#2}}}
\DeclareDocumentCommand{\evSubst}{O{i} O{\tpesub} }{\evc{\varepsilon_{#1}^{#2}}}
\DeclareDocumentCommand{\unrolling}{ m O{\alpha} O{\store}  }{#3^{#1}(#2)}
\DeclareDocumentCommand{\unrollingRel}{ m O{\alpha} O{\store}  }{#3^{#1}(#2)}
\DeclareDocumentCommand{\downtn}{m O{n}}{{#1}\downarrow^{#2}}
\DeclareDocumentCommand{\downtng}{m O{n} O{\store}}{{#1}\downarrow^{#2}}
\newcommand{\rls}{R}

\newcommand{\redopn}{\delta}
\newcommand{\redop}[1]{\delta(\aop, #1)}

\newcommand{\keys}{K}
\newcommand{\locks}{L}
\DeclareDocumentCommand{\key}{ O{} }{ {\IfNoValueTF{#1}{{\keys}}{{\keys_{#1}}}}}
\DeclareDocumentCommand{\keyp}{ O{} }{ {\IfNoValueTF{#1}{{\keys'}}{{\keys'_{#1}}}}}
\DeclareDocumentCommand{\keypp}{ O{} }{ {\IfNoValueTF{#1}{{\keys''}}{{\keys''_{#1}}}}}
\DeclareDocumentCommand{\lock}{ O{} }{ {\IfNoValueTF{#1}{{\locks}}{{\locks_{#1}}}}}
\DeclareDocumentCommand{\lockp}{ O{} }{ {\IfNoValueTF{#1}{{\locks'}}{{\locks'_{#1}}}}}
\DeclareDocumentCommand{\lockpp}{ O{} }{ {\IfNoValueTF{#1}{{\locks''}}{{\locks''_{#1}}}}}
\DeclareDocumentCommand{\rl}{ O{} }{ {\IfNoValueTF{#1}{{\rls}}{{\rls_{#1}}}}}
\DeclareDocumentCommand{\rlp}{ O{} }{ {\IfNoValueTF{#1}{{\rls'}}{{\rls'_{#1}}}}}

\DeclareDocumentCommand{\store}{ O{} }{ {\IfNoValueTF{#1}{{\Xi}}{{\Xi_{#1}}}}}
\DeclareDocumentCommand{\storep}{ O{} }{ {\IfNoValueTF{#1}{{\Xi'}}{{\Xi'_{#1}}}}}
\DeclareDocumentCommand{\sstore}{ O{} }{ {\IfNoValueTF{#1}{{\Sigma}}{{\Sigma_{#1}}}}}

\newcommand{\W}{W}

\newcommand{\K}{\kappa}
\newcommand{\Wstore}{\store}
\newcommand{\getStore}[1][]{{\W\!.\Wstore_{#1}}}

\DeclareDocumentCommand{\getStorew}{O{} O{} }{{\W_{#1}\!.\Wstore_{#2}}}
\DeclareDocumentCommand{\getStorewp}{O{} O{} }{{\W'_{#1}\!.\Wstore_{#2}}}
\DeclareDocumentCommand{\getStorewpp}{O{} O{} }{{\W''_{#1}\!.\Wstore_{#2}}}
\DeclareDocumentCommand{\getStorewppp}{O{} O{} }{{\W'''_{#1}\!.\Wstore_{#2}}}
\DeclareDocumentCommand{\getStorewpppp}{O{} O{} }{{\W''''_{#1}\!.\Wstore_{#2}}}
\DeclareDocumentCommand{\getStorewg}{O{} O{} O{} }{{\W^{#3}_{#1}\!.\Wstore_{#2}}}
\DeclareDocumentCommand{\getStorewgp}{O{} O{} O{} }{{\W'^{#3}_{#1}\!.\Wstore_{#2}}}
\newcommand{\getStorep}[1][]{{\W'.\Wstore_{#1}}}
\newcommand{\getRel}{{\W\!.\K}}
\newcommand{\getRelp}{{\W'\!.\K}}

\newcommand{\tpesub}{\rho}

\newcommand{\tpesubSI}[2][]{\rho_{#1}(#2)}

\DeclareDocumentCommand{\tpesubEnv}{m O{\getStore[i]} O{\rho} }{#3(#1, #2)}

\newcommand{\getIdx}[1][j]{{\W\!.{#1}}}
\newcommand{\getIdxp}[1][j]{{\W'\!.{#1}}}

\DeclareDocumentCommand{\getIdxg}{O{\W} O{} O{} }{{{#1}_{#2}^{#3}}.j}

\DeclareDocumentCommand{\lgru}{m m m m O{\rho} }{\lgrp{#1}{#2}{#3} \in \setu[#5]{#4}}
\DeclareDocumentCommand{\lgrt}{m m m m O{\rho} }{\lgrp{#1}{#2}{#3} \in \sett[#5]{#4}}
\DeclareDocumentCommand{\lgrv}{m m m m O{\rho} }{\lgrp{#1}{#2}{#3} \in \setv[#5]{#4}}
\DeclareDocumentCommand{\lgrg}{m m m m O{\rho} }{\lgrp{#1}{#2}{#3} \in \setg[#5]{#4}}

\DeclareDocumentCommand{\lgrum}{m m m m O{\rho} m}{\setu[#5]{#4} &=\quad \{\lgrp{#1}{#2}{#3} ~|~ #6 \}}
\DeclareDocumentCommand{\lgrvmext}{m m m m O{\rho} m m}{\setv[#5]{#4} &=\quad \{\lgrp{#1}{#2}{#3} ~|~ #6 \} \\ &\cup\quad \{ #7 \} }
\DeclareDocumentCommand{\lgrtm}{m m m m O{\rho} m}{\sett[#5]{#4} &=\quad \{\lgrpt{#1}{#2}{#3}[#4] ~|~ #6 \}}
\DeclareDocumentCommand{\lgrvmalpha}{m m m m O{\rho} m}{\setv[#5]{#4} &=\quad \{(#1, #2, #3) \in \atomvalue{#4}[\emptyset] ~|~ #6 \}}

\DeclareDocumentCommand{\lgrgm}{m m m m O{\rho} m}{\setg[#5]{#4} &=\quad \{\lgrp{#1}{#2}{#3} ~|~ #6 \}}

\DeclareDocumentCommand{\lgrvm}{m m m m O{\rho} m}{\setv[#5]{#4} &=\quad \{\lgrpp{#1}{#2}{#3}[#4] ~|~ #6 \}}

\DeclareDocumentCommand{\lgrvmexists}{m m m O{\rho} m}{\setv[#4]{#3} &=\quad \{\lgrppexists{#1}{#2}[{#3}] ~|~ #5 \}}

\DeclareDocumentCommand{\lgrppexists}{m m O{}}{(#1, #2) \in \atomvalue{#3}}

\newcommand{\translate}{\leadsto}
\newcommand{\backtranslate}{\reflectbox{$\leadsto$}}
\DeclareDocumentCommand{\translationgsfe}{O{t} O{t_{\ev}} O{\cT} O{\store; \Delta; \Gamma}}{#4 |- #1 \translate #2 : #3}

\DeclareDocumentCommand{\backtranslationgsfe}{O{t} O{t_{\ev}} O{\cT} O{\store; \Delta; \Gamma}}{#4 |- #1 \; \backtranslate \; #2 : #3}

\DeclareDocumentCommand{\tt}{O{} O{t}}{#2_{\ev[#1]}}
\DeclareDocumentCommand{\ctxev}{O{} O{\ctx}}{#2_{\ev[#1]}}

\newcommand{\tcsy}{\rightsquigarrow}
\newcommand{\terminationsy}{\Downarrow}
\newcommand{\divergesy}{\Uparrow}
\newcommand{\ctx}{C}
\newcommand{\casy}{\preceq^{\mathit{ctx}}}
\newcommand{\cesy}{\approx^{\mathit{ctx}}}

\DeclareDocumentCommand{\contaprox}{m m m O{\store; \Delta; \Gamma}}{{#4} \vdash #1 \casy #2 : {#3}}

\DeclareDocumentCommand{\contequiv}{m m m O{\store; \Delta; \Gamma}}{{#4} \vdash #1 \cesy #2 : {#3}}

\DeclareDocumentCommand{\contequivempty}{m m m }{#1 \cesy #2 : {#3}}

\DeclareDocumentCommand{\contequivs}{m m m O{\sstore; \Delta; \Gamma}}{{#4} \vdash #1 \cesy #2 : {#3}}

\DeclareDocumentCommand{\contequivg}{m m m O{\store; \Delta; \Gamma}}{{#4} \vdash #1 \cesy #2 : {#3}}

\DeclareDocumentCommand{\contequivgs}{m m m O{\store; \Delta; \Gamma}}{{#4} \vdash #1 \cesy_{\mathit{st}} #2 : {#3}}

\DeclareDocumentCommand{\contlrs}{ O{T} O{T'} O{\ctx_1} O{\ctx_2} O{\sstore; \Delta; \Gamma} O{\sstore';  \Delta' ; \Gamma' }}{\vdash #3 \preceq #4: (#5 \vdash #1) \tcsy (#6 \vdash #2)}

\DeclareDocumentCommand{\typingctxs}{O{T} O{T'} O{\ctx}  O{\sstore; \Delta; \Gamma} O{\sstore'; \Delta'; \Gamma' }}{\vdash #3: (#4 \vdash #1)\tcsy (#5 \vdash #2)}

\DeclareDocumentCommand{\typingctxg}{O{\cT} O{\cT'} O{\ctx}  O{\store; \Delta; \Gamma} O{\store'; \Delta'; \Gamma' }}{\vdash #3: (#4 \vdash #1)\tcsy (#5 \vdash #2)}

\DeclareDocumentCommand{\typingctx}{O{\cT} O{\cT'} O{\ctx} O{\store; \Delta; \Gamma} O{\store'; \emptyenv; \emptyenv}}{\vdash #3: (#4 \vdash #1)\tcsy (#5 \vdash #2)}

\DeclareDocumentCommand{\termination}  {m O{\store'}}          {\conf[#2]{#1\terminationsy}}
\DeclareDocumentCommand{\terminations} {m O{\ctx} O{\sstore'}} {\conf[#3]{#2[#1] \terminationsy}}
\DeclareDocumentCommand{\terminationg} {m O{\ctx} O{\store'}}  {\conf[#3]{#2[#1] \terminationsy}}

\DeclareDocumentCommand{\diverges}{m O{\ctx} O{\sstore'}}{\conf[#3]{#2[#1]} \divergesy}
\DeclareDocumentCommand{\diverge}{m O{\ctx} O{\store'}}{\conf[#3]{#2[#1]} \divergesy}

\renewcommand{\invdom}{\mathit{dom}}
\renewcommand{\invcod}{\mathit{cod}}

\newcommand{\invdoma}[1]{\evc{\invdom(#1)}}
\newcommand{\invcoda}[1]{\evc{\invcod(#1)}}

\DeclareDocumentCommand{\staticJ}{ O{\store; \Delta} m }{#1 |-_{u} #2}
\DeclareDocumentCommand{\staticgJ}{ O{\store; \Delta; \Gamma} m }{#1 |- #2}
\DeclareDocumentCommand{\conf}{ O{\store} m }{#1 \triangleright #2}

\DeclareDocumentCommand{\ceqrules}{m m O{\store; \Delta}}{#3 |- #1 \sim #2}
\DeclareDocumentCommand{\ceqrulesS}{m m O{\sstore; \Delta}}{#3 |- #1 \sim #2}

\newcommand{\gprecinv}{\sqsubseteq}
\DeclareDocumentCommand{\tprulesinv}{m m O{\store; \Delta}}{#1\ \gprecinv\  #2}
\DeclareDocumentCommand{\tprules}{m m O{\store; \Delta}}{#1\ \gprec\  #2}

\renewcommand{\cT}{G}

\DeclareDocumentCommand{\lgrp}{m m m }{(#1, #2, #3)}
\DeclareDocumentCommand{\lgrpp}{m m m O{}}{(#1, #2, #3) \in \atomvalue{#4}}
\DeclareDocumentCommand{\lgrpt}{m m m O{}}{(#1, #2, #3) \in \atomterm{#4}}

\newcommand{\setu}[2][\rho]{\mathcal{U}_{#1}\llbracket#2\rrbracket}
\newcommand{\sett}[2][\rho]{\mathcal{T}_{#1}\llbracket#2\rrbracket}
\newcommand{\setv}[2][\rho]{\mathcal{V}_{#1}\llbracket#2\rrbracket}
\newcommand{\setg}[2][\rho]{\mathcal{G}_{#1}\llbracket#2\rrbracket}
\newcommand{\sets}[1][\Wstore]{\mathcal{S}\llbracket#1\rrbracket}
\newcommand{\setd}[1][\Delta]{\mathcal{D}\llbracket#1\rrbracket}

\newcommand{\futureW}{\succeq}
\newcommand{\world}{\oblset{World}}
\newcommand{\worldn}{\oblset{World}_n}

\newcommand{\reln}[3][n]{\oblset{Rel}_{#1}[#2, #3]}
\newcommand{\nat}{\oblset{Nat}}
\newcommand{\tnstore}{\oblset{Store}}
\newcommand{\relj}[1][j]{\oblset{Rel}_#1}

\newcommand{\logaproxg}[4][\store; \Delta; \Gamma]{#1 |- #2 \preceq #3 : #4}

\newcommand{\logeqg}[4][\store; \Delta; \Gamma]{#1 |- #2 \approx #3 : #4}

\newcommand{\delequal}{\triangleq}
\newcommand{\krel}{\kappa}
\newcommand{\acotrel}[2]{\lfloor #1 \rfloor_{#2}}

\newcommand{\extworld}{\boxtimes}

\newcommand{\lcorchete}{[}
\newcommand{\rcorchete}{]}

\newcommand{\evout}{\ev[\mathit{out}]}
\newcommand{\cEU}{\cE_{*}}

\newcommand{\te}{t_{\varepsilon}}

\newcommand{\transc}{\trans{}}
\newcommand{\barr}{\!->\!}

\renewcommand{\emptyenv}{\cdot}

\newcommand{\downstep}[1][]{\downarrow_{#1}}

\newcommand{\interiorS}[1][\rl]{\interior{}}

\newcommand{\interiorTN}[1][\rl]{\interior{\mathit{TN}}}

\newcommand{\FTV}{FTV}

\definecolor{failcolor}{HTML}{A9341F}

\DeclareDocumentCommand{\getev}{m}{\mathit{ev}(#1)}

\newcommand{\lamB}{$\lambda B$\xspace}
\newcommand{\sysFc}{System~F$_{C}$\xspace}
\newcommand{\sysFg}{System~F$_{G}$\xspace}
\newcommand{\sysF}{System~F\xspace}
\newcommand{\sysFw}{System~F$_{\omega}$\xspace}
\newcommand{\csa}{CSA\xspace}
\newcommand{\gsf}{GSF\xspace}

\newcommand{\gsfe}{\mbox{\gsf\negthickspace$\varepsilon$}\xspace}
\newcommand{\gsfeex}{\mbox{GSF$_{\varepsilon}^{\exists}$}\xspace}

\newcommand{\glangev}{\gsfe}
\newcommand{\gsfex}{GSF$^{\exists}$\xspace}
\newcommand{\SPFLex}{SF$^{\exists}$\xspace}
\newcommand{\polyG}{PolyG$^\nu$\xspace}
\newcommand{\polyC}{PolyC$^\nu$\xspace}
\newcommand{\CBPV}{CBPV$_{\text{OSum}}$\xspace}

\newcommand{\cE}{E}

\newcommand{\enrich}[2]{\evliftop{#1}{#2}}

\newcommand{\substTermPaper}[3]{#3[#2/#1]}

\oblang{if[if\;],then[\;then\;],else[\;else\;],let[let\;],in[\;in\;],fi[if]}

\DeclareDocumentCommand{\pitwoeveq}{m m O{\tpesub}}{#3 |- #1 \equiSymbol #2}

\DeclareDocumentCommand{\sync}{m O{\W}}{\mathit{sync}(#1, #2)}

\newcommand{\B}{B}

\newcommand{\SPFL}{SF\xspace}
\newcommand{\scheme}{\mathit{schm}_{\footnotesize u}}
\newcommand{\schemeEx}{\mathit{schm}_{\footnotesize e}}
\newcommand{\cschemeEx}{\consistent{\schemeEx}}
\newcommand{\cscheme}{\consistent{\scheme}}
\newcommand{\insta}{\mathit{inst}}
\newcommand{\cinsta}{\consistent{\insta}}

\newcommand{\projts}{\mathit{proj}}
\newcommand{\projt}[2]{\projts_{#1}(#2)}
\newcommand{\cprojts}{\consistent{\projts}}
\newcommand{\cprojt}[2]{\cprojts_{#1}(#2)}

\newcommand{\unknowifys}{\mathit{const}}
\newcommand{\unknowify}[1]{\unknowifys(#1)}

\newcommand{\undefinedow}{\text{ undefined o/w }}

\newcommand{\unkL}{\?}
\newcommand{\unkR}{\?}
\newcommand{\unkG}[1]{\?}

\newcommand{\isjudgment}[3]{#1 \Vdash #2 \sim #3}
\newcommand{\ijudgment}[3]{#1 \Vdash \ceqrulessimpl{#2}{#3}}
\DeclareDocumentCommand{\ijudgment}{m O{\store; \Delta} m m}{#1 \Vdash \ceqrulessimpl[#2]{#3}{#4}}

\newcommand{\tu}{s}
\DeclareDocumentCommand{\insertAsc}{m m m O{\store}}{\mathit{asc?}(#1,#2,#3,#4)}
\newcommand{\mland}{\land {}}

\DeclareDocumentCommand{\staticgJx}{ O{\store; \emptyenv; \emptyenv} m }{#1 |- #2}
\DeclareDocumentCommand{\ceqrulessimplx}{O{\store; \emptyenv} m m}{\ceqrulessimpl[#1]{#2}{#3}}
\DeclareDocumentCommand{\ceqrulessimpls}{O{\sstore; \Delta} m m}{\ceqrulessimpl[#1]{#2}{#3}}
\DeclareDocumentCommand{\projx}{ m O{} }{\cprojt{#1}{#2}}

\DeclareDocumentCommand{\valuestore}{ m O{} }{\IfNoValueTF{#2}{v}{\conf[#2]{#1}}}

\DeclareDocumentCommand{\setn}{ m m O{\sstore} O{\tpesub} }{\mathcal{N}^{#3}_{#4}\llbracket#1 \gprec #2\rrbracket}
\DeclareDocumentCommand{\setc}{ m m O{\sstore} O{\tpesub} }{\mathcal{C}^{#3}_{#4}\llbracket#1 \gprec #2\rrbracket}

\DeclareDocumentCommand{\inStatic}{ m m O{\sstore} O{\tpesub} }{\text{ImpSV}^{#3}_{#4} [#1 \gprec #2]}
\DeclareDocumentCommand{\staticG}{ m m m O{\sstore} O{\tpesub} }{\uvalStatic(#1) \land \proj{2}[\ev] = \enrich{#5(#2)}{#4} \land \staticgJ[\sstore; \emptyenv; \emptyenv]{v} : #5(#3)}

\DeclareDocumentCommand{\setgg}{ m O{\sstore} O{\tpesub} }{\mathcal{G}^{#2}_{#3}\llbracket#1\rrbracket}
\DeclareDocumentCommand{\setdd}{ m O{\sstore}  }{\mathcal{D}^{#2}\llbracket#1\rrbracket}

\newcommand{\uvalStatic}{\mathit{static}}

\DeclareDocumentCommand{\norm}{m m m O{\store}}{\mathit{norm}(#1, #2, #3, #4)}

\newcommand{\impPolyName}{\mathit{dip}}
\newcommand{\impPoly}[3]{\impPolyName(#1,#2)}

\newcommand{\nsprec}{\not\leqslant}
\newcommand{\sprec}{\leqslant}
\newcommand{\functype}[2]{#1 -> #2}
\newcommand{\sdgg}{DGG$^\sprec$\xspace}
\newcommand{\gdgg}{DGG\xspace}

\newcommand\idx{\mathit{id_X}}
\newcommand\idu{\mathit{id_\?}}
\newcommand\idxu{\mathit{id_{X\?}}}